\documentclass[11pt,a4paper]{article}

\usepackage{authblk}
\usepackage[left=1in, right=1in, top=1in, bottom=1in]{geometry}
\usepackage{amsthm}
\usepackage{amsfonts}
\usepackage{mathrsfs}
\usepackage[pdftex]{hyperref}
\usepackage{amssymb}
\usepackage{amsmath}
\usepackage{cleveref}
\usepackage{xcolor}
\usepackage{graphicx}
\usepackage{xparse}
\usepackage{subcaption}
\usepackage{etoolbox}
\PassOptionsToPackage{dvipdfmx}{graphicx}

\usepackage[toc,page]{appendix}

\appto\appendix{\addtocontents{toc}{\protect\setcounter{tocdepth}{1}}}

\appto\listoffigures{\addtocontents{lof}{\protect\setcounter{tocdepth}{1}}}
\appto\listoftables{\addtocontents{lot}{\protect\setcounter{tocdepth}{1}}}

\numberwithin{equation}{section}

\usepackage{macros}

\title{Many-body Fu-Kane-Mele index}

\author[1]{Sven Bachmann \thanks{email: \href{sbach@math.ubc.ca}{sbach@math.ubc.ca}}}
\author[2]{Alex Bols \thanks{email: \href{abols01@phys.ethz.ch}{abols01@phys.ethz.ch}}}
\author[3]{Mahsa Rahnama \thanks{email: \href{mahsa.r@math.ubc.ca}{mahsa.r@math.ubc.ca}}}
\affil[2]{Institute for Theoretical Physics, ETH Z{\"u}rich}
\affil[1, 3]{Department of Mathematics, University of British Columbia, Vancouver}

\begin{document}

\date{\today}
\maketitle

\begin{center}
Dedicated to the memory of Professor Huzihiro Araki
\end{center}

\begin{abstract}
We define a $\Z_2$-valued index for stably short-range entangled states of two-dimensional fermionic lattice systems with charge conservation and time reversal symmetry. The index takes its non-trivial value precisely if the `fluxon', the state obtained by inserting a $\pi$-flux through the system, transforms under time reversal as part of a Kramers pair. This index extends the Fu-Kane-Mele index of free fermionic topological insulators to interacting systems.
\end{abstract}

\tableofcontents

\section{Introduction}

The discovery of the integer quantum Hall effect \cite{klitzing1980new} and its understanding in terms of the topology of band insulators \cite{TKNN,avron1983homotopy} set the stage for the study of topological phases of matter. The discovery of the quantum spin Hall effect for time reversal symmetric insulators \cite{kane2005quantum, kane2005z} quickly led to the recognition that the landscape of topological phases of free fermion insulators and superconductors is extremely rich, eventually leading to a complete classification in terms of $K$-theory, summarised in the periodic table of topological invariants~\cite{kitaev2009periodic, ryu2010topological}.

A natural question is whether the phases identified and classified by the periodic table of topological invariants survive if interactions are allowed. The answer is no in general \cite{fidkowski2010effects, fidkowski2011topological}, but yes at least for those phases for which the topological invariant has a clear physical interpretation. This is the case for the integer quantum Hall effect where the invariant is a linear response coefficient. The fact that the Hall conductance is a well-defined quantised invariant of topological phases of interacting matter was explained early~\cite{laughlin1981quantized} and it is by now very well understood, in particular in the framework of quantum lattice systems \cite{hastings2015quantization, bachmann2018quantization, bachmann2020many, oshikawa2020filling, bachmann2020note, kapustin2020hall}. Of course, the stability of the non-interacting phases comes in general with the appearance of new, purely interacting ones such as the fractional Hall phases that can also be classified~\cite{frohlich1997classification}.

The invariant that characterises the quantum spin Hall effect is the $\Z_2$-valued Fu-Kane-Mele index \cite{kane2005quantum, kane2005z, fu2006time, fu2007topological}. It was argued in \cite{ran2008spin} \cite{qi2008spin} that this invariant has a physical interpretation in terms of time reversal transformation properties of defects bound to $\pm\pi$ magnetic fluxes. If the invariant is non-trivial then these defects form a Kramers pair for the time reversal symmetry. We review the argument below. While early works placed some importance on the role of spin, they all recognized that the invariant is in fact only relying on fermionic time reversal symmetry: no component of spin must be conserved for the invariant to be meaningful.

Following these ideas we construct a $\Z_2$-valued index for short-range entangled states of two-dimensional interacting lattice fermion systems that are invariant under a fermionic time reversal symmetry that commutes with a $U(1)$-symmetry. We show that this index agrees with the Fu-Kane-Mele index of non-interacting systems, namely for symmetric quasi-free states, and is stable under locally generated automorphisms that respect the symmetries. In particular, we show that the quantum spin Hall phase indeed persists in the interacting regime.

\subsection{Idea behind the construction}

Let us now summarise the Kramers pair binding idea of \cite{ran2008spin, qi2008spin} and describe how our many-body index is constructed based on this idea.

The systems we consider are interacting generalisations of gapped free fermion Hamiltonians on the two-dimensional square lattice with a time reversal symmetry $\scrT$ that squares to $-\I$ (Altland-Zirnbauer class AII \cite{AltlandZirnbauer}). Such Hamiltonians come in two topological phases distinguished from each other by the $\Z_2$-valued Fu-Kane-Mele index. A nice characterisation of the index is obtained by piercing magnetic flux through the system \cite{DeNittisSchulzBaldes2016spectral} as follows. Let $\scrH$ be a gapped free fermion Hamiltonian such that $\scrT \scrH \scrT^* = \scrH$, and let $\scrH(\phi)$ be the free fermion Hamiltonian obtained from $\scrH$ by inserting flux $\phi$. Then the spectrum of $\scrH(\phi)$ inherits the bulk spectrum of $\scrH$, but may in addition have eigenvalues in the bulk gap. As a function of $\phi$, these eigenvalues form bands within the gap of $\scrH$ corresponding to modes bound to the magnetic flux, see Figure \ref{fig:flux insertion spectrum}. The time reversal symmetry forces the spectra of $\scrH(\phi)$ and $\scrH(-\phi)$ to be equal, and for the time reversal invariant fluxes $\phi_* = 0, \pm \pi$ the time reversal symmetry forces all eigenvalues of $\scrH(\phi_*)$ to be doubly degenerate due to Kramers pairing. Energy bands satisfying these constraints come in two classes distinguished from each other as follows. Let $E_F$ be a generic Fermi energy in the bulk gap of the Hamiltonian $\scrH$ (for simplicity we assume that $E_F$ is not an eigenvalue of $\scrH(\phi_*)$) and let $N \in \Z_2$ be the parity of the number of times a band crosses the Fermi energy in the interval $\phi \in [0, \pi)$. Then $N$ is invariant under any homotopy of the bands that respects the time reversal symmetry constraints. The Fu-Kane-Mele index of $\scrH$ (w.r.t.\ the chosen gap) is equal to $N$. The spectrum of Figure \ref{fig:flux insertion spectrum} corresponds to an $\scrH$ with a non-trivial Fu-Kane-Mele index. This characterisation of the Fu-Kane-Mele index parallels the characterisation for translation invariant free fermion systems through the bulk-boundary correspondence \cite{graf2013bulk, fonseca2020two, bols2023fredholm} where the bands associated to modes bound to the magnetic flux are replaced by bands formed by edge modes.

\begin{figure}[ht]
\centering \includegraphics[height=4cm]{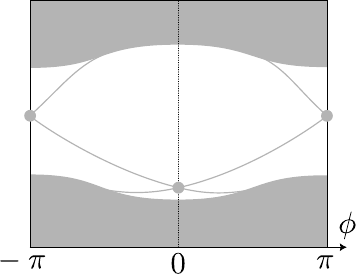}
\caption{The spectrum of a gapped free fermion Hamiltonian $\scrH(\phi)$ as a function of pierced flux. The filled grey regions are bulk spectrum. The bands represent spectrum associated to modes bound to the magnetic flux. The bands cross at time reversal invariant values of the flux $\phi_* = 0, \pm \pi$ forming Kramers pairs represented by dots. The spectrum depicted here corresponds to a non-trivial insulator because for any Fermi energy $E_F$ in the bulk gap, the eigenvalue bands cross $E_F$ an odd number of times (in fact, $\pm 1$) in the interval $\phi \in [0, \pi)$.}
\label{fig:flux insertion spectrum}
\end{figure}

We now imagine preparing the system at $\phi = 0$ in its ground state where all single particle states with energy below $E_F$ in the bulk gap are occupied, and increasing the flux $\phi$ adiabatically from $0$ to $\pi$. The  state is adiabatically transported along this path and at each value of $\phi \in [0, \pi]$ the occupied single particle states are indicated in the left panel of Figure~\ref{fig:flux insertion occupations}. We see that one mode of the Kramers pair at $\phi = \pi$ is occupied while the other mode of the pair remains unoccupied. If the flux is instead decreased from $\phi = 0$ to $\phi = -\pi$ we get a similar picture, except that now the other mode of the Kramers pair becomes occupied, see the right panel of Figure~\ref{fig:flux insertion occupations}. The two states obtained from the ground state by either increasing the flux from $0$ to $\pi$ or decreasing the flux from $0$ to $-\pi$ are mapped into each other by time reversal.

\begin{figure}[ht]
\centering \includegraphics[height=4cm]{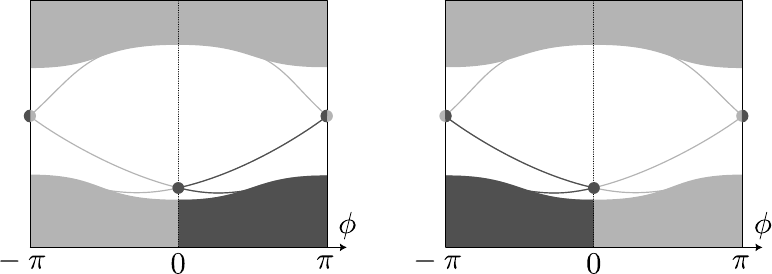}
\caption{The occupied single particle spectrum (dark) as the ground state at $\phi = 0$ is adiabatically transported from $\phi = 0$ to $\phi = \pi$ (left panel) or to $\phi = -\pi$ (right panel). The states at $\pm \pi$ differ only in the occupation of the Kramers pair of modes bound to the $\pm \pi$ flux. The mode of this Kramers pair that is occupied in the state obtained  by adiabatic transport from $0$ to $\pi$ flux is unoccupied in the state obtained by adiabatic transport from $0$ to $-\pi$ flux and vice versa.}
\label{fig:flux insertion occupations}
\end{figure}

Since these two states differ from each other in the occupation of a single Kramers pair it seems reasonable to suspect that the corresponding many-body states form a Kramers pair. The results of this paper provide a rigorous justification of this intuition by showing that the two states obtained by $\pm \pi$ flux insertion are unitary equivalent, and are therefore given by unit vectors in the same irreducible representation of the observable algebra. We prove that the time reversal symmetry is implemented in this representation by an antiunitary operator $T$, and that the vector representatives of the two states are indeed a Kramers pair for $T$.

To define a many-body $\Z_2$ index for a time reversal and $U(1)$ symmetric state $\omega$ we therefore want to construct states $\omega^{\pm}$ corresponding to adiabatically inserting $\pi$ or $-\pi$ magnetic flux in the system. Under the assumption that $\omega$ is a symmetric short-range entangled state one can construct a Hamiltonian $H = \sum_{x} H_x$ with symmetric terms $H_x$ localized near the site $x$ and such that $\omega$ is the unique gapped ground state of $H$. Let $H(\phi) = \sum_x H_x(\phi)$ with $H_x(\phi)$ obtained from $H_x$ by conjugation with the $U(1)$ symmetry restricted to the upper half plane. Since the $H_x$ are symmetric the $H_x(\phi)$ substantially differ from $H_x$ only for those $x$ near the horizontal axis. The Hamiltonians $H(\phi)$ have unique gapped ground states $\omega_{\phi}$ related to $\omega$ by conjugation by the $U(1)$ symmetry restricted to the upper half plane. The states $\omega_{\phi}$ differ from $\omega$ only near the horizontal axis. Moreover, $\omega_\pi = \omega_{-\pi}$ since the symmetry is $U(1)$. We then consider the quasi-adiabatic flow~\cite{hastingswen} for the family $H(\phi)$ (Section \ref{sec:gauge transformation}). Since $H(\phi)$ only changes as a function of $\phi$ near the horizontal axis, the quasi-adiabatic flow acts non-trivially only near this axis. An important technical merit of the quasi-adiabatic flow is that it is generated by a local interaction localised near the horizontal axis and so it can be restricted to act non-trivially only near the left horizontal axis (Section \ref{sec:defect states}). We define the states $\omega^{\pm}$ by acting with this restricted quasi-adiabatic flow for flux $\pm \pi$ on the ground state. Because we work in the half-line gauge, both states $\omega^+$ and $\omega^-$ look like $\omega_{\pi} = \omega_{-\pi}$ far to the left of the origin near the left horizontal axis, and look like $\omega$ far from the origin and away from the left axis, see Figure~\ref{fig:Half-line gauge}.
\begin{figure}[h]
\centering \includegraphics[width=0.5\textwidth]{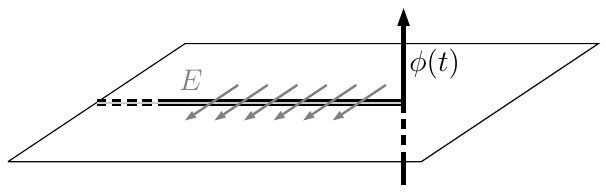}
\caption{The geometry of flux insertion. In the half-line gauge, the electric field $E$ generated by the time-dependent flux $\phi(t)$ exists only along the left half-line, affecting the state only in the vicinity of that line. The states $\omega^\pm$ are the result of the process $\phi=0\to\phi=\pm\pi$.}
\label{fig:Half-line gauge}
\end{figure}
Under the assumption that $\omega$ is short-range entangled it follows that the states $\omega^{\pm}$ are unitarily equivalent (Proposition \ref{prop:almost local unitary equivalence}). The states $\omega^{\pm}$ therefore have unit vector representatives $\Omega^{\pm} \in \caH$ in the GNS representation $(\pi, \caH)$ of either of them. Since they are time reversal images of each other, time reversal is implemented in the representation by an antiunitary operator $T$ (Lemma \ref{lem:TR implementation}). It turns out that either $T^2 \Omega^+ = \Omega^+$ or $T^2 \Omega^+ = - \Omega^+$ (Lemma~\ref{lem:Kramers or not}). The first possibility corresponds to a trivial value for the many-body index, while the second possibility corresponds to a non-trivial value of the index. In the non-trivial case we have that $\Omega^- = T \Omega^+$ and $\Omega^+$ form an orthogonal Kramers pair for~$T$.

\subsection{Organisation of the paper}

The paper is organised as follows. In Section \ref{sec:setup and results} we give precise definitions of the fermionic lattice systems and the class of pure states on these systems that we will study. We then state our main theorem, namely (i) that there is a $\Z_2$-valued index associated to these states, (ii) that this index is an invariant of symmetry protected topological phases (it is stable under locally generated automorphisms which respect the symmetries), and (iii) that this index extends the Fu-Kane-Mele index for the non-interacting quantum spin Hall effect. Specifically, we prove equality with the $mod\:2$ spectral flow of~\cite{DeNittisSchulzBaldes2016spectral}, which is known to be equivalent to the other expressions of the index whenever they can be compared. The rest of the paper is devoted to proving the main theorem. In Section \ref{sec:definition of the index} we define the index of a symmetric short-range entangled (SRE) state by carrying out the construction sketched above. In Section \ref{sec:properties of the index} we first prove that the index is well-defined in the sense that it is independent of various choices made in the definition. We then turn to its properties, namely multiplicativity under stacking and stability under locally generated automorphisms that respect the symmetries. We also show that symmetric product states have trivial index which allows us to extend the index to stably SRE states. Section \ref{sec:free fermion examples} is devoted to the non-interacting situation. We first show that all translation-invariant, time reversal invariant quasi-free states are stably SRE and therefore have a well defined index. We then prove that our index agrees with the Fu-Kane-Mele index for any symmetric stably SRE quasi-free state (in particular for  translation invariant non-interacting spin Hall states). The next Section \ref{sec:proof of the main theorem} collects these results to prove the main theorem. The appendices are devoted to some technical results. In Appendix \ref{app:almost local transitivity} we prove a transitivity result for fermionic SRE states that is used in Section \ref{sec:definition of the index} to show that the $\pm\pi$ flux states are unitarily equivalent. In Appendix \ref{app:flux insertion for free fermions} we describe the free fermion analogue of the adiabatic flux insertion used in Section \ref{sec:free fermion examples}.

\section{Setup and results} \label{sec:setup and results}

\subsection{Observable algebras}

Let $\caA$ be the CAR algebra over $l^2(\Z^2 ; \C^n)$, i.e. $\caA$ is the $C^*$-algebra generated by an identity element and annihilation operators $\{ a_{x, i} \}_{x \in \Z^2, i = 1, \cdots, n}$ that satisfy the canonical anticommutation relations
\begin{equation}
    \lbrace a_{x, i}, a_{y, j} \rbrace = \lbrace a^*_{x, i}, a^*_{y, j} \rbrace = 0, \quad \lbrace a_{x, i}, a^*_{y, j} \rbrace = \delta_{x, y} \delta_{i, j}
\end{equation}
for all \emph{sites} $x, y \in \Z^2$ and all $i, j = 1, \cdots, n$. For any $\Gamma \subset \Z^2$ we denote by $\caA_{\Gamma}$ the unital $C^*$-subalgebra of $\caA$ generated by the $\{ a_{x, i} \}_{x \in \Gamma, i=  1, \cdots, n}$. We also write $\caA_x = \caA_{\{x\}}$ for any site $x \in \Z^2$. The \emph{fermion parity} is the *-automorphism $\theta$ of $\caA$ uniquely determined by $\theta(a_{x, i}) = -a_{x, i}$ for all $x \in \Z^2$ and all $i = 1,\ldots, n$. We have $\theta( \caA_{\Gamma} ) = \caA_{\Gamma}$ for all $\Gamma \subset \Z^2$. An operator $O \in \caA$ is called \emph{even} if $\theta(O) = O$, \emph{odd} if $\theta(O) = -O$ and \emph{homogeneous} if it is either even or odd. For any $\Gamma \subset \Z^2$ the subset of even operators of $\caA_{\Gamma}$ is a $C^*$-subalgebra of $\caA_{\Gamma}$ which we denote by $\caA_{\Gamma}^+$. We will write $\Gamma \Subset \Z^2$ to mean that $\Gamma$ is a finite subset of $\Z^2$. If $A \in \caA_{\Gamma}$ for some $\Gamma\ \Subset \Z^2$ then $A$ is said to be a local operator. The union of all local operators is a *-subalgebra of $\caA$ which we denote by $\caA^{\loc}$. It is dense in the topology of the C*-norm, see~\cite{BratteliRobinson2}.

If $\caA$ is the CAR algebra over $l^2(\Z^2;\C^n)$ and $\caA'$ is the CAR algebra over $l^2(\Z^2;\C^{n'})$ then we denote by $\caA \gotimes \caA'$ the CAR algebra over $l^2(\Z^2;\C^{n + n'})$ which contains $\caA$ and $\caA'$ as unital subalgebras. This algebra describes the system obtained by stacking the systems described by $\caA$ and $\caA'$ on top of each other.

\subsection{Locally generated automorphisms}\label{sec:TDI & LGA}


For any $x \in \Z^2$ and any $r \in \N$ we let $B_x(r) := \{ y \in \Z^2 \, : \, \dist(x, y) \leq r \}$ be the ball of radius $r$ centered at $x$. Let $\caF$ be the collection of non-increasing, strictly positive functions $f:\R^+ \rightarrow \R^+$ such that $\lim_{r \uparrow \infty} r^p f(r) = 0$ for all $p > 0$. An operator $A \in \caA$ is said to be $f$-localised near a site $x \in \Z^2$ if there is a sequence of operators $A_r \in \caA_{\caB_x(r)}$ such that $\norm{A - A_{r}} \leq f(r)\Vert A\Vert$ for all $r \in \N$.  Any operator that is $f$-localised near some site for some $f \in \caF$ is called an almost local operator. The almost local operators form a *-subalgebra of $\caA$ which we denote by $\caA^{\aloc}$. Since $\caA^{\loc} \subset \caA^{\aloc}$, the algebra of almost local operators is norm-dense in $\caA$.

A \emph{0-chain} $F$ is a map $\bbZ^2\ni x\to F_x\in \caA^{\aloc}$ such that $F_x$ is even and self-adjoint, $\sup_x \norm{F_x}<\infty$ and such that each $F_x$ is $f$-localised near $x$ for the same $f \in \caF$. We say that $F$ is an $f$-local 0-chain. For any $A \in \caA^{\aloc}$ the sum $[F, A] := \sum_{x \in \Z^2} [F_x, A]$ converges in norm to an element of $\caA^{\aloc}$. Hence, to every 0-chain corresponds a derivation $\delta^F$ of the dense subalgebra $\caA^{\aloc}$ which generates a one-parameter family of automorphisms $\al^{F}_s$ with $s \in \R$ called the \emph{time evolution generated by $F$} and defined by $\al_0^F = \id$ and $-\iu \, (\dd \al^F_s(A) / \dd s) = \al^F_s\left([F,A]\right)$ for all $A \in \caA^{\aloc}$ and $s \in \R$.

Let us generalise this to a time-dependent setting. A family of 0-chains $F(s)$ for $s \in \bbR$ is called a time-dependent interaction (TDI) if $\sup_{s,x}\norm{F_x(s)} <\infty$ and for any $s_0\in\bbR$, there is an $f \in \caF$ such that $F(s)$ is an $f$-local 0-chain for all $|s| \leq s_0$. We assume moreover that the map $s\mapsto F_x(s)$ is norm continuous for all $x\in\bbZ^2$. As above, a TDI $F$ generates a strongly continuous family of automorphisms $\al^F_s$ of $\caA$ for $s \in \bbR$ which is given by
\begin{equation}\label{LGA DE}
    \al^F_s(A) = A + \iu \int_0^s  \dd u \, \al^F_u \big( [F(u), A] \big)
\end{equation}
for all $A \in \caA^{\aloc}$. Equivalently, $\al^F_s(A)$ is the unique solution of the initial value problem $-\iu \frac{\dd}{\dd s} \al^F_s(A) = \al^F_s ( [F(s), A] )$ with $\alpha_0^F(A) = A$. An important consequence of the Lieb-Robinson bound \cite{LR} is that
\begin{equation}\label{LR bound}
    \al^F_s\left(\caA^{\aloc}\right)\subset \caA^{\aloc},
\end{equation}
see e.g.~\cite{ScatteringQSS,LRForfermions}. Since $\al^F_s$ is bounded, it extends to an automorphism of the whole algebra $\caA$, which we denote by the same symbol.

\begin{definition} \label{def:LGA}
An automorphism $\al$ of $\caA$ is a \emph{locally generated automorphism} (LGA) if it is of the form $\al = \al^F_s$ for a finite $s \geq 0$ and a TDI $F$.
\end{definition}

For any $\Gamma\subset\bbZ^2$ the map $F^\Gamma$ defined by $F^\Gamma_x = \chi_{\Gamma}(x) F_{x} $ is again a TDI.  Since TDIs are formal sums of even operators we have that $\al^F_s \circ \theta = \theta \circ \al_s^F$ for any TDI $F$ and any $s \in \R$.

The set of LGAs is a group since $(\alpha_t^{F})^{-1}$ and $\alpha_t^F\circ\alpha_t^G$ are LGAs, generated by the TDIs $-\alpha_t^F(F(t))$ and $G(t) + (\alpha_t^G)^{-1}(F(t))$ respectively. The group is non-abelian with 
\begin{equation}\label{eq:LGA commutation}
    \alpha_t^F\circ\alpha_t^G = \alpha_t^{\alpha^F(G)}\circ\alpha_t^F,
\end{equation}
see~\cite{bachmann2022classification,kapustinNoether}.

\subsection{Symmetries} \label{sec:symmetries}

The \emph{charge} at site $x \in \Z^2$ is the self-adjoint operator $Q_x := \sum_{i = 1}^n a_{x, i}^* a_{x, i}$. Clearly, each $Q_x$ is an even self-adjoint operator of norm $n$ supported on $\{x\}$ so $Q$ is a 0-chain. We denote by $\rho_\phi = \alpha_\phi^Q$ the corresponding LGAs. Similarly, we let $\rho^\Gamma_\phi = \alpha_\phi^{Q^\Gamma}$ for any $\Gamma\subset\bbZ^2$ and note that $\rho_{\phi}^{\Gamma}(\caA_{\Gamma'}) = \caA_{\Gamma'}$ for any $\Gamma,\Gamma'\subset\bbZ^2$. Moreover, since $Q_x$ has integer spectrum,
\begin{equation}
    \rho^\Gamma_{2\pi} = \mathrm{id}
\end{equation}
and hence $\rho^{\Gamma}_{\phi + 2\pi} = \rho^{\Gamma}_{\phi}$ for any $\phi \in \R$. We shall refer to these automorphisms as \emph{$U(1)$ transformations}. We note that the commutation relation $[Q_x,a_{x,i}] = -a_{x,i}$ implies that
\begin{equation}
    \rho_{\pi} = \theta.
\end{equation}

We also equip $\caA$ with an antilinear automorphism $\tau$ such that
\begin{equation}
   \tau^2 = \theta, 
\end{equation}
which we call \emph{(fermionic) time reversal}. We further assume that its action is local in the sense that $\tau( \caA_{\Gamma} ) = \caA_{\Gamma}$ for all $\Gamma \subset \Z^2$. Finally, we assume that $\tau(Q_x) = Q_x$ for all $x \in \Z^2$. It follows from this and the antilinearity of $\tau$ that
\begin{equation}\label{TR of charge}
    \tau \circ \rho_{\phi} = \rho_{-\phi} \circ \tau
\end{equation}
for all $\phi\in\bbR$.

We call $(\caA, \tau)$ a \emph{fermion system with time reversal}. If $(\caA, \tau)$ and $(\caA', \tau')$ are fermion systems with time reversal, then $(\caA \gotimes \caA', \tau \gotimes \tau')$ is again a fermion system with time reversal.

We say that a TDI $F$ is time reversal symmetry preserving if $\tau(F_x(s)) = -F_x(s)$ for all $s \in \R$ and all $x \in \Z^2$. Similarly, we say the TDI $F$ is $U(1)$ symmetry preserving if $\rho_{\phi}( F_x(s) ) = F_x(s)$ for all $\phi \in \R$, all $s \in \bbR$ and $x\in\bbZ^2$. A TDI that is simultaneously time reversal and $U(1)$ symmetry preserving is called a \emph{symmetry preserving TDI}. If $F$ is a symmetry preserving TDI then the associated family of automorphisms $\al_s^F$ satisfy 
\begin{equation}
   \rho_{\phi} \circ \al_s^F \circ \rho_{\phi}^{-1} = \al_s^F = \tau \circ \al_s^F \circ \tau^{-1} 
\end{equation}
and $\al_s^F$ are called \emph{symmetry preserving LGAs}.

\subsection{Symmetric short-range entangled states}\label{sec:states}

Let $(\caA, \tau)$ be a fermion system with time reversal. A state $\psi : \caA \rightarrow \C$ is called \emph{homogeneous} if 
\begin{equation}
    \psi \circ \theta = \psi,
\end{equation}
 time reversal invariant if
 \begin{equation}
     \psi \circ \tau= \bar\psi,
 \end{equation}
 where $\bar\psi$ is the antilinear map $\caA\ni A\mapsto\bar\psi(A)=\psi(A^*)\in\bbC$, and $U(1)$-invariant if 
 \begin{equation}
     \psi \circ \rho_{\phi} = \psi
 \end{equation}
 for all $\phi \in \R$. If a state is simultaneously homogeneous, time reversal invariant, and $U(1)$-invariant, then we call the state \emph{symmetric}. A pure state $\psi_0$ is said to be a pure product state if the restriction of $\psi_0$ to each $\caA_x$ is pure.

\begin{definition} \label{def:SRE et al}
(i) A state $\psi : \caA \rightarrow \C$ is \emph{short-range entangled} (SRE) if there is a homogeneous pure product state $\psi_0$ and an LGA $\al$ such that $\psi = \psi_0 \circ \al$. \\
(ii) A state $\psi : \caA \rightarrow \C$ is \emph{stably SRE} if there is a homogeneous product state $\psi_0$ on an auxiliary fermion system $\caA'$ such that $\psi \gotimes \psi_0$ is SRE.
\end{definition}

\noindent Note that a (stably) SRE state is necessarily pure and homogeneous.

If $\caA$ is equipped with a time reversal and the SRE state is symmetric, we call it a symmetric SRE state on $(\caA, \tau)$. A stably SRE state is called symmetric if, additionally, the stabilising product state $\psi_0$ can be chosen symmetric w.r.t. a time reversal $\tau'$ on $\caA'$. If $\psi$ is a symmetric SRE state on $(\caA, \tau)$ and $\psi'$ is a symmetric SRE state on $(\caA', \tau')$ then $\psi \gotimes \psi'$ is a symmetric SRE state on $(\caA \gotimes \caA', \tau \gotimes \tau')$.

\subsection{Examples : Quasi-free symmetric states and the Fu-Kane-Mele index}\label{sec:non interacting examples}

We consider a single-particle Hilbert space $\caK_m = l^2(\Z^2;\C^{2m}) \simeq l^2(\Z^2) \otimes \C^{2m}$ for some positive integer $m$. We write the internal space as $\C^{2m} = \C^2 \otimes \C^m$ and fix an orthonormal tensor product basis $\{e_{\sigma, i}\}_{\sigma \in \{\uparrow, \downarrow\}, i \in \{1,\ldots, m\}}$. Although this will not play a role in the following, one may think of the label $\sigma$ as a spin-1/2 degree of freedom. Denote by $\{  \delta_x \}_{x \in \Z^2}$ the orthonormal position basis of $l^2(\Z^2)$. Then the vectors $\{ \delta_x \otimes e_{\sigma, i} \}$ form an orthonormal basis of $\caK_m$. Denote by $\scrK$ the complex conjugation on $\caK_m$ with respect to this basis, and let $\scrU_T$ be the unitary given by
\begin{equation}
    \scrU_T = \I \otimes \begin{bmatrix}
        0 &1 \\ -1 & 0
    \end{bmatrix} \otimes \I
\end{equation}
with respect to the tensor product decomposition $\caK_m = l^2(\Z^2) \otimes \C^2 \otimes \C^m$. We define an antiunitary time reversal on $\caK_m$ by
\begin{equation}\label{eq:first quantized T}
    \scrT = \scrK \scrU_{T}
\end{equation}
which satisfies $\scrT^2 = -\I$.

Let $\caA$ be the CAR algebra over $\caK_m$. Then $\tau(a(f)) = a(\scrT f)$ defines a fermionic time reversal on $\caA$, \ie $\tau^2 = \theta$ and for any $\Gamma \subset \Z^2$ we have $\tau(\caA_{\Gamma}) = \caA_{\Gamma}$, so $(\caA, \tau)$ is a fermion system with time reversal.

If $\scrP$ is an orthogonal projector on $\caK_m$ then the functional $\omega_{\scrP}$ defined by
\begin{equation}\label{QuasiFreeState}
    \omega_{\scrP}(a^*(f_{n'})\cdots a^*(f_1) a(g_1)\cdots a(g_n)) = \delta_{n',n}\det\big(\langle g_i, \scrP f_j\rangle_{i,j=1}^n\big)
\end{equation}
and extended to $\caA$ by linearity is a pure $U(1)$-invariant quasi-free state which we call the state corresponding to $\scrP$.

A projector $\scrP$ on $\caK_m$ is exponentially local if there are constants $C, \eta > 0$ such that
\begin{equation} \label{eq:locality of gapped Fermi projections}
    \abs{  \langle \delta_x \otimes e_{\sigma, i}, \, \scrP \, \delta_y \otimes e_{\sigma', j} \rangle } \leq C \ed^{-\eta \norm{x - y}_1}
\end{equation}
for all $x, y \in \Z^2$, all $\sigma, \sigma' \in \{ \uparrow \downarrow \}$, and all $i, j \in \{1,\ldots,m\}$. We use the norm $\norm{x}_1 = \abs{x_1} + \abs{x_2}$ to measure spatial distances on $\bbZ^2$.

A projector $\scrP$ on $\caK_m$ is time reversal invariant if $\scrT \scrP \scrT^* = \scrP$. If $\scrP$ is also exponentially local then we can regard $\scrP$ as the Fermi projection of a gapped exponentially local and time reversal invariant Hamiltonian $\scrH$ at some Fermi energy $\mu$. For concreteness, we may take $\scrH = \I - \scrP$ and $\mu = 1/2$. Let us define the corresponding family of flux Hamiltonians $\scrH_{\phi}$ by
\begin{equation} \label{eq:free fermion flux Hamiltonians}
    \scrH_{\phi}(x, i ; y, j) = \begin{cases}
        \ed^{\iu \phi \, \sgn( x_2 - y_2 )} \scrH(x, i ; y, j) & \text{if } x_1, y_1 \leq 0 \\
        \scrH(x, i; y, j) & \text{otherwise}
    \end{cases}
\end{equation}
where for any operator $\scrA \in \caB( \scrK_{m} )$ we write $\scrA(x, i; y, j) = \langle \delta_x \otimes e_i, \scrA \, \delta_y \otimes e_j \rangle$ for its matrix elements w.r.t. the basis $\{ \delta_x \otimes e_i \}_{x, i}$ (to simplify notations, we shall from here on drop the additional index $\sigma$ whenever it is irrelevant). Note that $\phi \mapsto \scrH_{\phi}$ is $2\pi$-periodic.

The Fu-Kane-Mele index of $\scrP$ is given by a spectral flow \cite{DeNittisSchulzBaldes2016spectral}
\begin{equation} \label{eq:FKM is spectral flow}
    \FKM(\scrP) := \SF_{\mu}( [0, \pi] \ni \phi \mapsto \scrH_{\phi} ) \, \mod 2
\end{equation}
for any $\mu \in (0, 1)$ in the gap of $\scrH$.

A quasi-free state $\omega_{\scrP}$ corresponding to an exponentially local and time reversal invariant projector $\scrP$ is called an $\AII$ state, referring to the relevant Altland-Zirnbauer symmetry class.

A special case of $\AII$ states is obtained by taking the trivial projection $\scrP = 0$ on the single particle space $\caK_m$. Then the corresponding quasi-free state $\omega_{\mathrm{empty}} = \omega_{\scrP}$ is characterized by $\omega_{\mathrm{empty}}(Q_x) = 0$ for all $x \in \Z^2$. It follows that $\omega_{\mathrm{empty}}$ is a product state because it restricts to each $\caA_x$ as the unique state on $\caA_x$ with zero expectation value for $Q_x$. Since $\scrT \scrP \scrT^* = \scrP$ we see that $\omega_{\mathrm{empty}}$ is time reversal invariant, so $\omega_{\mathrm{empty}}$ is a symmetric product state, in particular SRE.

\subsection{Symmetry protected phases}

\begin{definition} \label{def:Equivalence}
Two symmetric states $\omega_1$ and $\omega_2$ defined on the same system $(\caA, \tau)$ are called \emph{equivalent} if there is a symmetry preserving LGA $\al$ of $\caA$ such that $\omega_2 = \omega_1 \circ \al$.
\end{definition}

We say that a symmetric SRE state $(\omega_0, \caA, \tau)$ is a \emph{symmetric product state} if $\omega_0$ is a product state. We call two symmetric states $(\omega_1, \caA_1, \tau_1)$ and $(\omega_2, \caA_2, \tau_2)$ \emph{stably equivalent} if there are symmetric product states $(\omega'_1, \caA'_1, \tau'_1)$ and $(\omega'_2, \caA'_2, \tau'_2)$ such that $(\omega_1 \gotimes \omega'_1, \caA_1 \gotimes \caA'_1, \tau_1 \gotimes \tau'_1)$ and $(\omega_2 \gotimes \omega'_2, \caA_2 \gotimes \caA'_2, \tau_2, \gotimes \tau'_2)$ are equivalent.

Denote by $\caP$ the class of all symmetric stably SRE states. Equivalence and stable equivalence are equivalence relations on $\caP$. If $\omega_1, \omega_2 \in \caP$ are stably equivalent then we write $\omega_1 \sim \omega_2$. The stable equivalence classes of $\caP$ are called symmetry protected topological (SPT) phases. Note that any two symmetric product states are stably equivalent. The equivalence class containing any symmetric product state is called the \emph{trivial phase}.

\subsection{The many-body $\bbZ_2$-index}

With these definitions, we can now state our main theorem which posits the existence of a $\bbZ_2$-valued index associated with any symmetric stably SRE state. The $\AII$ states described in Section~\ref{sec:non interacting examples} give examples of symmetric stably SRE states for which the index takes its two possible values. The index is multiplicative under stacking, and it is an invariant under the action of symmetry preserving LGAs and therefore constant on SPT phases. The theorem thus shows that there are two distinct such phases.

\begin{theorem}\label{thm:main theorem}
    There is a map $\Ind : \caP \rightarrow \Z_2 = \{+1, -1\}$ such that:
    \begin{enumerate}
          \item If $\omega_0$ is a product state then $\Ind(\omega_0) = +1$. \\ If $\omega_{\scrP}$ is an $\AII$ state that is  stably SRE, then $\Ind(\omega_{\scrP}) = \FKM(\scrP)$.
          \label{thm:item 1}
        \item  For any $\omega_1,\omega_2\in\caP$, $\Ind(\omega_1 \gotimes \omega_2) = \Ind(\omega_1) \times \Ind(\omega_2)$. \label{thm:item 2}
        \item Let $\omega_1,\omega_2\in\caP$. If $\omega_1\sim\omega_2$ then $\Ind(\omega_1) = \Ind(\omega_2)$. \label{thm:item 3}
    \end{enumerate}
\end{theorem}

The first point of this Theorem says that our index generalizes the Fu-Kane-Mele index of $\AII$ states that are stably SRE. While we expect that all $\AII$ states are stably SRE, currently we can only prove it in the translation invariant case:

\begin{proposition} \label{prop:free SRE}
    Translation invariant $\AII$ states are stably SRE.
\end{proposition}

This proposition, which is proven in Appendix \ref{app:free SRE}, provides plenty of states for which our index takes its possible values.

\section{Definition of the index}\label{sec:definition of the index}

Throughout this section we fix a symmetric SRE state $(\omega, \caA, \tau)$ on a fermion system with time reversal and let $\al = \al^F_1$ be an LGA and $\omega_0$ a homogeneous product state such that $\omega = \omega_0 \circ \al$.

\subsection{Flux insertion}

We first describe a construction which models the adiabatic insertion of a magnetic flux $\phi$ piercing the plane.

\subsubsection{Unique gapped ground states}

\begin{definition} \label{def:unique gapped ground state}
    A pure state $\omega$ is a \emph{ground state} of a 0-chain $F$ if
    \begin{equation}\label{def: GS}
        \omega \big(  A^* [F, A] \big) \geq 0, \quad \forall A \in \caA^{\loc}.
    \end{equation}
    The pure state $\omega$ is a \emph{locally unique gapped ground state} of $F$ with gap $\Delta$ if moreover
    \begin{equation}\label{def: Gap}
        \omega \big( A^* [F, A] \big) \geq \Delta \, \omega( A^* A )
    \end{equation}
    for all $A \in \caA^{\loc}$ for which $\omega(A) = 0$.
\end{definition}

 It follows by density that if $\omega$ is a locally unique gapped ground state with gap $\Delta$, then (\ref{def: GS},\ref{def: Gap}) hold for all $A$ in the domain of the derivation~$\delta^F$, and in particular for all $A\in\caA^\aloc$.

\subsubsection{Symmetric parent Hamiltonian for \texorpdfstring{$\omega$}{}}\label{sec:Parent H}

Here we show the existence of a symmetric parent Hamiltonian for the symmetric SRE state $\omega$ which we can use to describe the flux insertion process.

\begin{definition} \label{def:symmetric parent Hamiltonian}
    A $0$-chain $H$ is a \emph{symmetric parent Hamiltonian} for a symmetric pure state $\omega$ if $\omega$ is a locally unique gapped ground state of $H$ and $H_x = \rho_{\phi}(H_x) = \tau(H_x)$ for all $x$ and all $\phi \in \R$.
\end{definition}
\noindent Note that a symmetric parent Hamiltonian is not time reversal \emph{preserving} as defined in Section~\ref{sec:symmetries}, it is instead time reversal \emph{invariant}, which is not the same thing because time reversal is an anti-linear symmetry.

\begin{lemma} \label{lem:existence of symmetric parent Hamiltonian}
    The symmetric SRE state $\omega$ has a symmetric parent Hamiltonian. 
\end{lemma}

\begin{proof}
Since $\omega_0$ is a homogeneous product state there is a 0-chain $H^{(0)}$ such that $H^{(0)}_x \in \caA_{\{x\}}$ is even for each $x \in \Z^2$ and $\omega_0$ is the unique gapped ground state of $H^{(0)}$ with gap one. Since $\omega = \omega_0 \circ \al$ we claim that $\omega$ is the unique gapped ground state of the 0-chain $H'$ with gap one, where $H'_{x} := \al^{-1}(H^{(0)}_x)$. Indeed, since $\alpha(A)\in \caA^\aloc$, the series $\sum_{x\in\bbZ^2} [H^{(0)}_x,\alpha(A)]$ converges in norm, and therefore $[H',A] = \sum_{x\in\bbZ^2}[ \al^{-1}(H^{(0)}_x),A] = \sum_{x\in\bbZ^2} \alpha([H^{(0)}_x,\alpha(A)])$ is also convergent since $\alpha$ is an automorphism. But then $\omega(A^*[H',A]) = \omega_0(\alpha(A)^*[H^{(0)},\alpha(A)])\geq 0$ for all $A\in\caA^\loc$ by the remark after Definition~\ref{def:unique gapped ground state}. Hence $\omega$ is a ground state of $H'$. If $H'$ has another ground state $\omega'$, then $\omega'\circ\alpha^{-1}$ would be a ground state of $H_0$, and so $\omega$ is the unique ground state of $H'$. That it is gapped follows by a similar argument.  Note that the $H'_x$ are still even, but they may not be symmetric since $\alpha$ is in general not symmetric. Therefore, we average $H'_x$ over the time reversal and $U(1)$ symmetries to obtain a symmetric Hamiltonian $H$:
    \begin{equation}
        H_x := \frac{1}{4\pi} \, \int_{0}^{2\pi} \, \dd \phi \,  \rho_{\phi} \left(  H'_x   +   \tau \big( H'_x \big)  \right).
    \end{equation}
    Since $\omega$ is a symmetric state,
    \begin{equation}
        \omega(A^*[H,A]) = \frac{1}{4\pi} \, \int_{0}^{2\pi} \, \dd \phi \,  \omega\Big(\rho_{-\phi}(A)^*\left[ \left(  H'_x   +   \tau \big( H'_x \big)  \right),\rho_{-\phi}(A)\right]\Big)
    \end{equation}
and the integrand is greater than $2 \omega(A^* A)$ for all $A$ such that $\omega(A)=0$ since $\omega$ is a gapped ground state of $H'$. Hence $\omega$ is a locally unique gapped ground state of $H$ with gap one.
\end{proof}

\subsubsection{$U(1)$ transformations on half-planes and quasi-adiabatic generators}\label{sec:gauge transformation}

A \emph{half-line} in $\R^2$ is determined by a base point $a \in \R^2$ and a non-zero direction vector $v \in \R^2$. To any half-line $l = (a, v)$ we can associate a \emph{half-plane} (with marked boundary)
\begin{equation}
    h_{l} := \{ x \in \R^2 \, :\, (x - a) \cdot (R v) > 0  \}
\end{equation}
where $R$ is the clockwise rotation by $\pi/2$. In the following, we will use the same notation for a subset $S\subset \R^2$ and $S\cap \Z^2\subset \Z^2$. 

Let $h = h_{l}$ be the half-plane associated to the half-line $l=(a, v)$ and consider for any symmetric 0-chain $H$ the 0-chains $H^{h, \phi}$ with
\begin{equation}\label{Def K}
    H^{h, \phi}_x = \rho^h_{-\phi} \big( H_x \big).
\end{equation}
Since each $H_x$ is time reversal and $U(1)$ invariant, and using $\tau \circ \rho_{\phi}^h = \rho_{-\phi}^h \circ \tau$, we find that the $H^{h, \phi}_x$ are $U(1)$ invariant and $\tau( H^{h, \phi}_x ) = H^{h, - \phi}_x$.

If $H$ is a symmetric parent Hamiltonian of a symmetric pure state $\omega$ then $\omega^{h, \phi} := \omega \circ \rho_{\phi}^h$ is the locally unique gapped ground state of $H^{h, \phi}$ with gap one for all $\phi \in \R$. It follows that the states $\omega^{h,\phi}$ are related to $\omega = \omega^{h, 0}$ by the quasi-adiabatic flow~\cite{LRForfermions}, see also~\cite{hastingswen,HastingsQuasiAdiabatic,AutomorphicEq,OgataMoon}, which we briefly recall. Let
\begin{equation}\label{K Ops}
    K^{h, H}_x(\phi) := \int_\bbR \dd t \, W(t) \, \al_t^{H^{h, \phi}} \Big( \frac{\dd H^{h, \phi}_x}{ \dd \phi } \Big)
\end{equation}
where $W : \R \rightarrow \R$ is an odd bounded function such that $\lim_{t \uparrow \infty} |t|^p W(t) = \lim_{t \uparrow \infty} |t|^p W(-t) = 0$ for all $p \in \N$ and such that its Fourier transform $\hat W(E)$ equals $-\iu / E$ whenever $\abs{E} > 1$. Then:
 \begin{enumerate}
        \item The family $K_x^{h, H}(\phi)$ defines a TDI $K^{h, H}$.
        \item There is $g \in \caF$ such that
        \begin{equation}\label{Support of K}
            \norm{ K_x^{h, H}(\phi) } \leq g( \dist(x, \partial h)  )
        \end{equation}
        for all $x \in \Z^2$ and all $\phi \in \R$.
        \item We have
        \begin{equation}\label{K covariances}
            K_x^{h, H}(\phi) = \rho_{\phi'} \big(  K_x^{h, H}(\phi) \big) = \theta \big( K_x^{h, H}(\phi) \big) = \tau \big( K_x^{h, H}(-\phi) \big)
        \end{equation}
        for all $x \in \Z^2$ and all $\phi, \phi' \in \R$.
    \end{enumerate}
    Points 1 and 2 are standard consequences of the Lieb-Robinson bound, the fast decay of $W$ and the $U(1)$ invariance of all $H_x$. The first two equalities of 3 follow from $U(1)$ invariance while the last equality follows from the antilinearity of $\tau$ through $\tau\left(W(t)\al_t^{H^{h, \phi}}(\partial_\phi  H^{h, \phi}_x)\right) = W(t)\al_{-t}^{H^{h,-\phi}}(\partial_\phi  H^{h,-\phi}_x)$ and the fact that $W$ is odd.

The TDI $K^{h, H}$ is called the quasi adiabatic generator for the family $H^{h, \phi}$. It generates a quasi-adiabatic flow $\al_{\phi}^{h, H} := \al_{\phi}^{K^{h, H}}$ which has the following crucial property:
\begin{equation}\label{parallel transport}
    \omega^{h, \phi} = \omega \circ \al_{\phi}^{h, H}
\end{equation}
for all $\phi$, where again $\omega = \omega^{h, 0}$. In other words, the quasi-adiabatic flow $\al_{\phi}^{h, H}$ implements the $U(1)$ transformations $\rho_{\phi}^h$ on the state $\omega$. Unlike the $U(1)$ transformations $\rho_{\phi}^h$ whose generator is supported on the full half-plane $h$, the generator of the quasi-adiabatic flow $K^{h, H}$ is supported only near the line $\partial h$, see~(\ref{Support of K}).

\subsubsection{Defect states}\label{sec:defect states}

The above construction allows us to define an automorphism that corresponds to inserting a magnetic flux at the point $a$ in the plane. Let $l = (a, v)$ be any half-line with base point $a$. We define a TDI $K^{l, H}$ as a restriction of the TDI $K^{h, H}$ as follows:
\begin{equation}\label{Flux Insertion K}
    K^{l ,H}_x(\phi) = \begin{cases}  K^{h, H}_x(\phi) \quad & \text{if} \,\,  (x - a) \cdot v > 0 \\
    0 & \text{otherwise}\end{cases}
\end{equation}
for all $\phi \in \R$. We denote by $$\flux_{\phi}^{l, H} := \al_{\phi}^{K^{l , H}}$$ the LGA it generates, and we shall refer to such LGAs as \emph{flux insertion automorphisms}. We define \emph{defect states} by
\begin{equation}\label{def of defect state}
    \omega^{l, H, \phi} := \omega \circ \flux_{\phi}^{l, H}.
\end{equation}
As emphasized by the notation, the states $\omega^{l ,H, \phi}$ depend on the symmetric parent Hamiltonian $H$, on the choice of half-line $l$, and on the flux $\phi$.

\begin{lemma} \label{lem:properties of flux insertion}
    Let $H$ be a symmetric parent Hamiltonian of $\omega$ and let $l$ be a half-line. For any $\phi, \phi' \in \R$ we have
    \begin{equation}
        \flux_{\phi}^{l, H} \circ \rho_{\phi'} = \rho_{\phi'} \circ \flux_{\phi}^{l, H}, \quad  \flux_{\phi}^{l, H} \circ \theta = \theta \circ \flux_{\phi}^{l, H} \quad \text{and} 
 \,\,\, \flux_{\phi}^{l, H} \circ \tau = \tau \circ \flux_{-\phi}^{l, H}.
    \end{equation}
    It follows in particular that
    \begin{equation}
        \omega^{l, H, \phi} = \omega^{l, H, \phi} \circ \rho_{\phi'} = \overline{\omega^{l, H, -\phi}} \circ \tau.
    \end{equation}
\end{lemma}

\begin{proof}
The first two identities are immediate consequences of~(\ref{K covariances}). For the last one, we let $\tilde \flux_{\phi}^{l, H} := \tau^{-1} \circ \flux_{\phi}^{l, H} \circ \tau$. Then for any $A \in \caA$
\begin{align}
    \tilde \flux_{\phi}^{l, H}(A)
    &= A - \iu\int_0^\phi  \dd \phi' \, \tilde \flux_{\phi'}^{l, H}\left([\tau^{-1}(K^{l, H}(\phi')),A]\right)
\end{align}
where we used the antilinearity of $\tau^{-1}$. By~(\ref{K covariances}), $\tau^{-1}(K^{l, H}(\phi')) = \tau(K^{l, H}(\phi')) = K^{l, H}(-\phi')$. Hence, $\tilde \flux_{\phi}^{l, H}(A)$ solves the same equation~(\ref{LGA DE}) as $\flux_{-\phi}^{l, H}(A)$, and since $A$ is arbitrary we conclude that $\tau^{-1} \circ \flux_{\phi}^{l, H} \circ \tau = \flux_{-\phi}^{l, H}$ indeed. The identity $\omega^{l, H, \phi} = \omega^{l, H, \phi} \circ \rho_{\phi'}$ is now immediate. As for the last one, we use the invariance of $\omega$ under parity $\theta = \tau^2$ to conclude that
\begin{equation}
    \omega^{l, H, \phi} = \omega\circ\tau^2\circ \flux_{\phi}^{l, H}
    =\omega\circ\tau\circ \flux_{-\phi}^{l, H}\circ\tau
    =\bar\omega\circ \flux_{-\phi}^{l, H}\circ\tau
\end{equation}
since $\omega$ is time reversal invariant.
\end{proof}

\subsubsection{Locality}\label{sec:locality}

We conclude this section with a locality result for LGAs generated by TDIs satisfying~(\ref{Support of K}).

Let $l = (a,v)$ be a half-line. For any $\varsigma \in (0, 2\pi)$, the open subset of $\R^2$ given by
\begin{equation}
	\Lambda_{l, \varsigma} := \{ x \in \R^2 \, : \, (x - a) \cdot v > \norm{x-a} \cos (\varsigma/2)   \}.
\end{equation}
will be called the \emph{cone} with apex at $a \in \R^2$, axis $v \in \R^2$ of unit length, and opening angle $\varsigma$.

\begin{lemma} \label{lem:Locality along line}
    Let $l$ be a half-line and let $\Lambda_{l,\varsigma}$ be a cone. Let $F$ be an $f$-local TDI. Assume moreover that there is $g\in\caF$ such that $\Vert F_x\Vert \leq g(\mathrm{dist}(x,l))$ for all $x\in\bbZ^2$. Then for any $s_0\in\bbR$ and any $\varsigma \in (0, 2\pi)$, there is $h\in\caF$ such that 
    \begin{equation}
        \Vert \alpha_s^F(A) - A\Vert\leq \Vert A\Vert h(r)
    \end{equation}
    for all $|s|\leq s_0$ and all $A\in\caA_{\Lambda_{l,\varsigma}^c\cap B_a(r)^c}$. Moreover, the same holds if $A$ is replaced by $\alpha^G_{s'}(A)$ for any TDI $G$.
\end{lemma}
\begin{proof}
For the duration of the proof, we write $\Lambda = \Lambda_{l,\varsigma}$ and $\tilde \Lambda = \Lambda_{l,\varsigma/2}$, see Figure~\ref{fig:cones}.
\begin{figure}[ht]
\centering \includegraphics[width=0.3\textwidth]{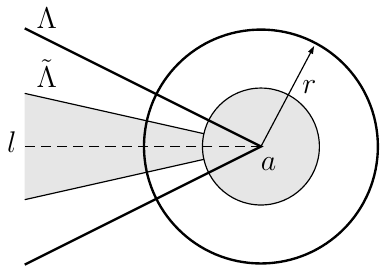}
\caption{The sets forming the keyhole.}
\label{fig:cones}
\end{figure}
If $x\in\tilde\Lambda\cup B_a(\frac{r}{2})$ belongs to the `keyhole' then
\begin{equation}
    \Vert [F_x,A]\Vert
    =\Vert [F_x-F_{x,\mathrm{dist}(x,\Lambda^c\cap B_a(r)^c)},A]\Vert
    \leq 2\Vert A\Vert f(\mathrm{dist}(x,\Lambda^c\cap B_a(r)^c)
\end{equation}
since $F_x$ is $f$-localized near $x$, see Section~\ref{sec:TDI & LGA}. On the other hand, if $x\in\tilde\Lambda^c\cap B_a(\frac{r}{2})^c$, then
\begin{equation}
    \Vert [F_x,A]\Vert \leq 2\Vert F_x\Vert \Vert A\Vert
    \leq 2\Vert A\Vert g(\mathrm{dist}(x,l))
\end{equation}
We now sum these estimates over $x$. If $x\in B_a(\frac{r}{2})$, then $(\mathrm{dist}(x,\Lambda^c\cap B_a(r)^c)\geq \frac{r}{2}$. If $x\in \tilde\Lambda\setminus B_a(\frac{r}{2})$, then $(\mathrm{dist}(x,\Lambda^c\cap B_a(r)^c)\geq \frac{\theta}{4}\Vert x\Vert$. Hence,
\begin{equation}
    \sum_{x\in\tilde\Lambda\cup B_a(\frac{r}{2})}\Vert [F_x,A]\Vert
    \leq 2\Vert A\Vert \Big( c r^2 f(\frac{r}{2}) + \sum_{x\in\tilde\Lambda,\Vert x\Vert>r/2}f(\frac{\theta}{4}\Vert x\Vert)\Big)
\end{equation}
which is bounded above by $2\Vert A\Vert \tilde f(r)$ for some $\tilde f\in\caF$. For the complement of the keyhole, we have that $\mathrm{dist}(x,l)\geq \frac{\theta}{4}\Vert x\Vert$ and so by assumption
\begin{equation}
    \sum_{x\in\tilde\Lambda^c\cap B_a(\frac{r}{2})^c}\Vert [F_x,A]\Vert
    \leq 2\Vert A\Vert \sum_{x\in\tilde\Lambda^c,\Vert x\Vert>r/2}g(\frac{\theta}{4}\Vert x\Vert)
\end{equation}
which is again bounded by $2\Vert A\Vert \tilde g(r)$ for some $\tilde g\in\caF$. We conclude that there is $\tilde h\in\caF$ such that
\begin{equation}
    \Vert[F,A]\Vert\leq \Vert A\Vert \tilde h(r).
\end{equation}
The first claim of the lemma now follows from~(\ref{LGA DE}). For the second claim, it suffices to use~(\ref{eq:LGA commutation}) to conclude that
\begin{equation}
    \Vert \alpha_s^F(\alpha_r^G(A)) - \alpha_r^G(A)\Vert
    =\Vert (\alpha_r^G)^{-1}\circ \alpha_s^F\circ \alpha_r^G(A) - A\Vert
    =\Vert \alpha_s^{(\alpha_r^G)^{-1}(F)}(A) - A\Vert
\end{equation}
and note that $(\alpha_r^G)^{-1}(F)$ is a TDI satisfying the assumptions of the lemma.
\end{proof}

\subsection{Defect states at \texorpdfstring{$\pm \pi$}{} flux}

For the remainder of this section we fix a half-line $l$ and a symmetric parent Hamiltonian $H$ of $\omega$, and write $\omega^{\pm} := \omega^{l, H, \pm\pi}$. These two states correspond to inserting a $\pm \pi$ magnetic flux defect into $\omega$. Our main goal is to show that these states are unitary equivalent and related to each other by time reversal. We can then study the transformation properties of vector representatives of these states under time reversal, finally leading to the definition of the index.

\subsubsection{Almost local equivalence}

We show that the states $\omega^{\pm}$ are almost local perturbations of each other, which will imply by the results of Appendix \ref{app:almost local transitivity} that they are unitarily equivalent through an operator that is almost local.

\begin{definition} \label{def:f-close}
    Let $f \in \caF$ and let $\psi$ and $\psi'$ be states on $\caA$. We say $\psi$ and $\psi'$ are \emph{$f$-close} if for all $r > 0$ we have
    \begin{equation}
        \abs{ \psi(A) - \psi'(A) } \leq f(r) \norm{A} \quad \forall A \in \caA^{\loc}_{B(r)^c}.
    \end{equation}
    The state $\psi'$ is an \emph{almost local perturbation} of the state $\psi$ if $\psi,\psi'$ are $f$-close for some $f\in\caF$.
\end{definition}

Recall that $\omega = \omega_0 \circ \al$ for some $\al = \al^F_1$ and define homogeneous states
\begin{equation}\label{tilde defect state}
    \tilde \omega^\pm := \omega^\pm \circ (\al \circ \flux_{\pi}^{l, H})^{-1}.
\end{equation}
It follows immediately from these definitions that
\begin{equation}
    \tilde \omega^+ = \omega_0,\qquad 
    \tilde \omega^- =  \omega_0 \circ \al \circ \flux_{-\pi}^{l, H} \circ \big( \flux_{\pi}^{l, H} \big)^{-1} \circ \al^{-1}.
\end{equation}

\begin{lemma} \label{lem:f-close on half-planes}
    Let $\Lambda = \Lambda_{l, \pi}$. There is $f \in \caF$ such that for any $A \in \caA^{\loc}_{\Lambda\cap B_a(r)^c}\cup\caA^{\loc}_{\Lambda^c\cap B_a(r)^c}$ we have
    \begin{equation}
        \abs{ \tilde \omega^-(A) - \tilde \omega^+(A) } \leq f(r) \norm{A}.
    \end{equation}
    for all $r>0$.
\end{lemma}

\begin{proof}
    In this proof, $f$ stands for a function in $\caF$ and it changes from equation to equation. We write $\tilde \omega^- = \omega_0 \circ \tilde \beta$, where $\tilde \beta := \al \circ \flux_{-\pi}^{l, H} \circ \big( \flux_{\pi}^{l, H} \big)^{-1} \circ \al^{-1}$. Let first $A\in\Lambda\cap B_a(r)^c$. Let $h = h_{l}$ be the half-plane associated to the half-line $l = (a, v)$ and let 
    \begin{equation}
        \beta = \al \circ \al_{-\pi}^{h, H} \circ \big( \al_{\pi}^{h, H} \big)^{-1} \circ \al^{-1}.
    \end{equation} 
    We claim that
    \begin{equation}\label{beta betatilde}
        \left\Vert\tilde \beta(A) - \beta(A)\right\Vert \leq f(r) \norm{A}
    \end{equation}
    where $f$ is independent of $A$. We first prove the claim with the replacements $\tilde \beta\to \flux_{-\pi}^{l, H} \circ \big( \flux_{\pi}^{l, H} \big)^{-1}$ and $\beta\to \al_{-\pi}^{h, H} \circ \big( \al_{\pi}^{h, H} \big)^{-1}$ since this yields the claim by the second part of Lemma~\ref{lem:Locality along line}. The family of automorphisms $\big(\flux_{\phi}^{l, H}\big)^{-1} \circ \al_{\phi}^{h, H}$ is an LGA generated by
    \begin{equation}\label{eq:complicated G}
        (K^{h,H}(\phi) - K^{l,H}(\phi)) + \iu\big(\alpha^{h,H}_\phi\big)^{-1}
        \Big[\int_0^\phi \dd u \,  \alpha^{h,H}_u\left(\left[K^{h,H}(u) - K^{l,H}(u), \gamma^{l,H}_{\phi,u}(K^{l,H}(u))\right]\right) \Big],
    \end{equation}
 see e.g.~\cite[Lemma~5.4]{bachmann2022classification}. If we denote by $-l$ the half-line $-l=(a,-v)$, we have that  $K^{l,H} - K^{h,H} = K^{-l,H}$ which satisfies the assumptions of Lemma~\ref{lem:Locality along line}, hence so does the TDI~(\ref{eq:complicated G}). The Lemma (with $\Lambda$ replacing $\Lambda^c$ since it is applied to the cone $\Lambda^c$) shows that
    \begin{equation}
        \left\Vert \big(\flux_{-\pi}^{l, H}\big)^{-1} \circ \al_{-\pi}^{h, H}\left((\alpha^{h,H}_\pi\big)^{-1}(A)\right) - \left((\alpha^{h,H}_\pi\big)^{-1}(A)\right)\right\Vert\leq f(r) \Vert A\Vert
    \end{equation}
    for all $A\in\Lambda\cap B_a(r)^c$. Similarly, $\flux_{\phi}^{l, H} \circ\big( \al_{\phi}^{h, H} \big)^{-1}$ is generated by $-\al_{\phi}^{h, H}(K^{h,H}(\phi) - K^{l,H}(\phi))$ which implies that
    \begin{equation}
        \left\Vert
        A 
        -\flux_{\pi}^{l, H} \circ\big( \al_{\pi}^{h, H} \big)^{-1}(A)
        \right\Vert \leq f(r) \Vert A\Vert.
    \end{equation}
    These estimates yield the claim~(\ref{beta betatilde}) since
    \begin{multline*}
        \left\Vert
        \flux_{-\pi}^{l, H} \circ \big( \flux_{\pi}^{l, H} \big)^{-1}(A)
        -\al_{-\pi}^{h, H} \circ \big( \al_{\pi}^{h, H} \big)^{-1}(A)
        \right\Vert 
        =\left\Vert
        A 
        -\flux_{\pi}^{l, H} \circ\big( \flux_{-\pi}^{l, H} \big)^{-1}\circ\al_{-\pi}^{h, H} \circ \big( \al_{\pi}^{h, H} \big)^{-1}(A)
        \right\Vert  \\
        \leq 
        \left\Vert
        A 
        -\flux_{\pi}^{l, H} \circ\big( \al_{\pi}^{h, H} \big)^{-1}(A)
        \right\Vert
        +\left\Vert
        \flux_{\pi}^{l, H} \circ\big( \al_{\pi}^{h, H} \big)^{-1}(A)
        -\flux_{\pi}^{l, H} \circ\big( \flux_{-\pi}^{l, H} \big)^{-1}\circ\al_{-\pi}^{h, H} \circ \big( \al_{\pi}^{h, H} \big)^{-1}(A)
        \right\Vert.
    \end{multline*}

    Now note that
    \begin{equation}
        \omega \circ \al_{-\pi}^{h, H} \circ \big( \al_{\pi}^{h, H} \big)^{-1} = \omega^{h, -\pi} \circ \big( \al_{\pi}^{h, H} \big)^{-1} = \omega^{h, \pi} \circ \big( \al_{\pi}^{h, H} \big)^{-1} = \omega,
    \end{equation}
    where the first and last equalities are by the definition of the states~(\ref{parallel transport}) and the second one follows from $\omega^{h, -\pi} = \omega \circ \rho_\pi =\omega \circ \rho_{-\pi} = \omega^{h, \pi}$. It follows that $\omega_0 \circ \beta = \omega_0 = \tilde \omega^+$ and so for all $A\in\Lambda\cap B_a(r)^c$ we find
    \begin{equation}
        \abs{\tilde \omega^-(A) -  \tilde \omega^+(A)} = \abs{ \omega_0 \circ \left(  (\tilde \beta - \beta)(A) \right)} \leq f(r) \norm{A}.
    \end{equation}
    If $A\in\Lambda^c\cap B_a(r)^c$, then Lemma~\ref{lem:Locality along line} implies that $\big\Vert\big( \flux_{\phi}^{l, H} \big)^{-1}(A)-A\big\Vert\leq f(r)\Vert A\Vert$ by similar arguments and therefore that
    \begin{equation}
        \left\Vert\tilde \beta(A) - A \right\Vert \leq f(r) \norm{A}
    \end{equation}
    instead of~(\ref{beta betatilde}), with the same conclusion as above since $\omega_0 = \tilde \omega^+$.
\end{proof}

\begin{proposition} \label{prop:almost local unitary equivalence}
    There exists a homogeneous unitary $U \in \caA^{\aloc}$ such that $\omega^+ = \omega^- \circ \Ad[U]$.
\end{proposition}

\begin{proof}
    Since $\tilde \omega^+ = \omega_0$ is a product state, it follows from Lemma~\ref{lem:f-close on half-planes} and Lemma \ref{lem:f-close on cones implies f-close} that the homogeneous state $\tilde \omega^+$ is an almost local perturbation of $\tilde \omega^-$, see Definition~\ref{def:f-close}. By Lemma \ref{prop:f-close implies almost local unitary equivalence} this implies that there is a homogeneous unitary $\widetilde U \in \caA^{\aloc}$ such that $\tilde \omega^+ = \tilde \omega^- \circ \Ad[\widetilde U]$. Recalling that $\omega^{\pm} = \tilde \omega^{\pm} \circ (\al \circ \flux_{\pi}^{l, H})$, we have that
    \begin{equation}
        \omega^+ = \tilde \omega^{+} \circ (\al \circ \flux_{\pi}^{l, H})
        =\tilde \omega^- \circ \Ad[\widetilde U]\circ (\al \circ \flux_{\pi}^{l, H})
        =\tilde \omega^- \circ(\al \circ \flux_{\pi}^{l, H})\circ \Ad[U] 
        =\omega^-\circ \Ad[U] 
    \end{equation}
    where $U = (\al \circ \flux_{\pi}^{l, H})^{-1}(\widetilde U)$ is again almost local by~(\ref{LR bound}), and is homogeneous since $\al$ and $\flux_{\pi}^{l, H}$ commute with parity.
\end{proof}

\subsubsection{GNS representation and time reversal}

Let $(\Pi, \caH, \Omega^{+})$ be the GNS triple of $\omega^+$ and fix a homogeneous unitary $U \in \caA^{\aloc}$ such that $\omega^+ = \omega^- \circ \Ad[U]$, which exists be Proposition \ref{prop:almost local unitary equivalence}. Then the unit vector $\Omega^- := \Pi( U ) \, \Omega^+$ represents the state $\omega^-$:
\begin{equation}
    \omega^-(A) = \langle \Omega^-, \, \Pi(A) \, \Omega^- \rangle
\end{equation}
for all $A \in \caA$. We now investigate the transformation properties of the vectors $\Omega^{\pm}$ under time reversal. To this end we first note the following general fact.

\begin{lemma} \label{lem:antiunitary implementation}
	Let $\nu : \caA \rightarrow \C$ be a pure state that is invariant under an antilinear automorphism $\tau : \caA \rightarrow \caA$, namely $\nu\circ\tau = \bar\nu$. Let $(\Pi_{\nu}, \caH_{\nu}, \Omega_\nu)$ be the GNS triple of $\nu$. Then there exists a unique antiunitary operator $T$ acting on $\caH_\nu$ such that $T  \Omega_\nu = \Omega_\nu$ and
	\begin{equation}\label{T implementation}
		T \, \Pi_{\nu}(A) \, T^* = \Pi_{\nu} \big( \tau(A) \big)
	\end{equation}
	for all $A \in \caA$.
\end{lemma}

\begin{proof}
 Let $K$ be a conjugation on $\caH_{\nu}$ such that $K\Omega_\nu = \Omega_\nu$. Then the map $\tilde \Pi_{\nu} = \Ad[K] \circ \Pi_{\nu} \circ \tau : \caA \rightarrow \caB(\caH_\nu)$ is a (linear) *-representation of $\caA$ on $\caH_{\nu}$. Moreover, $\Omega_\nu$ is cyclic for $\tilde \Pi_{\nu}$ and
	\begin{equation}
		\langle \Omega_\nu, \tilde \Pi_{\nu} (A) \, \Omega_\nu \rangle = \langle  \Pi_{\nu}( \tau(A))K\Omega_\nu, K \Omega_\nu \rangle = \overline{\nu( \tau(A) ) }= \nu(A)
	\end{equation}
	for all $A \in \caA$. It follows that $(\tilde \Pi_{\nu}, \caH_{\nu}, \Omega_\nu)$ is a GNS triple for $\nu$. By uniqueness of the GNS representation, there is a unique unitary $V$ such that $V \Omega_\nu = \Omega_\nu$ and
	\begin{equation}
		\tilde \Pi_{\nu}(A) = V \Pi_{\nu}(A) V^*.
	\end{equation}
	This is equivalent to
	\begin{equation}
		\Pi_{\nu}( \tau(A) ) = (KV) \, \Pi_{\nu}(A) \, (KV)^*,
	\end{equation}
	and $(KV) \Omega_\nu = \Omega_\nu$. Hence (\ref{T implementation}) holds for the choice $T = KV$.

	If $T'$ were another antiunitary that satisfies the conditions of the lemma, then $T' = K V'$ for a unitary $V' = K T'$, and $V' \Omega_\nu = \Omega_\nu$. Moreover,
	\begin{equation}
		V' \, \Pi_{\nu}(A) \, (V')^* = K T' \, \Pi_{\nu}(A) \, (T')^* K = (\Ad[K] \circ \Pi_\nu \circ \tau)(A) = \tilde \Pi_\nu(A).
	\end{equation}
	Since $V$ was the unique unitary with these properties, we conclude that $V' = V$ and therefore $T' = K V' = KV = T$, showing uniqueness.
\end{proof}

With this in hand, we come back to our specific setting with $(\Pi, \caH, \Omega^+)$ the GNS triple of $\omega^+$ and $\Omega^- \in \caH$ given by $\Omega^- = \Pi(U) \Omega^+$.

\begin{lemma} \label{lem:TR implementation}
	There is a unique antiunitary $T$ acting on the GNS Hilbert space $\caH$ such that $ \Omega^{-}=T \Omega^+$ and such that $T$ implements time reversal, \ie
	\begin{equation}
		\Pi(\tau(A)) = T \Pi(A) T^*
	\end{equation}
	for all $A \in \caA$.
\end{lemma}

\begin{proof}
    The unitary $U \in \caA^{\aloc}$ is such that
    \begin{equation}
		\omega^- = \omega^+ \circ \Ad[U^*].
    \end{equation}
    Since $\omega^{+} = \overline{\omega^{-}} \circ \tau$ by Lemma~\ref{lem:properties of flux insertion}, it follows that
    \begin{equation}
		\omega^{-} = \overline{\omega^{-}} \circ \tau \circ \Ad[U^*].
    \end{equation}
    In other words, the antilinear automorphism $\tilde \tau = \tau \circ \Ad[U^*]$ leaves the state $\omega^{-}$ invariant. By Lemma \ref{lem:antiunitary implementation} it follows that there is a unique antiunitary $\tilde T$ such that
    \begin{equation}
		\Pi \big( \tilde \tau( A ) \big) 
        = \tilde T \, \Pi(A) \tilde T^*
    \end{equation}
    for all $A \in \caA$, and $\tilde T \Omega^- = \Omega^-$. It follows that $\tilde T \Pi(U) \Omega^+ = \tilde T \Omega^- =  \Omega^-$ and
    \begin{equation}
		\Pi \big(  \tau(A) \big)
        = \Pi\left(\tilde\tau\circ\Ad[U](A)\right)
        = \tilde T \Pi(U) \Pi(A) \Pi(U^*) \tilde T^*
    \end{equation}
    for all $A \in \caA$. We see that $T := \tilde T \Pi(U)$ satisfies the requirements of the lemma.
    
    Suppose now that $T'$ is another such operator. Then $\tilde T' = T' \Pi(U^*)$ satisfies
    \begin{equation}
		\tilde T' \Pi(A) (\tilde T')^* = \Pi \big(  (\tau \circ \Ad[U^*])(A) \big) = \Pi( \tilde \tau(A) )
    \end{equation}
    and $\tilde T' \Omega^- = T' \, \Pi( U^* ) \Omega^- = T' \Omega^+ = \Omega^-$. The uniqueness of $\tilde T$ implies that $\tilde T' = \tilde T$ and hence $T' = T$.
\end{proof}

\subsubsection{The states \texorpdfstring{$\omega^+$}{} and \texorpdfstring{$\omega^-$}{} have the same parity}

With these preliminary constructions, we can now prove a property that is essential to the definition of the index: that the defect states $\omega^\pm$ have the same parity.

Recall that $(\Pi, \caH, \Omega^+)$ is the GNS triple of $\omega^+$ and we fixed a homogeneous unitary $U \in \caA^{\aloc}$ such that $\omega^+ = \omega^{-} \circ \Ad[U]$, whose existence is guaranteed by Proposition \ref{prop:almost local unitary equivalence}. Let $\Omega^- = \Pi(U) \Omega^+$ be the corresponding vector representative of $\omega^-$. Since $\omega^+$ is homogeneous there is a unique unitary $\Theta \in \caB(\caH)$ such that
\begin{equation}\label{eq: implementing parity}
    \Theta \Omega^+ = \Omega^+,\qquad 
    \Pi \big( \theta(A) \big) =  \Theta \, \Pi(A) \, \Theta^*
\end{equation}
for all $A \in \caA$. The fact that $\omega^-$ is also homogeneous implies that $\Omega^-$ is an eigenvector of the operator $\Theta^* = \Theta$ (since $\Theta^2 = \I$). We conclude that $\Theta \Omega^- = \pm \Omega^-$. We note that the property $\Theta \Omega^+ = \Omega^+$ is really a choice of the overall sign of $\Theta$, so the eigenvalue $\pm1$ above should be understood as relative to that choice.

\begin{lemma} \label{lem:same parity}
    Let $\Theta$ be such that~(\ref{eq: implementing parity}) hold. Then $\Theta \Omega^- = \Omega^-$ and there is an even unitary $U \in \caA^{\aloc}$ such that $\omega^+ = \omega^- \circ \Ad[U]$.
\end{lemma}

In the following, we shall say that two states have the same parity whenever they are, as above, unitarily equivalent through an even unitary element of the algebra. 

\begin{proof}
    Let $T$ be the antiunitary operator on $\caH$ provided by Lemma \ref{lem:TR implementation}. Since $\omega^+$ is invariant under the $2\pi$-periodic family of automorphisms $\rho_{\phi}$ for $\phi \in \R$ there is a unique strongly continuous 1-parameter group of unitaries $V_{\phi} \in \caB(\caH)$ that implement the $U(1)$ transformations, namely $V_\phi\Omega^+ = \Omega^+$ and
    \begin{equation}
        V_\phi\Pi(A)V_\phi^* = \Pi(\rho_\phi(A))
    \end{equation}
    for all $A\in\caA$. The map $\phi\mapsto V_\phi$ is $2 \pi$-periodic and such that $V_0 = \I$. Since $\omega^-$ is also invariant under $\rho_{\phi}$, and $\Pi$ is irreducible since $\omega^+$ is pure, we have $V_{\phi} \Omega^- = \ed^{\iu n \phi} \Omega^-$ for all $\phi \in \R$ and some $n \in \Z$.

    The $2 \pi$-periodic family of unitaries $\tilde V_{\phi} := \ed^{-\iu n \phi} T^* V_{\phi} T$ implements $\tau^{-1} \circ \rho_{\phi} \circ \tau = \rho_{-\phi}$ in the GNS representation:
    \begin{equation}
        \tilde V_\phi\Pi(A)\tilde V_\phi^* = T^* V_\phi T\Pi(A) T^* V_\phi^* T = \Pi(\tau^{-1}\circ\rho_\phi\circ\tau(A)) = \Pi( \rho_{-\phi}(A) ).
    \end{equation}
    Moreover,
    \begin{equation}
        \tilde V_{-\phi} \Omega^+ = \ed^{\iu n \phi} \, T^* V_{-\phi} \, T \, \Omega^+ = \ed^{\iu n \phi} \, T^* \, V_{-\phi} \, \Omega^- = \ed^{\iu n \phi} \, T^* \, \ed^{\iu n \phi} \, \Omega^- = \Omega^+,
    \end{equation}
    where we used antilinearity of $T^*$. It follows from uniqueness of the $V_{\phi}$ that $\tilde V_{-\phi} = V_{\phi}$, in other words,
    \begin{equation} \label{eq:time reversal of U(1)-symmetry in the GNS}
        V_{\phi} = \ed^{\iu n \phi} \, T^* \, V_{-\phi} \, T.
    \end{equation}

    Now note that $V_{\pi} = \Theta$ by uniqueness of $\Theta$, and since $\tau^2 = \theta$ we have that the unitary $T^2$ also implements parity in the GNS representation. It follows that $T^2 = \lambda \Theta$ for some $\lambda \in U(1)$. In particular, $T^2 =\lambda V_{\pi}$ commutes with $V_{\phi}$ for all $\phi \in \R$. Using~\eqref{eq:time reversal of U(1)-symmetry in the GNS} twice yields
    \begin{equation}
        V_{\phi} = \ed^{\iu n \phi} \, T^* \, \big( \ed^{-\iu n \phi} \, T^* \, V_{\phi} \, T \big) \, T = \ed^{2 \iu n \phi} \, (T^*)^2 \, V_{\phi} \, T^2 = \ed^{2 \iu n \phi} \, V_{\phi}
    \end{equation}
    for all $\phi \in \R$. We conclude that $n = 0$, so $\Theta \Omega^- = V_{\pi} \Omega^- = \Omega^-$, as required.

    The existence of a homogeneous unitary $U \in \caA^{\aloc}$ such that $\omega^+ = \omega^- \circ \Ad[U]$ is guaranteed by Proposition \ref{prop:almost local unitary equivalence}, so we only need to show that this unitary is even. Since $\Omega^-$ was chosen so that $\Omega^- = \Pi(U)\Omega^+$. By the above we conclude that $\Theta^* \Omega^- = \Pi(U)\Theta^* \Omega^+$ and hence $\Omega^- = \Pi(\theta(U))\Omega^+$. Since $U$ is homogeneous, this is compatible with $\Omega^- = \Pi(U)\Omega^+$ only if $U$ is even.
\end{proof}

\subsection{Definition of the index}

\begin{definition} \label{def:Kramers pairs and Kramers singlets}
    Let $(\Pi, \caH)$ be an irreducible representation of a fermion system with time reversal $(\caA, \tau)$ such that the time reversal is implemented by an antiunitary operator $T$ acting on $\caH$, namely $\Pi( \tau(A) ) = T \Pi(A) T^*$ for all $A \in \caA$.

    \noindent A vector $\Psi \in \caH$ is a Kramers singlet if $T^2 \Psi = \Psi$, and $\Psi$ belongs to a Kramers pair if $T^2 \Psi = -\Psi$.

    \noindent If $(\Pi, \caH, \Psi)$ is a GNS representation of a pure state $\psi$, then we say that $\psi$ is a Kramers singlet if $\Psi$ is, and we say that $\psi$ belongs to a Kramers pair if $\Psi$ does.
\end{definition}

\begin{remark}
    If $\Psi$ belongs to a Kramers pair then if $\Psi' = T \Psi$ we have
    \begin{equation}
        \langle \Psi, \Psi' \rangle = \langle \Psi, T \Psi \rangle = \langle T^2 \Psi, T \Psi \rangle = - \langle \Psi, \Psi' \rangle
    \end{equation}
    which shows that $\Psi$ and $\Psi' = T \Psi$ are orthogonal. The `Kramers pair' to which $\Psi$ belongs is the pair of vectors $\Psi$ and $\Psi'$.
\end{remark}

We now show that the defect state $\omega^+$ is either a Kramers singlet, or it belongs to a Kramers pair with the defect state $\omega^-$. Note first that this statement makes sense since Lemma \ref{lem:TR implementation} shows that time reversal is implemented by an antiunitary operator $T$ in the GNS representation $(\Pi, \caH, \Omega^+)$ of $\omega^+$.

\begin{lemma} \label{lem:Kramers or not}
    The defect state $\omega^+$ is either a Kramers singlet or it belongs to a Kramers pair.
\end{lemma}

\begin{proof}
    Let $\Theta$ be as in~(\ref{eq: implementing parity}). Since $\tau^2 = \theta$ we must have $T^2 = \lambda \Theta$ for some phase $\lambda \in U(1)$. It follows that
    \begin{equation}\label{T2Omega+}
        T^2 \Omega^+ = \lambda \Theta \Omega^+ = \lambda \Omega^+.
    \end{equation}
    By Lemma \ref{lem:same parity} the vector $\Omega^- = T \Omega^+$ satisfies $\Theta \Omega^- = \Omega^-$. By acting on both sides of~(\ref{T2Omega+}) with $T$ we obtain 
    \begin{equation}
        T^3 \Omega^+ = \bar \lambda T \Omega^+ = \bar \lambda\Omega^-.
    \end{equation}
    On the other hand,
    \begin{equation}
        T^3 \Omega^+ = \lambda \Theta T \Omega^+ = \lambda \Theta \, \Omega^- = \lambda \, \Omega^- .
    \end{equation}
    We conclude that $\lambda$ is real, hence $\lambda \in \{+1, -1\}$ as claimed.
\end{proof}

\begin{definition} \label{def:index of triple}
    Let $\omega$ be a symmetric SRE state with symmetric parent Hamiltonian $H$. For any half-line~$l$ we define
    \begin{equation}
        \Ind( \omega, H, l) := \begin{cases}
            +1 \qquad &\text{if } \omega^+ \text{ is a Kramers singlet,} \\
            -1 \qquad &\text{if } \omega^+ \text{ belongs to a Kramers pair}.
        \end{cases}
    \end{equation}
\end{definition}
We will show in Section \ref{sec:independence} that $\Ind(\omega, H, l)$ is actually independent of the choice of symmetric parent Hamiltonian $H$ and the choice of half-line $l$. We will anticipate this result and denote the index from now on by $\Ind(\omega)$.

Note that the index is for now defined only for SRE states and not yet for the larger class of stably SRE states. This extension will be carried out in Section~\ref{sec:extension to stably SRE} as a consequence of the multiplicativity under stacking and the triviality of the index for product states.

Before we study the properties of the index we have just defined, we first show that it has an expression that does not rely on the GNS representation at all. While this expression loses it transparency to physical interpretation, this more `algebraic' formulation hints at the announced invariance under symmetry preserving LGAs. Both formulations can be used in the proofs of the next section but we have chosen to consistently rely only on Definition~\ref{def:index of triple}.

\begin{proposition}\label{prop:Algebraic Index}
    Let $\omega^-$ be the $-\pi$ flux state and let $U\in\caA$ be a homogeneous unitary such that $\omega^+ = \omega^- \circ \Ad[U]$. Then
    \begin{equation}
    \mathrm{Ind}_2(\omega) = \omega^-\big(U \tau(U)\big).
    \end{equation}
\end{proposition}
\begin{proof}
    We recall first that the index is the value of $T^2$ in the subspace spanned by the GNS vectors $\Omega^\pm$ for $\omega^{\pm}$ and secondly that we can choose the phases of $\Omega^{\pm}$ and of $U$ so that $\Omega^- = T \Omega^+$ and $\Omega^- = \Pi( U ) \, \Omega^+$. It follows that
    \begin{align}
    \mathrm{Ind}_2(\omega) 
    &=\langle \Omega^+, T^2 \Omega^+\rangle 
    = \langle T \Omega^+, T^* \Omega^+\rangle 
    = \langle \Omega^-, T^* \Pi( U )^* \Omega^-\rangle 
    = \langle \Pi( U )\Omega^+, T^* \Pi( U )^* \Omega^-\rangle \\
    &=\langle T^*\Omega^-,\Pi( U )^* T^* \Pi( U )^* \Omega^-\rangle
    =\langle \Omega^-,\Pi( U ) T \Pi( U ) T^*\Omega^-\rangle = \omega^- \big( U \tau(U) \big)
    \end{align}
    where we used repeatedly the antiunitarity of $T$.
\end{proof}

\section{Properties of the index} \label{sec:properties of the index}

\subsection{Independence of choices} \label{sec:independence}

As pointed out earlier, the construction leading to the definition of $\Ind(\omega)$ requires many choices that are not unique. We shall now show that these choices do not affect the value of the index. More importantly, we prove in this section that the index is a bonafide \emph{topological invariant}, namely that it is constant under deformations by symmetry preserving LGAs.

\subsubsection{Kramers pairing is invariant under even unitaries}\label{sec:Inv local unitary}

We first note that `being a Kramers singlet' or `belonging to a Kramers pair' is invariant under even unitaries, a property that we will use repeatedly in this section.
\begin{lemma} \label{lem:independence under even unitaries}
    Let $\psi$ be a homogeneous pure state on a fermion system with time reversal $(\caA, \tau)$ that is a Kramers singlet or belongs to a Kramers pair. Let $V \in \caA$ be an even unitary. Then $\psi' = \psi \circ \Ad[V^*]$ is a Kramers singlet if $\psi$ is, and $\psi'$ belongs to a Kramers pair if $\psi$ does.
\end{lemma}

\begin{proof}
    Let $(\Pi, \caH, \Psi)$ be the GNS triple of $\psi$. The state $\psi'$ is represented by the vector $\Psi' = \Pi(V) \Psi$. Since $\tau^2 = \theta$ we have that the unitary $T^2$ implements fermion parity. Since $V$ is even we therefore find
    \begin{equation}
        T^2 \Psi' = T^2 \Pi(V) \Psi = \Pi( \theta(V) ) T^2 \Psi = \Pi( V ) T^2 \Psi = \pm \Psi'
    \end{equation}
    by Definiton~\ref{def:Kramers pairs and Kramers singlets}
\end{proof}

\subsubsection{Independence from symmetric parent Hamiltonian}\label{Sec: Independence of insertion}

The role played by the parent Hamiltonian $H$ constructed in Section~\ref{sec:Parent H} is to allow for a definition of the flux insertion automorphism and therefore of the defect states $\omega^\pm$. Different parent Hamiltonians yield different defect states, but they all result in the same index.
\begin{proposition} \label{prop:parent}
    Let $\omega$ be a symmetric SRE state and let $H_1,H_2$ be two parent Hamiltonians for $\omega$. Then
    \begin{equation}
        \Ind(\omega, H_1, l) = \Ind(\omega, H_2, l).
    \end{equation}
\end{proposition}
\begin{proof}
    We claim that the two defect states $\omega_j^- = \omega\circ\flux_{-\pi}^{l, H_j}, j=1,2$ are unitarily equivalent through an even unitary. The argument follows closely the proof of Lemma~\ref{lem:f-close on half-planes} and Proposition~\ref{prop:almost local unitary equivalence}. Instead of considering $\flux_{-\pi}^{l, H} \circ \big( \flux_{\pi}^{l, H} \big)^{-1}$, we have here $\flux_{-\pi}^{l, H_1} \circ \big( \flux_{-\pi}^{l, H_2} \big)^{-1}$. The same argument allows one to replace $\flux_{-\phi}^{l, H_1} \circ \big( \flux_{-\phi}^{l, H_2} \big)^{-1}$ by $\al_{-\phi}^{h, H_1} \circ \big( \al_{-\phi}^{h, H_2} \big)^{-1}$ for all $\phi\in[-\pi,\pi]$ and all $A\in\Lambda\cap B_a(r)^c$, see~(\ref{beta betatilde}). We also have that
    \begin{equation}
        \omega\circ\al_{-\phi}^{h, H_1} \circ \big( \al_{-\phi}^{h, H_2} \big)^{-1}
        = \omega^{h,-\phi} \circ \big( \al_{-\phi}^{h, H_2} \big)^{-1}
        = \omega
    \end{equation}
    and hence again
    \begin{equation}\label{eq:indep of symmetric parent Hamiltonian}
        \vert \omega_0(A) - \omega_0\circ\alpha\circ\flux_{-\phi}^{l, H_1} \circ \big( \flux_{-\phi}^{l, H_2} \big)^{-1}\circ\alpha^{-1}(A)\vert\leq f(r)\Vert A\Vert
    \end{equation}
    for all $\phi\in[-\pi,\pi]$ and all $A\in\Lambda\cap B_a(r)^c$. The case $A\in\Lambda^c\cap B_a(r)^c$ follows similarly with $\flux_{-\phi}^{l, H_1} \circ \big( \flux_{-\phi}^{l, H_2} \big)^{-1}$ being close to the identity, which yields the same estimate. By Lemma~\ref{lem:f-close on cones implies f-close}, we conclude that the two states are local perturbations of each other, see Definition~\ref{def:f-close}, and hence by Proposition~\ref{prop:f-close implies almost local unitary equivalence} that there is a unitary $W_\phi$ such that $\omega_0\circ\alpha\circ\flux_{-\phi}^{l, H_1} \circ \big( \flux_{-\phi}^{l, H_2} \big)^{-1}\circ(\alpha)^{-1} = \omega_0\circ\Ad[W_{-\phi}]$. The family of homogeneous states $\omega_0\circ\alpha\circ\flux_{-\phi}^{l, H_1} \circ \big( \flux_{-\phi}^{l, H_2} \big)^{-1}\circ\alpha^{-1}$ is weak-* continuous so the $W_{-\phi}$ are all even by Lemma~\ref{lem:continuity of parity}.    
    
    Altogether, we conclude that
    \begin{equation}
        \omega_1^- = \omega_0\circ\Ad[W_{-\pi}]\circ\alpha\circ\flux_{-\pi}^{l, H_2} = \omega_2^-\circ\Ad[(\alpha\circ\flux_{-\pi}^{l, H_2})^{-1}(W_{-\pi})].
    \end{equation}
    where the unitary $V = (\alpha\circ\flux_{-\pi}^{l, H_2})^{-1}(W_{-\pi})$ is even. 
    We conclude by Lemma~\ref{lem:independence under even unitaries}.
\end{proof}

\subsubsection{Independence of the half-line}

We finally turn to the irrelevance of the choice of half-line $l$.

An argument similar to the above yields invariance under moving the endpoint of the half-line along the half-line.  If $l_1 = (a, v)$ and $l_2 = (a + bv, v)$ then $\gamma_{\phi}^{l_1, H} \circ (\gamma_{\phi}^{l_2, H})^{-1}$ is generated by $-\gamma_{\phi}^{l_2, H}(K^{l_1, H}(\phi) - K^{l_2, H}(\phi))$ which is an almost local even observable and hence $\gamma_{\phi}^{l_1, H} \circ (\gamma_{\phi}^{l_2, H})^{-1} = \Ad[V_\phi]$ for a family of even unitaries $V_{\phi}$. Invariance of the index then follows from Lemma~\ref{lem:independence under even unitaries}.

Let us now consider half-lines $l_1 = (a,v_1)$ and $l_2= (a,v_2)$ with common endpoint $a$. Now the corresponding defect states are truly distinct at infinity. However, they differ at infinity by a $U(1)$ transformation with phase $\pi$ that does not impact the flux, which explains why the index is invariant:

\begin{lemma} \label{lem:Kramers or not is invariant under U(1)}
    Let $\Gamma \subset \Z^2$. If $\psi$ is a Kramers singlet or belongs to a Kramers pair then the state $\psi' = \psi \circ \rho_{\pi}^{\Gamma}$ is a Kramers singlet if $\psi$ is, and $\psi'$ belongs to a Kramers pair if $\psi$ did.
\end{lemma}

\begin{proof}
    Let $(\Pi, \caH, \Psi)$ be the GNS triple of $\psi$ and let $T$ implement time reversal. Then $T^2 \Psi = \pm \Psi$ depending on whether $\psi$ is a Kramers singlet, or belongs to a Kramers pair. Since $\psi' = \psi \circ \rho^{\Gamma_{\pi}}$ we have that $(\Pi', \caH, \Psi)$ with $\Pi' = \Pi \circ \rho_{\pi}^{\Gamma}$ is the GNS triple of $\psi$. Using (\ref{TR of charge}) and $\rho_{-\pi}^{\Gamma} = \rho_{\pi}^{\Gamma}$ we see that time reversal is still implemented in $\Pi'$ by $T$. The claim follows immediately.
\end{proof}

\begin{proposition} \label{prop:Indep of line}
Let $\omega$ be a symmetric SRE state. Then
    \begin{equation}
        \Ind(\omega,H,l_1) = \Ind(\omega,H,l_2)
    \end{equation}
    for any two half-lines $l_1, l_2$.
\end{proposition}
\begin{proof}
     By the remarks at the beginning of this section we can restrict our attention to $l_1 = (a,v_1)$ and $l_2= (a,v_2)$. Let $\varsigma = \arccos{v_1\cdot v_2}$ be the angle between $l_1$ and $l_2$. It is sufficient to consider the case $\varsigma < \pi$, then $v = \frac{1}{2}(v_1 + v_2)$ is non-zero and we put $l = (a,v)$. We denote
     \begin{equation}
         \Lambda = \Lambda_{l,\varsigma},\qquad
         \Lambda_1 = \Lambda_{l_1,\varsigma}.
     \end{equation}
     see Figure~\ref{fig:cones 2}

    \begin{figure}[ht]
        \centering \includegraphics[width=0.25\textwidth]{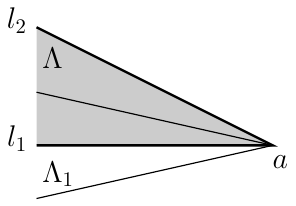}
        \caption{The $U(1)$ transformation in $\Lambda$ slides the half-line $l_1$ onto $l_2$.}
        \label{fig:cones 2}
    \end{figure}

    First of all, if $A\in \caA_{B_a(r)^c\cap\Lambda_1^c}$, then $\Vert \gamma_{-\phi}^{l_1, H} \circ \rho^\Lambda_{\phi}(A)-\rho^\Lambda_{\phi}(A)\Vert\leq f(r)\Vert A \Vert$ by Lemma~\ref{lem:Locality along line}. Hence
    \begin{equation}
        \vert \omega  \circ  \gamma_{-\phi}^{l_1, H} \circ \rho^\Lambda_{\phi}(A)-\omega\circ \rho^\Lambda_{\phi}(A)\vert \leq f(r)\Vert A \Vert.
    \end{equation}
    We now use that $\rho_\phi^\Lambda = \rho_\phi\circ\rho_{-\phi}^{\Lambda^c}$ and the $U(1)$ invariance of $\omega$ to conclude that
    \begin{equation}
        \omega\circ \rho^\Lambda_{\phi}
        = \omega\circ \rho^{\Lambda^c}_{-\phi}
        = \omega\circ \al_{-\phi}^{K^{\Lambda^c,H}}
    \end{equation}
    where $K^{\Lambda^c,H}(\phi)$ is defined as in~(\ref{K Ops}) but with the cone $\Lambda^c$ replacing the half-plane $h$. Lemma~\ref{lem:Locality along line}  yields that $\Vert \alpha_{-\phi}^{K^{\Lambda^c,H}}(A) - \gamma_{-\phi}^{l_2, H}(A)\Vert\leq f(r)\Vert A\Vert$ by the choice of $\Lambda$. Altogether, we have that
    \begin{equation}\label{outside}
    \vert \omega\circ\gamma_{-\phi}^{l_1, H} \circ \rho^\Lambda_{\phi}(A)-\omega\circ \gamma^{l_2, H}_{-\phi}(A)\vert \leq f(r)\Vert A \Vert
    \end{equation}
    for all $A \in \caA_{B_a(r)^c\cap\Lambda_1^c}$. Secondly, if $A \in \caA_{B_a(r)^c\cap\Lambda_1}$ we use that $\rho^\Lambda_{\phi}(A)$ is almost localized in the same set to conclude that $\Vert \gamma_{-\phi}^{l_1, H} \circ \rho^\Lambda_{\phi}(A) - \rho^\Lambda_{-\phi}\circ \rho^\Lambda_{\phi}(A) \Vert \leq f(r)\Vert A\Vert$ and hence
    \begin{equation}
    \vert \omega \circ \gamma_{-\phi}^{l_1, H} \circ \rho^\Lambda_{\phi}(A) - \omega\circ\rho^\Lambda_{-\phi}\circ \rho^\Lambda_{\phi}(A)\vert\leq f(r)\Vert A\Vert,
    \end{equation}
    namely $\vert \omega \circ \gamma_{-\phi}^{l_1, H} \circ \rho^\Lambda_{\phi}(A) - \omega(A)\vert\leq f(r)\Vert A\Vert$. It remains to note that $\Vert \gamma^{l_2, H}_{-\phi}(A) - A\Vert\leq f(r)\Vert A \Vert$ to conclude that
    \begin{equation}\label{inside}
    \vert \omega\circ\gamma_{-\phi}^{l_1, H} \circ \rho^\Lambda_{\phi}(A) -\omega\circ \gamma^{l_2, H}_{-\phi}(A)\vert \leq f(r)\Vert A \Vert
    \end{equation}
    for all $A\in \caA_{B_a(r)^c\cap\Lambda_1}$. The two estimates~(\ref{outside},\ref{inside}) and Lemma~\ref{lem:f-close on cones implies f-close} imply that $\omega_0 \circ \al \circ \gamma_{-\phi}^{l_1, H} \circ \rho_\phi^\Lambda \circ \big( \gamma_{-\phi}^{l_2, H} \big)^{-1} \circ  \alpha^{-1}$ is an almost local perturbation of $\omega_0$ and so Lemma~\ref{prop:f-close implies almost local unitary equivalence} yields a unitary $W_{-\phi}$ such that
    \begin{equation}
    \omega \circ \gamma^{l_1, H}_{-\phi} \circ \rho^\Lambda_{\phi} = \omega \circ \gamma^{l_2, H}_{-\phi}\circ\Ad[W_{-\phi}].
    \end{equation}
    By continuity, the unitaries are even, see Lemma~\ref{lem:continuity of parity}. The proof is concluded by putting $\phi = -\pi$ and applying Lemmas \ref{lem:Kramers or not is invariant under U(1)} and \ref{lem:independence under even unitaries}.
\end{proof}

Propositions~\ref{prop:parent} and~\ref{prop:Indep of line} allow us to make the following definition announced earlier.
\begin{definition} \label{def:index}
    Let $\omega$ be a symmetric SRE state. We define
    \begin{equation}
        \Ind(\omega) := \Ind(\omega, H, l)
    \end{equation}
    where $H$ is any symmetric parent Hamiltonian for $\omega$ and $l$ is any half-line.
\end{definition}
\noindent Recall that the existence of a symmetric parent Hamiltonian is guaranteed by Lemma \ref{lem:existence of symmetric parent Hamiltonian}.

With this, we have concluded the proof that the index of a symmetric SRE state is a well-defined property of the state. We conclude this section with two fundamental properties of the index: its multiplicativity under stacking and its `topological invariance', namely the fact that it is invariant under the action of symmetry preserving LGAs. This makes it an index of the symmetry protected phases of symmetric stably SRE states.

\subsection{Stacking}

\begin{proposition} \label{prop:stacking}
    Let $(\omega_1, \caA_1, \tau_1)$ and $(\omega_2, \caA_2, \tau_2)$ be symmetric SRE states. Then $\omega_1 \gotimes \omega_2$ is a symmetric SRE state on the fermion system with time reversal $(\caA_1 \gotimes \caA_2, \tau_1 \gotimes \tau_2)$ and
    \begin{equation}
        \Ind(\omega_1 \gotimes \omega_2) = \Ind(\omega_1) \times \Ind(\omega_2).
    \end{equation}
\end{proposition}

\begin{proof}
    Let $(\Pi_i, \caH_i, \Omega_i^+)$ be the GNS representation of $\omega_i^+$ and let $T_i$ be the antiunitary that implements time reversal in the representation $(\Pi_i, \caH_i)$. The index $\Ind(\omega_i)$ is defined by $T_i^2 \Omega_i^+ = \Ind(\omega_i) \, \Omega_i^+$, see Lemma \ref{lem:Kramers or not}.

    Let $\omega = \omega_1 \gotimes \omega_2$. Then the 0-chain $H$ with $H_x = H^{(1)}_x \gotimes \I + \I \gotimes H^{(2)}_x$ is a symmetric parent Hamiltonian for $\omega$ which we can use to construct a defect state $\omega^+ = \omega \circ \gamma_{\pi}^{l, H}$. Let $\Theta_1 \in \caB(\caH_1)$ be the unique unitary that implements parity in the representation $(\Pi_1, \caH_1)$ such that $\Theta_1 \Omega_1^+ = \Omega_1^+$. One easily checks that $\omega^+ = \omega_1^+ \gotimes \omega_2^+$, and the GNS representation of $\omega^+$ can be identified with $(\Pi, \caH, \Omega^+)$ where $\caH = \caH_1 \otimes \caH_2$, $\Omega^+ = \Omega_1^+ \otimes \Omega_2^+$, and $\Pi$ is determined by $\Pi(A_1 \gotimes A_2) = \Pi(A_1) \Theta_1^{\sigma(A_2)} \otimes \Pi_2(A_2)$ for homogeneous $A_1 \in \caA_1$ and $A_2 \in \caA_2$ and where $\sigma(A) = 0$ if $A$ is even, and $\sigma(A) = 1$ if $A$ is odd. Time reversal $\tau = \tau_1 \gotimes \tau_2$ is implemented in this representation by $T = T_1  \otimes T_2$. The index $\Ind(\omega)$ is defined by $T^2 \Omega^+ = \Ind(\omega) \Omega^+$, but
    \begin{equation}
        T^2 \Omega^+ = ( T_1^2 \Omega_1^+ ) \otimes (T_2^2 \Omega_2^+) = (\Ind(\omega_1) \Omega_1^+) \otimes (\Ind(\omega_2) \Omega_2^+) = \Ind(\omega_1) \times \Ind(\omega_2) \, \Omega^+.
    \end{equation}
    This proves the claim.
\end{proof}

\subsection{Invariance under symmetry preserving LGAs}

We now turn to showing that $\Ind(\omega)$ is a bonafide topological invariant.

In order to show invariance of the index under symmetry preserving LGAs we recall Definition~\ref{def:Kramers pairs and Kramers singlets} and first prove the following lemma:

\begin{lemma} \label{lem:Kramers or not is invariant under TR invariant automorphisms}
    Let $\psi$ be a pure state on a fermionic system with time reversal $(\caA, \tau)$ and let $\beta$ be an automorphism of $\caA$ that commutes with $\tau$. If $\psi$ is a Kramers singlet then so is $\psi \circ \beta$, and if $\psi$ belongs to a Kramers pair then so does $\psi \circ \beta$.
\end{lemma}

\begin{proof}
    Let $(\Pi, \caH, \Psi)$ be the GNS triple of $\psi$ and suppose time reversal is implemented by an antiunitary operator $T$ acting on $\caH$. The GNS triple of $\psi' = \psi \circ \beta$ is $(\Pi', \caH, \Psi)$ with $\Pi' = \Pi \circ \beta$. Since $\tau^{-1} \circ \beta \circ \tau = \beta$ we have $\Pi'(\tau(A)) = T \Pi'(A) T^*$, so time reversal is also implemented by $T$ in the representation $\Pi'$. The claim follows immediately.
\end{proof}

\begin{proposition} \label{prop:invariance under symmetry preserving LGAs}
    Let $\omega$ be a symmetric SRE state and let $G$ be a symmetry preserving TDI. Then
    \begin{equation}
        \Ind(\omega) = \Ind(\omega\circ\alpha^G_1).
    \end{equation}
\end{proposition}

\begin{proof}
    Let $G^R$ be the restriction of the TDI $G$ to the right half plane $R = \{ x \in \Z^2 \, : \, x_1 > 0 \}$ and put $G^L = G-G^R$. Then $\alpha_s^G\circ(\alpha_s^{G^R})^{-1}$ is an LGA generated by $\alpha_s^{G^R}(G^L(s))$. The TDIs $G^R_s$ and $\tilde G^L=\alpha_s^{G^R}(G^L(s))$ satisfy the assumptions of Lemma~\ref{lem:Locality along line} for the left half plane and right half plane respectively, yielding a decomposition
    \begin{equation}
        \alpha_s^G =  \al_s^{\tilde G^L}\circ \alpha_s^{G^R}.
    \end{equation}
    We will write $\al^G = \al_1^G$, $\al^L = \al_1^{\tilde G^L}$ and $\al^R = \al_1^{G^R}$ so $\al^G = \al^L \circ \al^R$. Let furthermore $\tilde \omega = \omega \circ \al^G$ and consider the state $\tilde \omega_L = \omega \circ \alpha^L$ which looks like $\omega$ far to the right and like $\tilde \omega$ for to the left of the vertical axis. We will prove the proposition by showing that the index of $\omega$ is equal to the index of $\tilde \omega_L$, and that the index of $\tilde \omega_L$ is equal to the index of $\tilde \omega$.
    
    If $H$ is a symmetric parent Hamiltonian for $\omega$ then the collection of almost local operators $(\alpha^L)^{-1}(H_x)$ for $x\in\bbZ^2$ define a symmetric parent Hamiltonian $\tilde H^L$ for $\tilde \omega_L$. The $\pi$-flux defect states for $\omega$ and $\tilde \omega_L$ may be taken to be
    \begin{equation}
        \omega^+ = \omega \circ \gamma^+, \qquad \tilde \omega_L^+ = \tilde \omega_L \circ \tilde \gamma^+
    \end{equation}
    where $\gamma^+ = \gamma_{\pi}^{l, H}$, $\tilde \gamma^+ = \gamma_{\pi}^{l, \tilde H_L}$ and $l = (0, (1, 0))$ is the right horizontal axis.
    
    Lemma~\ref{lem:Locality along line} implies that $\tilde H^L - H$ is a TDI that is almost localized on the left half plane. Denoting by $h$ the upper half-plane, it follows from this and the definition~(\ref{Def K}) that $K^{h,\tilde H^L} - K^{h,H}$ is localized along the left horizontal axis $(0, (-1, 0))$.  The LGA $\flux_{\phi}^{l, \tilde H^L}\circ(\flux_{\phi}^{l, H})^{-1}$ is generated by $\flux_{\phi}^{l, H}(K^{l,\tilde H^L} - K^{l,H})$. Since $K^{l,\cdot}$ is the restriction of $K^{h,\cdot}$ to the right half-plane, it follows by the above that $K^{l,\tilde H^L} - K^{l,H}$ is in fact an almost local even operator so there is a family of even unitaries $\tilde V_\phi$ such that $\flux_{\phi}^{l, \tilde H^L} = \Ad[\tilde V_\phi]\circ\flux_{\phi}^{l, H}$, hence
    \begin{equation}
        \flux_{\phi}^{l, \tilde H^L} = \flux_{\phi}^{l, H}\circ \Ad[V_\phi],
    \end{equation}
    where the $V_\phi = (\flux_{\phi}^{l, H})^{-1}(\tilde V_\phi)$ are again even. We therefore find
    \begin{equation} \label{eq:expression for new defect state}
        \tilde \omega_L^+ = \omega^+\circ \beta^+_L \circ \Ad[V_{\pi}]
    \end{equation}
    where $\beta^+_L = (\flux^+)^{-1}\circ \al_L \circ \flux^+$.

   In order to unburden the notation, let us write $K(s) = K^{l, H}(\pi s)$ so $\gamma^+ = \al_{1}^{K}$ and put $\widetilde K = \al_1^{\tilde G^L}(K)$. Using~(\ref{eq:LGA commutation}) we find that
    \begin{equation}
        \beta^+_L = (\al_1^{K})^{-1} \circ \al_1^{\widetilde K} \circ \al_L.
    \end{equation}
    We will show that the automorphism $(\al_1^{K})^{-1} \circ \al_1^{\widetilde K}$ is given by conjugation with an almost local even unitary. Indeed we have $\al_1^{\widetilde K} = \al_1^{\widetilde K - K + K} = \al_1^J \circ \al_1^{K}$ with $J(s) = \al_s^K( \widetilde K(s) - K(s) )$. Now,
    \begin{equation}
        \widetilde K(s) - K(s) = \iu \int_0^1 \dd s' \, \al_{s'}^{\tilde G^L} \big( [ \tilde G^L(s'), K(s)  ] \big)
    \end{equation}
    is an even almost local self-adjoint operator because $\tilde G^L$ is supported near the left half plane while $K$ is supported near the right horizontal axis. It follows that $J(s)$ is an almost local even self-adjoint operator for all $s$ and so $\al_1^J = \Ad[\widetilde W]$ for an even almost local unitary $\widetilde W$. We therefore obtain that
    \begin{equation}
        \beta_L^+ = (\al_1^K)^{-1} \circ \Ad[\widetilde W] \circ \al_1^K \circ \al_L = \al_L \circ \Ad[W]
    \end{equation}
    where $W = (\al_1^K \circ \al_L)^{-1}(\widetilde W)$ is also an almost local even unitary.
    Recalling~\eqref{eq:expression for new defect state} this yields $\tilde \omega_L^+ = \omega^+ \circ \al^L \circ \Ad[ W V_{\pi} ]$. Since $W V_{\pi}$ is an even unitary and $\al^L$ commutes with $\tau$, we conclude by Lemma~\ref{lem:independence under even unitaries} and Lemma~\ref{lem:Kramers or not is invariant under TR invariant automorphisms} that the defect state $\tilde \omega_L$ is a Kramers singlet/belongs to a Kramers pair if and only if $\omega^+$ does.

    Repeating the same argument with the replacements $\omega \rightarrow \tilde\omega_L$ and $\tilde \omega_L \rightarrow \tilde \omega = \omega \circ \alpha^G$ concludes the proof.
\end{proof}

\subsection{Extension of the index to symmetric stably SRE states}\label{sec:extension to stably SRE}

In order to extend the index to symmetric stably SRE states we must first show that the index of symmetric product states is trivial.

\begin{lemma} \label{lem:index for symmetric product states}
    Let $\omega_0$ be a symmetric product state, then $\Ind(\omega) = 1$.
\end{lemma}

\begin{proof}
    Since $\omega_0$ is a symmetric product state we can choose a symmetric parent Hamiltonian $H^{(0)}$ that consists of purely on-site terms, see the proof of Lemma~\ref{lem:existence of symmetric parent Hamiltonian}. It follows that $H^{h, \phi}_x = H^{h, 0}_x$ for all $\phi$, see~(\ref{Def K}), and in turn that the generator $K^{l ,H}_x(\phi)$ of the flux insertion automorphism vanishes for any half-line $l$, see~(\ref{K Ops}) and (\ref{Flux Insertion K}). It follows that $\omega^+ = \omega^-$ so the vectors $\Omega^+,\Omega^-$ in the GNS representation must be proportional to each other. They can therefore not form a Kramers pair and $T^2 \Omega^+ = \Omega^+$. This means by definition that $\Ind(\omega_0) = 1$. 
\end{proof}

\begin{lemma} \label{lem:index of stabilisation is independent of the stabilising state}
    Let $\omega$ be a symmetric stably SRE state and suppose $\omega'$ and $\omega''$ are symmetric product states such that $\omega \gotimes \omega'$ and $\omega \gotimes \omega''$ are symmetric SRE states. Then
    \begin{equation}
        \Ind(\omega \gotimes \omega') = \Ind(\omega \gotimes \omega'').
    \end{equation}
\end{lemma}

\begin{proof}
    Since the index of a symmetric product state is trivial (Lemma \ref{lem:index for symmetric product states}) and the index is multiplicative under stacking (Proposition \ref{prop:stacking}) we have
    \begin{equation}
        \Ind(\omega \gotimes \omega') = \Ind(\omega \gotimes \omega' \gotimes \omega'') = \Ind(\omega \gotimes \omega'' \gotimes \omega') = \Ind(\omega \gotimes \omega'')
    \end{equation}
    where we used Proposition \ref{prop:invariance under symmetry preserving LGAs} and the fact that $\omega' \gotimes \omega''$ and $\omega'' \gotimes \omega'$, being symmetric product states on isomorphic fermion systems with time reversal, are in the same SPT phase and therefore related by a symmetry preserving LGA.
\end{proof}

We can now make the following extension of the index to symmetric stably SRE states:
\begin{definition} \label{def:index for symmetric stably SRE states}
    Let $\omega$ be a symmetric stably SRE state, and let $\omega'$ be any symmetric product state such that $\omega \gotimes \omega'$ is symmetric SRE. We define
    \begin{equation}
        \Ind(\omega) := \Ind(\omega \gotimes \omega').
    \end{equation}
\end{definition}

 \noindent By Lemma \ref{lem:index of stabilisation is independent of the stabilising state}, this definition is independent of the choice of $\omega'$.

\section{Index of \texorpdfstring{$\AII$}{AII} states}\label{sec:free fermion examples}

In this section, we establish that stably SRE $\AII$ states have an index that agrees with the Fu-Kane-Mele invariant, as claimed in Theorem \ref{thm:main theorem}. We fix a single particle Hilbert space $\caK_{m} = l^2(\Z^2 ; \C^{2m})$ with time reversal $\scrT$ as described in Section \ref{sec:non interacting examples}. We consider an orthogonal projector $\scrP$ on $\caK_m$ that is exponentially local (see~\eqref{eq:locality of gapped Fermi projections}) and time reversal invariant ($\scrT \scrP \scrT^* = \scrP$). We further assume that the corresponding quasi-free state $\omega_{\scrP}$ is stably SRE, \ie we can stack with an empty product state $\omega_{\emp}$ on the CAR algebra of some $\scrK_{m'}$ to obtain an SRE state which is again quasi-free, corresponding to the exponentially local time reversal invariant projector $\scrP \oplus 0$ on $\scrK_m \oplus \scrK_{m'}$. By definition, the index of $\omega_{\scrP}$ is equal to the index of this SRE state. Since the Fu-Kane-Mele index of $\scrP$ and $\scrP \oplus 0$ are the same, we can without loss of generality assume that $\omega_{\scrP}$ is SRE to begin with.

\subsection{Flux insertion for $\AII$ states in the free fermion formalism} \label{sec:free fermion flux insertion}

The many-body index is defined in terms of an adiabatic flux insertion which we describe for the projection $\scrP$ as follows. Let $l = (0, (-1, 0))$ be the left horizontal axis so $h = h_{l}$ is the upper half plane, and let $\Pi_h$ be the orthogonal projection onto $h$, namely $\Pi_h \, | x \otimes e_i \rangle = \chi_h(x) | x \otimes e_i \rangle$ for all $x \in \Z^2$ and all $i = 1, \cdots, 2m$. The analog of the many-body $U(1)$-transformation of the upper half-plane $h$ is now
\begin{equation}\label{FF Upper plane transformation}
	\scrP^h_{\phi} := \ed^{\iu \phi \Pi_h} \scrP \ed^{-\iu \phi \Pi_h}.
\end{equation}
Since $\scrP^h_{\phi} (\partial_{\phi} \scrP^h_{\phi}) \scrP^h_{\phi} = (\scrP^h_{\phi})^{\perp} (\partial_{\phi} \scrP^h_{\phi}) (\scrP^h_{\phi})^{\perp} = 0$ we have
\begin{equation}
	\big[ [\partial_{\phi} \scrP^h_{\phi}, \scrP^h_{\phi}],  \scrP^h_{\phi} \big]  = \partial_{\phi} \scrP^h_{\phi},
\end{equation}
and therefore the flow $\phi \mapsto \scrP^h_{\phi}$ is also generated by 
\begin{equation}\label{eq:Kato generator of U(1) transformations}
	\scrK^h_{\phi} := - \iu [\partial_{\phi} \scrP^h_{\phi}, \scrP^h_{\phi}].
\end{equation}
If we let $\scrU^h_{\phi}$ be the unitary solution to
\begin{equation}
	\scrU^h_0 = \I, \qquad
 \partial_{\phi} \scrU^h_{\phi} = \iu \scrK^h_{\phi} \scrU^h_{\phi},
\end{equation}
then $\scrP^h_{\phi} = \scrU^h_{\phi} \scrP (\scrU^h_{\phi})^*$. Let now $\Pi_{L}$ be the projection on the left half-plane. Adiabatic flux insertion in the left horizontal half-line gauge is generated by the restriction
\begin{equation}\label{eq:Kato generator for flux insertion}
	\scrK_{\phi} := \Pi_{L} \scrK^h_{\phi} \Pi_L,
\end{equation}
which we call the generator for quasi-adiabatic flux insertion. Namely, if  $\scrU^{\qa}_{\phi_2, \phi_1}$ is the solution of
\begin{equation}
	\scrU^{\qa}_{\phi_1, \phi_1} = \I, \qquad \iu \partial_{\phi_2} \scrU^{\qa}_{\phi_2, \phi_1} = \scrK_{\phi_2} \scrU^{\qa}_{\phi_2, \phi_1}
\end{equation}
then the projection
\begin{equation}\label{eq:projection with phi flux}
	\scrQ^{\qa}_{\phi} := \scrU^{\qa}_{\phi, 0} \scrP \, (\scrU^{\qa}_{\phi, 0})^*
\end{equation}
corresponds to a state obtained from $\scrP$ by inserting a $\phi$ flux at the origin. We are now interested in whether the quasi-free state $\omega_{\scrQ^{\qa}_{\pi}}$ belongs to a Kramers pair or not.

\subsection{Does the state \texorpdfstring{$\omega_{\scrQ^{\qa}_{\pi}}$}{} belong to a Kramers pair?} \label{sec:quasi-free flux state is Kramers?}

For any $\mu \in \R$ we denote by $\scrP_{\phi}^{(\mu)} = \chi_{(-\infty, \mu]}(\scrH_{\phi})$ the Fermi projections for the Hamiltonians $\scrH_{\phi}$ introduced in \eqref{eq:free fermion flux Hamiltonians}.

The crucial step which will allow us to connect the time reversal behaviour of $\scrQ^{\qa}_{\pi}$ to the Fu-Kane-Mele index is expressing the spectral flow~\eqref{eq:FKM is spectral flow} defining $\FKM(\scrP)$ as the index of the pair of projections $(\scrQ^{\qa}_{\pi},\scrP^{(\mu)}_{\pi})$.
\begin{proposition} \label{prop:adiabatic vs spectral}
    Let $\mu \in \Delta$ be such that all eigenvalues crossings of $\phi \mapsto \scrH_{\phi}$ across $\mu$ are simple. Then the difference $\scrQ^{\qa}_{\phi} - \scrP^{(\mu)}_{\phi}$ is trace class and
    \begin{equation}
        \FKM(\scrP) = \Tr \lbrace \scrQ^{\qa}_{\pi} - \scrP^{(\mu)}_{\pi} \rbrace \, \mod 2.
    \end{equation}
\end{proposition}

\begin{proof}
    Since all eigenvalue crossing across $\mu$ are assumed simple, there is a finite number of values $0 < \phi_1 < \cdots < \phi_n < \pi$ at which crossings occur, and near any one of these values there is an eigenvalue branch locally described by a function $E_i(\phi)$ that is smooth and such that $E_i(\phi_i) = \mu$ while $\abs{\partial_{\phi} E_i(\phi)} > 0$. This implies in particular that we can find $\ep_i, \delta_i > 0$ such that for all $\phi \in [\phi_i - \ep_i, \phi_i + \ep_i]$ the spectrum of $\scrH_{\phi}$ in the interval $(E_i(\phi_i) - \delta_i, E_i(\phi_i) + \delta_i)$ consists solely and entirely of the eigenvalue $E_i(\phi)$, see Figure \ref{fig:crossing}.

    \begin{figure}[ht]
    \centering \includegraphics[width=0.35\textwidth]{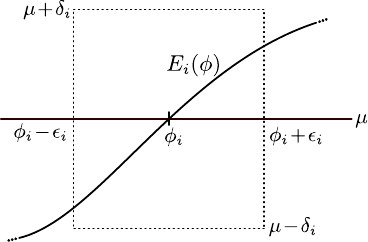}
    \caption{An eigenvalue branch $E_i(\phi)$ of the family $\phi \mapsto \scrH_{\phi}$ crossing the Fermi energy $\mu$.}
    \label{fig:crossing}
    \end{figure}

    For any $i = 1, \cdots, n$ we then have
    \begin{equation}
        \scrP^{(\mu + \sigma_i \delta_i)}_{\phi_{i} + \ep_i} - \scrP^{(\mu)}_{\phi_{i} + \ep_i} = \sigma_i \scrR_{i}
    \end{equation}
    where $\scrR_i = \chi_{(\mu - \delta_i, \mu + \delta_i)}(\scrH_{\phi_i + \ep_i})$ is a finite rank projector and $\sigma_i = +1$ if $\partial_{\phi} E_i(\phi_i) > 0$ while $\sigma_i = -1$ if $\partial_{\phi} E_i(\phi_i) < 0$. Since there are no crossings of the line $\mu + \sigma_i \delta_i$ in the interval $[\phi_i - \ep_i, \phi_i + \ep_i]$ (see again Figure~\ref{fig:crossing}), we have by Proposition~\ref{prop:spectral - qa is trace class} that
    \begin{equation}
        \scrA_i := \scrU^{\qa}_{\phi_i + \ep_i, \phi_i - \ep_i} \, \scrP^{(\mu + \sigma_i \delta_i)}_{\phi_i - \ep_i} \, (\scrU^{\qa}_{\phi_i + \ep_i, \phi_i - \ep_i})^* - \scrP^{(\mu + \sigma_i \delta_i)}_{\phi_i + \ep_i}
    \end{equation}
    is trace class and has vanishing trace. Since $\scrP_{\phi_i - \ep_i}^{(\mu + \sigma_i \delta_i)} = \scrP_{\phi_i - \ep_i}^{(\mu)}$ we therefore find that 
    \begin{equation} \label{eq:step across the crossing}
        \scrU^{\qa}_{\phi_i + \ep_i, \phi_i - \ep_i} \, \scrP^{(\mu)}_{\phi_i - \ep_i} \, (\scrU^{\qa}_{\phi_i + \ep_i, \phi_i - \ep_i})^* - \scrP^{(\mu)}_{\phi_{i} + \ep_i} = \sigma_i \scrR_i + \scrA_i.
    \end{equation}

    Similarly, for $i = 1,\ldots, n-1$ there are no eigenvalue crossings of $\scrH_{\phi}$ across $\mu$ in the interval $[\phi_i + \ep_i, \phi_{i+1} - \ep_{i+1}]$ so Proposition~\ref{prop:spectral - qa is trace class} implies that
    \begin{equation} \label{eq:step from one crossing to the next}
        \scrB_i := \scrU^{\qa}_{\phi_{i+1} - \ep_{i+1}, \phi_i + \ep_i} \, \scrP^{(\mu)}_{\phi_i + \ep_i} \, (\scrU^{\qa}_{\phi_{i+1} - \ep_{i+1}, \phi_i + \ep_i})^* - \scrP^{(\mu)}_{\phi_{i+1} - \ep_{i+1}}
    \end{equation}
    is trace class and has vanishing trace.

    Finally, again by Proposition~\ref{prop:spectral - qa is trace class} and the fact that there are no level crossings of $\scrH_{\phi}$ across $\mu$ in the intervals $[0, \phi_1 - \ep_1]$ and $[\phi_n + \ep_n, \pi]$, we have that
    \begin{align}
        \scrB_0 &:= \scrU^{\qa}_{\phi_{1} - \ep_{1}, 0} \, \scrP \, (\scrU^{\qa}_{\phi_{1} - \ep_{1}, 0})^* - \scrP^{(\mu)}_{\phi_{1} - \ep_{1}} \\
        \scrB_{n+1} &:= \scrU^{\qa}_{\pi, \phi_n + \ep_n} \, \scrP^{(\mu)}_{\phi_n + \ep_n} \, (\scrU^{\qa}_{\pi, \phi_n + \ep_n})^* - \scrP^{(\mu)}_{\pi}
    \end{align}
    are trace class and have vanishing traces. This together with a repeated application of \eqref{eq:step across the crossing} and \eqref{eq:step from one crossing to the next} yields
    \begin{equation}
        \Tr \left\lbrace \scrQ^{\qa}_{\pi} - \scrP^{(\mu)}_{\pi} \right\rbrace = \sum_{i = 1}^n  \sigma_i \Tr R_i = \SF_{\mu}([0, \pi] \ni \phi \mapsto \scrH_{\phi})
    \end{equation}
    where we use cyclicity of the trace and the fact that $\scrU^{\qa}_{\phi'', \phi'} \scrU^{\qa}_{\phi', \phi} = \scrU^{\qa}_{\phi'', \phi}$ for any $\phi, \phi', \phi'' \in \R$. The claim now follows from~\eqref{eq:FKM is spectral flow}.
\end{proof}

\begin{proposition} \label{prop:connection between Kramers and spectral flow}
    The state $\omega_{\scrQ^{\qa}_{\pi}}$ is a Kramers singlet if $\FKM(\scrP) = 0$ while it belongs to a Kramers pair if $\FKM(\scrP) = 1$.
\end{proposition}

\begin{proof}
    We recall that the index of a pair of projections $\scrP_1, \scrP_2$ acting on the same Hilbert space and such that $\scrP_1 - \scrP_2$ is compact is given by
    \begin{equation}
        \ind(\scrP_1, \scrP_2) = \mathrm{dim}\Ker(\scrP_1 - \scrP_2 - \I) - \mathrm{dim}\Ker( \scrP_1 - \scrP_2 + \I),
    \end{equation}
    see~\cite{AvronSeilerSimon}. 
    Since the difference $\scrQ^{\qa}_{\pi} - \scrP_{\pi}^{(\mu)}$ is trace class by Proposition \ref{prop:adiabatic vs spectral} the index of the pair $(\scrQ^{\qa}_{\pi}, \scrP_{\pi}^{(\mu)})$ is well-defined and given by the difference of the ranks of the orthogonal projectors $\scrN_{\pm}$ onto the finite dimensional kernels of $\scrQ^{\qa}_{\pi} - \scrP_{\pi}^{(\mu)} \mp \I$. We claim that $\scrN_+$ is a subprojector of $\scrQ^{\qa}_{\pi}$ while $\scrN_-$ is orthogonal to $ \scrQ^{\qa}_{\pi} $. First, $0\leq \scrN_\pm \scrQ^{\qa}_{\pi} \scrN_\pm \leq \scrN_\pm$ follows from $0\leq \scrQ^{\qa}_{\pi} \leq \I$ and the fact that $\scrN_\pm$ is an orthogonal projector. Second, the definition of $\scrN_\pm$ is equivalent to $\scrN_\pm ( \scrQ^{\qa}_{\pi} - \scrP_{\pi}^{(\mu)} ) \scrN_\pm = \pm \scrN_\pm$, which implies that
    \begin{equation}\label{eq:NQP}
        \scrN_\pm \scrQ^{\qa}_{\pi} \scrN_\pm = \scrN_\pm (\pm \I + \scrP_{\pi}^{(\mu)}  ) \scrN_\pm .
    \end{equation}
    Hence $\scrN_+ \scrQ^{\qa}_{\pi} \scrN_+\geq \scrN_+$. We conclude that $\scrN_+ \scrQ^{\qa}_{\pi} = \scrQ^{\qa}_{\pi} \scrN_+ = \scrN_+$ and in turn that $\scrP_{\pi}^{(\mu)} \scrN_+ = \scrN_+ \scrP_{\pi}^{(\mu)} = 0$. The second alternative of~(\ref{eq:NQP}) yields $\scrN_- \scrQ^{\qa}_{\pi} = \scrQ^{\qa}_{\pi} \scrN_- =0$. Since, moreover, $\scrN_\pm$ are orthogonal to each other, we conclude that $\widetilde \scrQ^{\qa}_{\pi} := \scrQ^{\qa}_{\pi} - \scrN_+ + \scrN_-$ is an orthogonal projector. Hence,
    \begin{equation}
        \ind(\widetilde \scrQ^{\qa}_{\pi}, \scrP_{\pi}^{(\mu)}) = \Tr \lbrace \widetilde \scrQ^{\qa}_{\pi} - \scrP_{\pi}^{(\mu)} \rbrace = \Tr \lbrace \scrQ^{\qa}_{\pi} - \scrP_{\pi}^{(\mu)} \rbrace - \Tr \lbrace \scrN_+ - \scrN_- \rbrace = 0,
    \end{equation}
    where we used the fact that the index of a pair of projections whose difference is trace class is given by the trace of their difference (\cf Proposition 2.2 of \cite{AvronSeilerSimon_Charge})). It now follows from the argument in Section 4.2.1 of \cite{bols2024absolutely} (see also Section VII B of \cite{cedzich2018topological}) that there exists a unitary $\scrV$ such that $\scrV - \I$ is trace class and $\widetilde \scrQ^{\qa}_{\pi} = \scrV \scrP_{\pi}^{(\mu)} \scrV^*$. We therefore have
    \begin{equation} \label{eq:adiabatic is local perturbation of spectral}
        \scrQ^{\qa}_{\pi} = \scrV \scrP_{\pi}^{(\mu)} \scrV^* - \scrN_+ + \scrN_-.
    \end{equation}
    Since $\scrH_{-\pi} = \scrH_{\pi}$ we have $\scrT \scrP^{(\mu)}_{\pi} \scrT^* = \scrP_{-\pi}^{(\mu)} = \scrP_{\pi}^{(\mu)}$ is time reversal invariant and therefore so is the quasi-free state $\omega_{\scrP_{\pi}^{(\mu)}}$. By Lemma \ref{lem:antiunitary implementation} there is an antiunitary $T$ that implements time reversal in the GNS representation $(\Pi_{\pi}, \caH_{\pi}, \Omega )$ of $\omega_{\scrP_{\pi}^{(\mu)}}$ which we can take such that $T \Omega =  \Omega$. It follows from~\eqref{eq:adiabatic is local perturbation of spectral} that the quasi-free state $\omega_{\scrQ^{\qa}_{\pi}}$ is represented in $\caH_{\pi}$ by the vector
    \begin{equation}
     \Omega^+ = A B \Gamma(\scrV) \Omega
    \end{equation}
    where $\Gamma(\scrV)$ is the unitary which implements the Bogoliubov transformation $a(f) \mapsto a(\scrV f)$ (existence of $\Gamma(\scrV)$ is guaranteed by the Shale-Stinespring criterion~\cite{ShaleStinespring}, see also \cite[Theorem 6.16]{evans1998quantum}) and
    \begin{equation}
        A = \Pi_{\pi} \left(  \prod_{i = 1}^{\dim \scrN_-} \, a^*(u_i) \right), \quad B = \Pi_{\pi} \left( \prod_{i = 1}^{\dim \scrN_+} \, a(v_i) \right)
    \end{equation}
    where $\{u_i\}$ and $\{v_i\}$ are orthonormal bases of the range of $\scrN_-$ and $\scrN_+$ respectively. Since $\Gamma(\scrV)$ is even (it preserves particle number) we find
    \begin{align}
        T^2 \Omega^+  &= T^* A B \Gamma(\scrV) \Omega  = (-1)^{\dim \scrN_+ - \dim \scrN_-} \, \Omega^+ \\ 
        &= (-1)^{\Tr \lbrace \scrQ^{\qa}_{\pi} - \scrP_{\pi}^{(\mu)} \rbrace} \,  \Omega^+  = (-1)^{\FKM(\scrP)} \, \Omega^+ 
    \end{align}
    where we used Proposition \ref{prop:adiabatic vs spectral} in the last equality. This concludes the proof.
\end{proof}

\subsection{The index extends the Fu-Kane-Mele index for $\AII$ states}

The state $\psi^+:=\omega_{\scrQ^{\qa}_{\pi}}$ was obtained from $\omega := \omega_{\scrP}$ by inserting a $\pi$ flux defect using a quasi-adiabatic generator~(\ref{eq:Kato generator for flux insertion}) in the single particle Hilbert space. Since we defined our many-body index using defect states obtained through Hasting's quasi-adiabatic evolution, it remains to show that both flux insertion protocols lead to the same index. For this, we show that states obtained using the second quantization of~(\ref{eq:Kato generator for flux insertion}) and using Hastings' evolution are unitarily equivalent for all values of the flux $\phi\in[0,\pi]$. The argument follows the same strategy used in the proof of Proposition~\ref{prop:almost local unitary equivalence}.

In order to compare the defect state $\psi^+$ obtained using the free fermion generator with a defect state $\omega^+$ obtained by Hasting's quasi-adiabatic evolution we first give a many-body description of the free fermion generator \eqref{eq:Kato generator for flux insertion}. It is shown in the Appendix, Lemma \ref{lem:basic locality estimates}, that $\scrK_{\phi}^h$ is exponentially local, and exponentially localised near the horizontal axis, namely
\begin{equation}
    \abs{\scrK^h_{\phi}(x, i; y, j)} \leq c \mathrm{e}^{-\xi ( \abs{x_1-y_1}  +  \abs{x_2} + \abs{y_2})}
\end{equation} 
for some positive constants $c$ and $\xi$. It follows immediately that the free fermion generator $\scrK_{\phi}$, see~(\ref{eq:Kato generator for flux insertion}), is exponentially localised near the left horizontal axis $l = (0, -(1, 0))$.

In particular, $\scrK_{\phi}$ satisfies the assumptions of Lemma \ref{lem:path of Fermi projections yields path of equivalent states} so that the projections $\scrQ^{\qa}_{\phi}$ defined in~(\ref{eq:projection with phi flux}) correspond to quasi-free pure states $\psi^{\phi}$ such that $\psi^+ = \psi^{\pi}$ and there is a TDI $K^l$ such that
\begin{equation}
    \psi^{\phi} = \omega \circ \al_{\phi}^{K^l}
\end{equation}
for all $\phi \in [-\pi, \pi]$. The TDI may be taken to be (see the proof of Lemma \ref{lem:path of Fermi projections yields path of equivalent states})
\begin{equation} \label{eq:second quantised Kato}
    K^l_x(\phi) = \frac{1}{2} \sum_{y \in \Z^2} \sum_{i, j = 1}^n \, \left(  \scrK_{\phi}(x, i ; y, j) a_{x, i}^* a_{y, j} + \overline{ \scrK_{\phi}(x, i; y, j) } a_{y, j}^* a_{x, i}  \right).
\end{equation}

Let $H$ be any symmetric parent Hamiltonian for $\omega$ (which exists by Lemma \ref{lem:existence of symmetric parent Hamiltonian}) and denote by $\omega^{\phi} = \omega \circ \gamma_{\phi}^{l, H}$ the corresponding defect states obtained by quasi-adiabatic evolution described in Section \ref{sec:defect states}.

\begin{lemma} \label{lem:from Kato to Hastings}
    There is an $f \in \caF$ and for each $\phi \in [0, \pi]$ an $f$-localized even unitary $V_{\phi} \in \caA^{\aloc}$ such that $\omega^{\phi} = \psi^{\phi}\circ\Ad[V_\phi]$.
\end{lemma}

\begin{proof}
    Since the argument is essentially the same as the ones proving Lemma \ref{lem:f-close on half-planes} and Proposition \ref{prop:almost local unitary equivalence}, we allow ourselves to gloss over some of the technical details.
    
    Let $l$ be the left horizontal axis and $h = h_{l}$ be the upper half plane. The projections $\scrP_{\phi}^h$ correspond to the states $\omega^{h, \phi} := \omega \circ \rho_{\phi}^h$. By quasi-adiabatic evolution, we have that $\omega^{h, \phi} := \omega \circ \alpha_{\phi}^{h,H}$, see~(\ref{parallel transport}).

    Similarly to~\eqref{eq:second quantised Kato} we can introduce a TDI $K^h$ given by
    \begin{equation}\label{eq:second quantised Kato on plane}
        K^h_x(\phi) = \frac{1}{2} \sum_{y \in \Z^2} \sum_{i, j = 1}^n \, \left(  \scrK^h_{\phi}(x, i ; y, j) a_{x, i}^* a_{y, j} + \overline{ \scrK^h_{\phi}(x, i; y, j) } a_{y, j}^* a_{x, i}  \right),
    \end{equation}
    which is so that
    \begin{equation}
        \omega^{h, \phi} = \omega \circ \al_{\phi}^{K^h}
    \end{equation}
    for all $\phi \in \R$.

    Since $\omega$ is SRE there is a homogeneous product state $\omega_0$ and an LGA $\al$ such that $\omega = \omega_0 \circ \al$. For all $\phi \in \R$ we define
    \begin{equation}
        \tilde \psi^{\phi} = \psi^{\phi} \circ (\al \circ \al_{\phi}^{K^l} )^{-1}, \quad \text{and} \quad
        \tilde \omega^{\phi} = \omega^{\phi} \circ (\al \circ \al_{\phi}^{K^l})^{-1}.
    \end{equation}
    Then
    \begin{equation}
        \tilde \psi^{\phi} = \omega_0, \quad \text{and} \quad \tilde \omega^{\phi} = \omega_0 \circ \tilde \beta^{\phi}
    \end{equation}
    with
    \begin{equation}
        \tilde \beta^{\phi} = \al \circ \gamma_{\phi}^{l, H} \circ ( \al_{\phi}^{K^l} )^{-1} \circ \al^{-1}.
    \end{equation}
    Let
    \begin{equation}
        \beta^{\phi} = \al \circ \al_{\phi}^{h, H} \circ ( \al_{\phi}^{K^h} )^{-1} \circ \al^{-1}.
    \end{equation}
We now argue by locality again. Because $K^l$ is a restriction of $K^h$ to the left half plane, see~(\ref{eq:Kato generator for flux insertion}) and~(\ref{eq:second quantised Kato},\ref{eq:second quantised Kato on plane}), the LGAs $\al_{\phi}^{K^l}$ and $\al_{\phi}^{K^h}$ are almost equal on any observable localized in the left half-plane $A \in \Lambda_{\pi, l} \cap B(r)^c$, see Lemma~\ref{lem:Locality along line}. The same holds for $\gamma_{\phi}^{l, H}$ and $\al_{\phi}^{h, H}$ with the former being generated by the restriction of the generator of the latter to the left half-plane. We conclude that there is $g_1 \in \caF$ such that
    \begin{equation}
        \norm{ \tilde \beta^{\phi}(A) - \beta^{\phi}(A) } \leq g_1(r) \norm{A}
    \end{equation}
    for all $\phi \in [0, \pi]$ and all $A \in \Lambda_{\pi, l} \cap B(r)^c$. Since $\tilde \psi^{\phi}= \omega_0$, we have $\tilde \psi^{\phi} \circ \beta^{\phi} = \tilde \psi^{\phi}$. It follows that
    \begin{equation}
        \abs{ \tilde \psi^{\phi}(A) -  \tilde \omega^{\phi}(A)} = \abs{ \omega_0 \big( \beta^{\phi}(A) - \tilde \beta^{\phi}(A)  \big)} \leq g_1(r) \norm{A}
    \end{equation}
    for all $A \in \Lambda_{\pi, l} \cap B(r)^c$. What is more, both $\gamma_{\phi}^{l, H}$ and $\al_{\phi}^{K^l}$ act almost trivially on any observable localized in the right half-plane $A \in \Lambda_{\pi, l} \cap B(r)^c$, again by Lemma~\ref{lem:Locality along line} and the fact that the generators are almost localised on the left half-plane. In turn, this implies that there is $g_2\in\caF$ such that 
    \begin{equation}
        \norm{\tilde \beta(A) - A} \leq g_2(r) \norm{A}
    \end{equation}
for all $\phi \in [0, \pi]$ and all $A \in \Lambda_{\pi, -l} \cap B(r)^c$. This implies that
    \begin{equation}
        \abs{ \tilde \psi^{\phi}(A) -  \tilde \omega^{\phi}(A)} = \abs{ \omega_0 \big( A - \tilde \beta^{\phi}(A)  \big)} \leq g_2(r) \norm{A}
    \end{equation}
    for all $A \in \Lambda_{\pi, -l} \cap B(r)^c$.
    
    Since $\tilde \psi^{\phi} = \omega_0$ is a homogeneous product state it now follows from Lemma \ref{lem:f-close on cones implies f-close} that the states $\tilde \omega^{\phi}$ are all $g$-close to $\omega_0$ for $\phi \in [0, \pi]$ for a $g \in \caF$ that depends only on $g_1,g_2$, and by Proposition \ref{prop:f-close implies almost local unitary equivalence} there is an $\tilde f\in\caF$ depending only on $g$, and $\tilde f$-local homogeneous unitaries $W_{\phi} \in \caA^{\aloc}$ such that $\tilde \omega^{\phi} = \omega_0 \circ \Ad[W_{\phi}]$.  It follows that $\omega^{\phi} = \psi^{\phi} \circ \Ad[V_{\phi}]$ with $V_{\phi} = (\al \circ \al_{\phi}^{K})^{-1} (W_{\phi})$, which are uniformly $f$-local for some $f \in \caF$.

    The family of states $\tilde \omega^{\phi}$ are related to the pure homogeneous product state $\omega_0$ by conjugation with the homogeneous almost local unitaries unitaries $W_{\phi}$ which satisfy the assumptions of Lemma \ref{lem:continuity of parity}. Since $\tilde \omega^{0} = \omega_0$ we must have that $W_{0}$ is even and so all the $W_{\phi}$ are even. Since $\al$ and $\al_{\phi}^K$ are parity preserving, we conclude that also all the $V_{\phi}$ are even.
\end{proof}

\begin{proposition} \label{prop:index of AII states extends FKM}
    We have $\Ind(\omega_{\scrP}) = \FKM(\scrP)$.
\end{proposition}

\begin{proof}
    We must show that the $\pi$-flux state $\omega^+$ obtained from $\omega_{\scrP}$ by quasi-adiabatic flux insertion is a Kramers singlet if $\FKM(\scrP) = 0$, and belongs to a Kramers pair of $\FKM(\scrP) = 1$ (see Definitions \ref{def:index of triple} and~\ref{def:index for symmetric stably SRE states}). From Proposition \ref{prop:connection between Kramers and spectral flow} we known this to be true for the $\pi$-flux state $\psi^+ = \omega_{\scrQ^{\qa}_{\pi}}$ obtained from $\omega_{\scrP}$ using a free fermion generator. By Lemma \ref{lem:from Kato to Hastings} there is an even unitary $V \in \caA^{\aloc}$ such that $\psi^+ = \omega^+ \circ \Ad[V]$. We conclude by Lemma \ref{lem:independence under even unitaries}
\end{proof}

\section{Proof of the main theorem}\label{sec:proof of the main theorem}

We collect the results obtained in the previous sections to give a proof of Theorem \ref{thm:main theorem}.

\begin{proof}
    \textbf{Proof of (\ref{thm:item 1}).} Triviality of the index for symmetric product states is Lemma \ref{lem:index for symmetric product states}. The claim about the $\AII$ states is Proposition~\ref{prop:index of AII states extends FKM}.
    
    \textbf{Proof of (\ref{thm:item 2}).} Multiplicativity of the index under stacking is the statement of Proposition~\ref{prop:stacking}.
    
    \textbf{Proof of (\ref{thm:item 3}).} By items (\ref{thm:item 1}) and (\ref{thm:item 2}), the index is invariant under stacking with a product state. With this, the claim (\ref{thm:item 3}) follows from the invariance under symmetry preserving LGAs, Proposition~\ref{prop:invariance under symmetry preserving LGAs}.

\end{proof}

\noindent \textbf{Acknowledgments.} The work of S.B. and M.R. is supported by NSERC of Canada.

\noindent \textbf{Data availability.} There is no data associated with this work.

\appendix

\section{Almost local transitivity} \label{app:almost local transitivity}

The classical Kadison transitivity theorem implies that two GNS equivalent representations associated with two pure states are necessarily unitarily equivalent through a unitary element of the algebra, see e.g.~\cite{Dixmier}. Our goal in this section is two-fold. First to prove that if $\psi_0$ is a product state and $\psi_0,\psi$ are almost local perturbations of each other when restricted to a cone and its complement, then $\psi_0,\psi$ are almost local perturbations of each other, see Definition~\ref{def:f-close}. Second of all to show that if they are moreover pure, then they are unitarily equivalent through a unitary in the almost local algebra.

\subsection{Lifting closeness at infinity on cones to closeness at infinity}

We start with the first goal, a version of which can also be found in Appendix~C of~\cite{kapustin2020hall}. It is a rather classical argument using relative entropy, except for the fact that we need Araki's extension~\cite{araki2003equilibrium} for the subadditivity of the quantum entropy. Recall the definition of a cone given in Section~\ref{sec:locality}, and Definition \ref{def:f-close} for two states to be almost local perturbations of each other.
\begin{lemma} \label{lem:f-close on cones implies f-close}
    Let $\Lambda$ be a cone. Let $f \in \caF$ and suppose $\psi_0$ is a homogeneous product state and $\psi$ a homogeneous state on $\caA$ such that
    \begin{equation} \label{eq:f-close on cone}
        \abs{ \psi_0(A) - \psi(A) } \leq f(r) \norm{A} \quad \forall A \in \caA^{\loc}_{\Lambda\cap B_a(r)^c}
    \end{equation}
    and
    \begin{equation} \label{eq:f-close on complement cone}
        \abs{\psi_0(A) - \psi(A)} \leq f(r) \norm{A} \quad \forall A \in \caA^{\loc}_{\Lambda^c\cap B_a(r)^c}.
    \end{equation}
    Then $\psi$ is an almost local perturbation of $\psi$.
\end{lemma}

\begin{proof}

    For any $x\in\Z^2$,  we denote by $\rho^{(x)}$ and $\rho_{0}^{(x)}$ the density matrices of the restrictions to $\caA_x$ of $\psi$ and $\psi_0$ respectively. By assumption we have for any $A \in \caA_x$ that $\abs{\psi(A) - \psi_0(A)} \leq f(\vert x\vert ) \norm{A}$, which implies $\norm{\rho^{(x)} - \rho_{0}^{(x)}}_1 \leq f( \vert x\vert )$. Let $S^{(x)} = - \Tr (\rho^{(x)} \log \rho^{(x)})$ be the von Neumann entropy of $\rho^{(x)}$. Since $\psi_0$ is pure, its restriction is pure as well by definition and so $S^{(x)}_0 = 0$. If $\vert x\vert $ is large enough so that $f( \vert x\vert ) < 1/e$ then Fannes' inequality yields
    \begin{equation}\label{eq:fannes}
        S^{(x)}\leq f( \vert x\vert ) \log( 2^n / f(\vert x\vert) ),
    \end{equation}
    where $2^n$ is the dimension of the local Fock space for $\caA_x$. Since $f$ decays faster than any polynomial we have that
    \begin{equation}
        S := \sum_{x \in \Lambda} S^{(x)} < \infty.
    \end{equation}
    By subadditivity of the von Neumann entropy (see Theorem 3.7 of \cite{araki2003equilibrium} for subadditivity for fermions), the quantity $S$ is an upper bound for the entanglement entropy of any finite $\Gamma \Subset \Z^2$ in the state $\psi$.

    We now show that $\psi_0$ and $\psi$ are $g$-close for some $g \in \caF$. Let $\Gamma \Subset \overline{B(r)^c}$ be a finite region outside the ball $B(r)$ and denote by $\rho^{(\Gamma)}$ and $\rho_0^{(\Gamma)}$ the density matrices of the restrictions to $\caA_{\Gamma}$ of the states $\psi$ and $\psi_0$ respectively. Similarly, let $\rho^{(\Gamma \cap \Lambda)}$ and $\rho^{(\Gamma \cap \Lambda^c)}$ be the density matrices of the restrictions of $\psi$ to $\caA_{\Gamma \cap \Lambda}$ and to $\caA_{\Gamma \cap \Lambda^c}$ respectively. Equivalently, $\rho^{(\Gamma \cap \Lambda)}$ is the restriction of $\rho^{(\Gamma)}$ to $\caA_{\Gamma \cap \Lambda}$, and similarly for $\rho^{(\Gamma \cap \Lambda^c)}$. Let
    \begin{equation}
        S(\rho^{(\Gamma)} \| \rho^{(\Gamma \cap \Lambda)} \gotimes \rho^{(\Gamma \cap \Lambda^c)}) = S(\rho^{(\Gamma)}) - S( \rho^{(\Gamma \cap \Lambda)} ) - S( \rho^{(\Gamma \cap \Lambda^c)} )
    \end{equation}
    be the relative entropy of $\rho^{(\Gamma)}$ and $\rho^{(\Gamma \cap \Lambda)} \gotimes \rho^{(\Gamma \cap \Lambda^c)}$. By non-negativity and subadditivity of the von Neumann entropy, and the bound~(\ref{eq:fannes}) on $S^{(x)}$ obtained above,
    \begin{equation}
        S(\rho^{(\Gamma)} \| \rho^{(\Gamma \cap \Lambda)} \gotimes \rho^{(\Gamma \cap \Lambda^c)}) \leq S(\rho^{(\Gamma)}) \leq \sum_{x  : \vert x\vert > r} S^{(x)} \leq \frac{h(r)}{2}
    \end{equation}
    for some $h \in \caF$ that only depends on $f$ and $n$. The quantum Pinsker inequality now yields
    \begin{equation}
        \norm{ \rho^{(\Gamma)} -  \rho^{(\Gamma \cap \Lambda)} \gotimes \rho^{(\Gamma \cap \Lambda^c)} }_1^2 \leq 2 S(\rho^{(\Gamma)} \| \rho^{(\Gamma \cap \Lambda)} \otimes \rho^{(\Gamma \cap \Lambda^c)}) \leq h(r).
    \end{equation}
    It follows from the assumptions~(\ref{eq:f-close on cone},\ref{eq:f-close on complement cone}) that
    \begin{equation}
        \norm{ \rho^{(\Gamma \cap \Lambda)} \gotimes \rho^{(\Gamma \cap \Lambda^c)} - \rho_0^{\Gamma} }_1 \leq f(r).
    \end{equation}
    Combining the two last inequalities finally yields
    \begin{equation}
        \norm{ \rho^{(\Gamma)} - \rho_0^{(\Gamma)} } \leq f(r) + h(r),
    \end{equation}
    thus proving the claim since $ f + h \in \caF$.
\end{proof}

\subsection{Closeness at infinity implies almost local unitary equivalence}

We now move to the transitivity itself. A similar result in the context of one-dimensional spin chains can be found in~\cite{kapustin2021classification}, and extended to the two-dimensional setting in~\cite{kapustin2020hall}

\begin{proposition} \label{prop:f-close implies almost local unitary equivalence}
    Let $\upsilon$ be a homogeneous pure product state. Assume that $\psi$ is a homogeneous pure state that is $f$-close to $\upsilon$. Then there is $\tilde f\in\caF$ which depends only on $f$, and a homogeneous $\tilde f$-local unitary $U \in \caA^{\aloc}$ such that $\psi = \upsilon \circ \Ad_U$. The unitary $U$ may be taken even if $\psi$ has the same parity as $\upsilon$, and odd of $\psi$ has opposite parity to $\upsilon$.
\end{proposition}

We first establish several lemmas which will be combined into a proof of the proposition at the end of Section \ref{sec:proof of transitivity} below.

\subsubsection{GNS representation for homogeneous pure product states}

Let $\upsilon = \gotimes_{x \in \Z^2} \, \upsilon_x$ be a homogeneous pure product state on $\caA$. For each $x \in \Z^2$ let $\pi_x : \caA_x \rightarrow \caB(\caH_x)$ be the defining representation of $\caA_x$. Since we assumed that each $\caA_x$ is a CAR algebra on $n$ fermion modes we have $\caH_x \simeq \C^{2^n}$. Let $\upsilon_x \in \caH_x$ be a unit vector representing the state $\upsilon_x$ in $\pi_x$ and let $\Theta_x \in \caA_x$ be the unique unitary such that $\theta(A) = \Ad_{\Theta_x}(A)$ for all $A \in \caA_x$ and $\pi_x(\Theta_x) \Psi_x = \Psi_x$. Note that $\Theta_x$ is even and $\Theta_x^2 = \I$. We call a vector $\Psi \in \caH_x$ homogeneous if it is an eigenvector of $\pi_x(\Theta_x)$, equivalently if $\Psi$ represents a homogeneous state on $\caA_x$.

Let $\{ \Upsilon_x^{(i)} \}_{i = 1}^{2^n}$ be an orthonormal basis of $\caH_x$ consisting of homogeneous vectors, and such that $\Upsilon_x^{(1)} = \Upsilon_x$.

Let $\caK$ be the linear space of sequences $\{ \Psi_x \}_{x \in \Z^2}$ such that $\Psi_x \in \caH_x$ and such that $\Psi_x \neq \Upsilon_x$ for only finitely many $x \in \Z^2$. Then the inner product
\begin{equation}
    \langle \{ \Psi_x \}, \{ \Psi'_x \} \rangle = \prod_{x \in \Z^2} \, \langle \Psi_x, \Psi'_x \rangle
\end{equation}
is well-defined for all $\{ \Psi_x \}, \{ \Psi'_x \} \in \caK$ and extends to the whole of $\caK$ by linearity. We let $\caH$ be the Hilbert space obtained by closing $\caK$ w.r.t. the norm induced by this inner product. The set of sequences $\{ \Upsilon_x^{(i_x)} \}$ with only finitely many $i_x \neq 1$ form an orthonormal basis for $\caH$.

We write $\otimes_{x \in \Z^2} \, \Psi_x := \{ \Psi_x \}$ for any sequence in $\caH_{\upsilon}$. It is clear that for any $X \Subset \Z^2$, the Hilbert space $\caH$ decomposes as a tensor product $\caH = \caH_X \otimes \caH_{X^c}$ and if $X$ is finite then $\caH_X$ is identified with $\bigotimes_{x \in X} \, \caH_x$.

We now represent $\caA$ on $\caH$ using a Jordan-Wigner construction. Fix an enumeration $x : \N \rightarrow \Z^2$ of the sites And let $X_N = \{ x(n) \,:\,  1 \leq n \leq N \}$ for all $N \in \N$. Assume moreover that the enumeration is such that $B_0(r) = X_{N_r}$ for all $r \in \N$ and some $N_r \in \N$. For any finite $X \Subset \Z^2$ and any homogeneous operator $A \in \caA_{x(N)}$ we let
\begin{equation}
    \pi(A) = \begin{cases} \I_{\caH_{X_{N-1}}} \, \otimes \,  \pi_{x(N)}(A) \, \otimes \, \I_{\caH_{X_N}^c} \quad &\text{if} \, A \, \text{is even} \\
    \left( \bigotimes_{x \in X_{N-1}} \pi_{x}( \Theta_x ) \right)  \, \otimes \, \pi_{x(N)}(A) \, \otimes \, \I_{\caH_{X_N^c}} \quad &\text{if} \, A \, \text{is odd}. \end{cases}
\end{equation}
The `tail' of grading operators $\Theta_x$ ensures that odd operators supported on disjoint regions anticommute in the representation. Since the on-site homogeneous operators generate the whole algebra, this defines the representation uniquely.

To see that this is a GNS representation for $\upsilon$, let $A_{x_1} \cdots A_{x_l}$ be a monomial of homogeneous operators $A_{x_i} \in \caA_{x_i}$, writing $\Upsilon := \otimes_{x} \, \Upsilon_x$, and use the fact that the $\Upsilon_x$ are even to find
\begin{equation}
    \langle \Upsilon, \pi( A_{x_1} \cdots A_{x_l} ) \, \Upsilon \rangle = \upsilon(A_{x_1}) \cdots \upsilon(A_{x_l}) = \upsilon( A_{x_1} \cdots A_{x_l}  ).
\end{equation}
By linearity and continuity it follows that $\Upsilon(A) = \langle \Upsilon, \pi(A) \, \Upsilon \rangle$ for any $A \in \caA$. Finally, the vector $\Upsilon \in \caH$ is cyclic for the representation $\pi$ by construction.

\subsubsection{Construction of a unitary} \label{sec:proof of transitivity}

Let $\psi$ be a homogeneous pure state on $\caA$ such that the hypotheses of Proposition \ref{prop:f-close implies almost local unitary equivalence} hold w.r.t. the product state $\upsilon$. \ie the state $\psi$ is $f$-close to $\upsilon$ for some $f \in \caF$.

For any $X \Subset \Z^2$ let $P_X$ be the projection on the state $\Upsilon_X = \otimes_{x \in X} \Upsilon_x \in \caH_X$. Note that each $P_X$ is even. Let $r ' > r$ and let $P_{r r'} = P_{B_{0}(r') \setminus B_0(r)}$. Denote by $\rho_r$ the density matrix acting on $\caH_{r} := \caH_{B_0(r)}$ corresponding to the restriction of $\psi$ to $\caA_{B_0(r)}$. Note that for all $A\in\caA_{B_0(r)}$,
\begin{equation}
    \Tr(\rho_r \pi(A)) = \psi(A) = \psi\circ\theta(A) = \Tr(\Theta_r^* \rho_r\Theta_r \pi(A)),
\end{equation}
by cyclicity of the trace, where $\Theta_r$ is the restriction of $\Theta$ to the $\caH_r$. Since the density matrix is unique, we conclude that $\rho_r$ is even. 

\begin{lemma} \label{lem:rho agrees with phi far out}
    For any $r, r' \in \N$ such that $r' > r$ we have
    \begin{equation} \label{eq:crucial bound}
        \norm{\rho_{r'} ( \I \gotimes P^\perp_{r r'})}_1^2 \leq f(r)
    \end{equation}
    where we interpret all operators as operators on $\caH_{r'}$ by tensoring with the identity if needed.
\end{lemma}

\begin{proof}
    Since $P_{r r'}^\perp$ is even we have $\I \gotimes P_{r r'}^{\perp} = \I \otimes P_{r r'}^{\perp}$. Then using H\"older's inequality we get
\begin{align}
    \norm{\rho_{r'} P_{r r'}^{\perp}}^2_{1} &\leq \norm{\sqrt{ \rho_{r'} }}_2^2 \norm{\sqrt{\rho_{r'} P_{r r'}^{\perp}}}_2^2 = \Tr \lbrace \rho_{r'} \rbrace \times \Tr \lbrace \rho_{r'} P_{r r'}^{\perp} \rbrace \\
							&= \psi(P_{r r'}^{\perp}) = (\psi - \upsilon)(P_{r r'}^{\perp}) \leq f(r)
\end{align}
by the assumption that $\psi$ and $\upsilon$ are $f$-close.
\end{proof}

Now let $Q_r = P_{B_0(r)^c}$, which is an even rank 1 projection on $\caH_{B_r^c}$. We have
\begin{lemma} \label{lem:a Cauchy sequence of density matrices}
    The sequence $\{ \rho_r \gotimes Q_r \}_{r \in \N}$ is a Cauchy sequence in trace norm.
\end{lemma}

\begin{proof}
    Note first that $Q_r = P_{r r'} \otimes Q_{r'}$ so
    \begin{equation} \label{eQ:step 1}
        \norm{\rho_{r'} \otimes Q_{r'} - \rho_r \otimes Q_r}_1 \leq \norm{\rho_{r'} - \rho_r \otimes P_{r r'}}_1.
    \end{equation}
    Let us denote by $\Tr_{r r'}$ the partial trace over $\caH_{r r'} = \caH_{B_0(r') \setminus B_0(r)}$. We bound the right-hand side of \eqref{eQ:step 1} by noting that
    \begin{align}
        \rho_r \otimes P_{r r'} &= \Tr_{r r'} \lbrace P_{r r'} \rho_{r'} P_{r r'} \rbrace \otimes P_{r r'} + \Tr_{r r'} \lbrace P_{r r'}^{\perp} \rho_{r'} P_{r r'}^{\perp} \rbrace \otimes P_{r r'} \\
                    &= P_{r r'} \rho_{r'} P_{r r'} + \Tr_{r r'} \lbrace P_{r r'}^{\perp} \rho_{r'} P_{r r'}^{\perp} \rbrace \otimes P_{r r'} \\
                    &= \rho_{r'} - P_{r r'} \rho_{r'} P_{r r'}^{\perp} - P_{r r'}^{\perp} \rho_{r'} P_{r r'} - P_{r r'}^{\perp} \rho_{r'} P_{r r'}^\perp \\
                    &\quad + \Tr_{r r'}\lbrace P_{r r'}^{\perp} \rho_{r'} P_{r r'}^{\perp} \rbrace \otimes P_{r r'}.
    \end{align}
    Here the last four terms are bounded in trace norm by $\sqrt{f(r)}$ using H\"older's inequality and  Lemma \ref{lem:rho agrees with phi far out} yielding
    \begin{equation}
        \norm{ \rho_{r'} \gotimes Q_{r'} - \rho_{r} \gotimes Q_r }_1 \leq 4 \sqrt{f(r)}.
    \end{equation}
    We conclude that the sequence $\{ \rho_{r} \gotimes Q_r \}_{r}$ is indeed Cauchy.
\end{proof}

Let us denote by $\rho \in \caB(\caH)$ the limit of the Cauchy sequence $\{ \rho_r \gotimes Q_r \}$. By construction $\rho$ represents the state $\psi$ which is homogeneous and pure, so $\rho$ is actually an even rank one projection $\rho = | \Psi \rangle \langle \Psi |$ for some homogeneous unit vector $| \Psi \rangle \in \caH$. The next step towards proving Proposition \ref{prop:f-close implies almost local unitary equivalence} is to construct a sequence of pure density matrices $\tilde \rho_r$ that converge to $\rho$, and such that $\tilde \rho_r$ agrees with product state $\Upsilon$ outside the ball $B_0(r)$.

\begin{lemma} \label{lem:sequence of pure sensity matrices}
    Let $r_0 \in \N$ be the smallest natural number such that $36 \sqrt{f(r)} < 1$ for all $r \geq r_0$. Then for all $r \geq r_0$ we have that the pure density matrices
    \begin{equation}
        \tilde \rho_r := \frac{(\I \otimes Q_r) \rho (\I \otimes Q_r)}{\Tr \lbrace (\I \otimes Q_r) \rho (\I \otimes Q_r) \rbrace }
    \end{equation}
    are well defined and 
    \begin{equation} \label{eq:crucial bound 2}
        \norm{ \rho - \tilde \rho_r }_1 \leq 9 \sqrt{ f(r) }.
    \end{equation}
\end{lemma}

\begin{proof}
First of all, 
    \begin{equation}
        \rho_{r'} P_{r r'}^{\perp}Q_{r'} - \rho Q_r^{\perp} = (\rho_{r'} Q_{r'} - \rho) + (\rho - \rho_{r'} Q_{r'}) Q_r
    \end{equation}
converges to zero as $r' \rightarrow \infty$ in trace norm by Lemma \ref{lem:a Cauchy sequence of density matrices}. ~\eqref{eq:crucial bound} then implies that
    \begin{equation} \label{eq:crucial bound 3}
        \norm{\rho ( \I \otimes Q_r^{\perp})}_1 \leq \sqrt{f(r)}.
    \end{equation}
    From
    \begin{equation}
        1 = \norm{ \rho }_1 \leq \norm{\rho (\I \otimes Q_r) }_1 + \norm{ \rho (\I \otimes Q_r^\perp) }_1
    \end{equation}
    we then obtain
    \begin{equation}
        \norm{ \rho (\I \otimes Q_r) }_1 \geq 1 - \sqrt{f(r)}.
    \end{equation}
    Using that $\rho$ is a rank one projection it follows that
    \begin{equation}
        \Tr \lbrace (\I \otimes Q_r) \rho (\I \otimes \, Q_r) \rbrace = \norm{\rho (\I \otimes \, Q_r)}_2^2 \leq \norm{\rho (\I \otimes \, Q_r)}_1^2 \leq \big(1 - \sqrt{f(r)} \big)^2.
    \end{equation}
    For $r \geq r_0$ we have $36 \sqrt{f(r)} < 1$ so certainly $\Tr \lbrace (\I \otimes Q_r) \rho (\I \otimes \, Q_r) \rbrace > 1/2$ and therefore $\tilde \rho_r$ is well-defined. (Note that $\tilde \rho_r$ is the rank one projection onto the span of $(\I \otimes \, Q_r) | \Psi \rangle$).

    To show the bound \eqref{eq:crucial bound 2} we note that~\eqref{eq:crucial bound 3} implies $\norm{\rho - Q_r \rho Q_r}_1 \leq \norm{Q_r^{\perp} \rho Q_r}_1 + \norm{Q_r \rho Q_r^\perp}_1 + \norm{Q_r^{\perp} \rho Q_r^{\perp}}_1 \leq 3 \sqrt{f(r)}$. Combining this with $\abs{ 1 - \Tr \lbrace Q_r \rho Q_r \rbrace} \leq \norm{\rho - Q_r \rho Q_r}_1$ and $\Tr \lbrace Q_r \rho Q_r \rbrace > 1/2$ for all $r \geq r_0$ we find
    \begin{equation}
        \norm{ \rho - \tilde \rho_r }_1 = \norm{\rho - Q_r \rho Q_r}_1 + \frac{1}{ \Tr \lbrace Q_r \rho Q_r \rbrace} \norm{ \big( 1 - \Tr \lbrace Q_r \rho Q_r \rbrace \big) \, Q_r \rho Q_r}_1 \leq 9 \sqrt{f(r)}
    \end{equation}
    as required.
\end{proof}

Let us denote by $\tilde \psi_r$ the pure states corresponding to the density matrices $\tilde \rho_r$. \ie $\tilde \psi_r(A) = \Tr \lbrace \tilde \rho_r \pi(A) \rbrace$ for all $A \in \caA$. Furthermore, (\ref{eq:crucial bound 2}) implies that the density matrices $\tilde \rho_r$ are all even, for $r$ large enough.
Indeed, vectors with different parity are orthogonal to each other, so the trace norm of the difference of their density matrices equals 2, but the bound is $9 \sqrt{f(r)} < 2$ for all $r$ large enough.

\begin{lemma} \label{lem:unitary steps}
    For all $r \geq r_0$ there is an even unitary $U_r \in \caA_{B_0(r+1)}$ such that $\tilde \psi_{r+1} = \tilde \psi_{r} \circ \Ad_{U_r}$ and such that
    \begin{equation} \label{eq:unitary step is small}
        \norm{\I - U_r} \leq 6 f(r)^{1/4}.
    \end{equation}
\end{lemma}

\begin{proof}
    From ~\eqref{eq:crucial bound 2} we have
    \begin{equation}
        \norm{\tilde \rho_{r+1} - \tilde \rho_r }_1 \leq \norm{\tilde \rho_{r+1} - \rho }_1 + \norm{\rho - \tilde \rho_r}_1 \leq 18 \sqrt{f(r)}.
    \end{equation}
    for all $r \geq r_0$. We can choose normalised vector representatives $|\Psi_r \rangle$ of $\tilde \rho_r$, all with the same parity and such that $a_r := \langle \Psi_r, \Psi_{r+1} \rangle > 0$ are real and positive. Note that once $\Psi_{r_0}$ is chosen, this uniquely fixes all the other $\Psi_r$. Then
    \begin{equation}
        \cos(\theta_r)^2 := \abs{ \langle \Psi_r, \Psi_{r+1} \rangle }^2 = \Tr \lbrace \tilde \rho_{r+1} \tilde \rho_r \rbrace = 1 + \Tr \lbrace (\tilde \rho_{r+1} - \tilde \rho_r) \tilde \rho_r \rbrace
    \end{equation}
    hence
    \begin{equation} \label{eq:bound on a_r}
        \cos(\theta_r)^2 \geq 1 - 18 \sqrt{f(r)}.
    \end{equation}
    We can decompose $\Psi_{r+1}$ as
    \begin{equation}
        \Psi_{r+1} = \cos(\theta_r)\Psi_r + \sin(\theta_r)\Psi_r^{\perp}
    \end{equation}
    where $\Psi_r^{\perp}$ is the normalized component of $\Psi_r$ that is perpendicular to $\Psi_{r+1}$.

By construction the states corresponding to the $\tilde \rho_r$ agree with the product state $\upsilon$ outside of $B_0(r)$, so we have $\Psi_r = \Psi'_r \otimes \Upsilon_{B_0(r)^c}$ for all $r \geq r_0$. Here $\Psi'_r \in \caH_{r}$ is a homogeneous vector and $\Upsilon_{B_0(r)^c} = \bigotimes _{x \in B_0(r)^c} \, \Upsilon_x \in \caH_{B_0(r)^c}$ is the restriction of $\Upsilon$ to the complement of $B_0(r)$. Let $V_r$ be the unitary which acts on the subspace spanned by $\Psi_r$ and $\Psi_{r+1}$ as the $\theta_r$-rotation which rotates $\Psi_{r+1}$ into $\Psi_{r}$, and acts as the identity on the complement of this subspace. Then $V_r$ is even and $V_r = \pi(U_r)$ for a unique even unitary $U_r \in \caA_{B_0(r+1)}$. Since $\pi(U_r) \Psi_{r+1} = \Psi_{r}$ we have $\tilde \psi_{r+1} = \tilde \psi_{r} \circ \Ad_{U_r}$ as required. With~\eqref{eq:bound on a_r} we find that
    \begin{equation}
        \norm{\I - U_r}^2 = \norm{\I - V_r}^2 \leq \theta_r^2 \leq \frac{  18  \sqrt{f(r)} }{1 - 18 \sqrt{f}} \leq 36 \sqrt{f(r)}
    \end{equation}
    which establishes the bound \eqref{eq:unitary step is small}.
\end{proof}

Let us finally combine the unitary steps provided by Lemma \ref{lem:unitary steps} into a single unitary, providing a proof of the main proposition.

\begin{proof} [Proof of Proposition\ref{prop:f-close implies almost local unitary equivalence}]
    Let $g_r := 6 f(r)^{1/4}$ so the unitaries $U_r \in \caA_{B_0(r+1)}$ provided by Lemma \ref{lem:unitary steps} are even and satisfy $\norm{U_r - \I} \leq g_r$. Since $f \in \caF$ the sequence $\{g_r\}_{r \geq r_0}$ is summable. For $R \geq r_0$, define $U^{(R)} := \prod_{r = r_0}^R U_r$, which is even. Then for $R' > R \geq r_0$ we have
    \begin{equation}
        \norm{ U^{(R')} - U^{(R)} } 
        \leq \sum_{r = R+1}^{R'} \, \norm{U^{(r)} - U^{(r-1)}}
        =  \sum_{r = R+1}^{R'} \, g_r.
    \end{equation}
    Since $\{g_r\}_{r \geq r_0}$ is summable, this shows that $\{ U^{(R)} \}_{R \geq r_0}$ is Cauchy and converges to a unitary $U^{(\infty)}$ which satisfies
    \begin{equation}
        \tilde \psi_{r_0} \circ \Ad_{U^{(\infty)}} = \lim_{R \uparrow \infty} \,\, \tilde \psi_{r_0} \circ \Ad_{U^{(R)}} = \lim_{R \uparrow \infty} \,\, \tilde \psi_{R} = \psi.
    \end{equation}
    Note moreover that $U^{(\infty)}$ is almost local because $\norm{U^{(\infty)} - U^{(R)}} \leq G(R)$ with $G(R) = \sum_{r = R}^{\infty} g_r$, which defines a rapidly decaying function in $\caF$. It is also even as the norm limit of even operators. 

    Finally, since $\tilde \psi_{r_0}$ and $\upsilon$ are both homogeneous and differ only in the region $B_{0}(r_0)$, there is a homogeneous unitary $V \in \caA_{B_0(r_0)}$ such that $\tilde \psi_{r_0} = \upsilon \circ \Ad_V$. It can be chosen even provided $\tilde \psi_{r_0},\upsilon$, and therefore $\psi,\upsilon$, have the same parity. It follows that the claim of the proposition holds true with $U = V U^{(\infty)}$.
\end{proof}

\subsection{A continuity result for parity of states}

\begin{lemma} \label{lem:continuity of parity}
    Let $a<b$ and let $[a,b] \ni \phi \mapsto \omega_{\phi}$ be a weakly-* continuous family of homogeneous pure states that are all unitarily equivalent to a homogeneous pure product state $\omega_0$. Suppose the unitaries $V_\phi$ are homogeneous and $f$-localised near the origin, where $f$ is independent of $\phi$. Then all $\omega_{\phi}$ have the same parity and so do all the $V_{\phi}$.
\end{lemma}

\begin{proof}
    Let $(\Pi, \caH, \Omega)$ be the GNS representation of $\omega_0$. Since $\omega_0$ is homogeneous, there is a unique unitary $\Theta \in \caB(\caH)$ such that
    \begin{equation}
        \Theta \Omega = \Omega, \quad \Pi( \theta(A) ) = \Theta^* \, \Pi(A) \, \Theta, \,\,\, \forall A \in \caA.
    \end{equation}
    Since $\Omega$ is cyclic and $\theta^2 = \id$, the first equality above implies that $\Theta^2 = \I$. Moreover, the purity of $\omega_0$ implies that $\caB(\caH)$ is the weak closure of $\Pi(\caA)$ and so the adjoint action of $\Theta$ extends the action of $\theta$ to all of $\caB(\caH)$.
    
    Let $\theta_L$ be the restriction of $\theta$ to the algebra $\caA_{B(L)}$. Since $\omega_0$ is a product state, it is invariant under $\theta_L$ and so there is a unitary $\Theta_L$ implementing that partial parity such that $\Theta_L \Omega = \Omega$. For any $A,B\in\caA_{\loc}$, there is $L$ large enough such that $\theta_L(A) = \theta(A)$ and so
    \begin{equation}
        \langle \Pi(A)\Omega,  \Theta_L \Pi(B)\Omega\rangle
        = \langle \Theta_L ^*\Pi(A) \Theta_L \Omega,  \Pi(B)\Omega\rangle
        = \langle \Pi(\theta(A)) \Omega, \Pi(B)\Omega\rangle
        =\langle \Pi(A)\Omega,  \Theta \Pi(B)\Omega\rangle.
    \end{equation}
    It now follows by density that $\Theta_L$ converges to $\Theta$ in the weak operator topology. 

    By assumption, the vectors $\Omega_{\phi} = \Pi(V_{\phi}) \Omega$ are unit vector representatives of the states $\omega_{\phi}$ in the GNS space $\caH$. Since all $\omega_{\phi}$ are homogeneous,
    \begin{equation}
        \Theta \Omega_{\phi} = \lambda_{\phi} \Omega_{\phi},\qquad \lambda_{\phi} \in \{+1, -1\}.
    \end{equation}
    We claim that $\phi \mapsto \lambda_{\phi}$ is constant. For any $\phi,\phi'$, 
    \begin{align}
        \lambda_{\phi} - \lambda_{\phi'}
        &= \langle \Omega_{\phi}, \Theta \, \Omega_{\phi} \rangle - \langle \Omega_{\phi'}, \Theta \, \Omega_{\phi'} \rangle
        = \langle \Omega, \Pi(V_{\phi})^* \Theta \Pi(V_{\phi}) \, \Omega \rangle -  \langle \Omega, \Pi(V_{\phi'})^* \Theta \Pi(V_{\phi'}) \, \Omega \rangle\\
        &= \langle \Omega, \Pi(\theta(V_{\phi}))^*  \Pi(V_{\phi})  \Omega \rangle -  \langle \Omega, \Pi(\theta(V_{\phi'}))^*   \Pi(U_{\phi'}) \Omega \rangle
    \end{align}
    since $\Theta \Omega = \Omega$. By the uniform almost locality of the family $V_{\phi}$, we can replace $\theta$ by its restriction $\theta_L$, run the sequence of equalities in reverse order with $\Theta\to\Theta_L$ and conclude that 
    \begin{equation}
        \abs{\lambda_{\phi} - \lambda_{\phi'}}
        \leq \abs{ \omega_{\phi}(\Theta_L) - \omega_{\phi'}(\Theta_L) } + c f(L).
    \end{equation}
    Taking $L$ large enough such that $f(L) < 1$ and choosing then $\phi$ and $\phi'$ close enough such that $\abs{ \omega_{\phi}(\Theta_L) - \omega_{\phi'}(\Theta_L) } < 1$ (which is possible by weak-* continuity) yields $\abs{\lambda_{\phi} - \lambda_{\phi'}} < 1$. Since $\lambda_{\phi}$ and $\lambda_{\phi'}$ take values in $\{+1, -1\}$, it follows that $\lambda_{\phi} = \lambda_{\phi'}$ if $\phi$ and $\phi'$ are close enough, which yields the claim.
\end{proof}

\section{Proof that translation invariant \texorpdfstring{$\AII$}{AII} states are stably SRE} \label{app:free SRE}

For any Hilbert space $V$ we denote by $\Proj_r(V)$ the space of rank $r$ orthogonal projections in $\caB(V)$. Recall that any continuous map $\hat \scrP : \T \rightarrow \Proj_r(V)$ yields a rank $r$ complex vector bundle $E_{\hat \scrP} = \{ (k, v) \in \T \times V \, : \, v \in \Ran \hat \scrP(k)  \}$ over $\T$ which is a subbundle of the trivial bundle $\T \times V$. We define the Chern number $\Ch(\hat \scrP)$ of the map $\hat \scrP$ to be the Chern number of this vector bundle.

\begin{lemma} \label{lem:homotopy of trivial bundle}
    Let $0 \leq r \leq n$ be integers and let $\hat \scrP : \T \rightarrow \Proj_{r}(\C^n)$ be a smooth map. If $\Ch(\hat \scrP) = 0$ and $n \geq 7r$ then there is a smooth homotopy $\hat \scrP : [0, 1] \times \T \rightarrow \Proj_r(\C^n) : (t, k) \mapsto \hat \scrP_t(k)$ such that $\hat  \scrP_0 = \hat \scrP$ and $\hat \scrP_1$ is constant.
\end{lemma}

\begin{proof}
    There is a one-to-one correspondence between elements of $\Proj_r(\C^n)$ and elements of the Grassmannian $\mathrm{Gr}_r(\C^n)$, which consists of the $r$-dimensional subspaces of $\C^n$. So the map $\hat \scrP$ can equivalently be regarded as a smooth map from $\T$ to $\mathrm{Gr}_r(\C^n)$, and we shall use this identification throughout the proof.

    Since isomorphism classes of complex vector bundles over the torus are completely classified by the Chern number, our assumption $\Ch(\hat\scrP) = 0$ implies that the bundle $E_{\hat \scrP}$ is isomorphic to a trivial bundle. It now follows from \cite[Theorem 23.10]{bott1982differential}, and noting that the torus has a good open cover of 7 sets (\cite[Theorem 5.3]{karoubi2017covering}), that there is a homotopy $\hat \scrP : [0, 1] \times \T \rightarrow \mathrm{Gr}_r(\C^n)$ such that $\hat \scrP_0 = \hat P$ and $\hat \scrP_1$ is constant. By~\cite[Theorem~10.22]{LeeSmoothManifolds} we can moreover take this homotopy to be smooth.
\end{proof}

Any smooth matrix-valued map $\hat \scrA: \T \rightarrow \caB(\C^n)$ determines a translation invariant operator $\scrA$ on $l^2(\Z^2 ; \C^n)$ by inverse Fourier transform through
\begin{equation}
    (\scrA f)(x) = \frac{1}{4\pi^2}\int_{\bbT} \,  \dd k \, \mathrm{e}^{\iu k\cdot x}\hat \scrA(k) \hat f(k)
\end{equation}
where $\hat f(k) = \sum_{x\in\bbZ^2} \mathrm{e}^{-\iu k\cdot x}f(x)\in\bbC^n$ for any $f\in l^2(\Z^2 ; \C^n) $. The smoothness of $\hat \scrA$ implies moreover that $\scrA$ has super-polynomial off-diagonal decay in the sense that there is a $g_A \in \caF$ such that
\begin{equation}\label{SuperPolyODD}
    \vert \langle \delta_x \otimes u, \scrA \, \delta_y \otimes v \rangle \vert \leq g_A( \dist(x, y))
\end{equation}
where $\{\delta_x\}_{x \in \Z^2}$ is the position basis of $l^2(\Z^2)$ and $u, v \in \C^n$ are arbitrary. A fortiori, a smooth projection valued map $\hat \scrP : \T \rightarrow \Proj_r(\C^n)$ determines a translation invariant projection $\scrP$ on~$l^2(\Z^2 ; \C^n)$.
\begin{lemma} \label{lem:a smooth homotopy of bundles determines a local Kato generator}
    Let 
    \begin{equation}
        \hat \scrP : [0, 1] \times \T \rightarrow \Proj_{r}(\C^n) : (t, k) \mapsto \hat \scrP_t(k)
    \end{equation}
    be a smooth homotopy and denote by $\scrP_t$ the corresponding projections on $l^2(\Z^2 ; \C^n)$. Then there is a one-parameter norm-continuous family of unitaries $[0, 1] \mapsto \scrV_t$ on $l^2(\Z^2; \C^n)$ and a family of bounded self-adjoint operators $[0, 1] \mapsto \scrG_t$ on $l^2(\Z^2 ; \C^n)$ such that
    \begin{equation}\label{SchrodingerEq}
        \iu \frac{\dd \scrV_t}{\dd t} = \scrG_t \scrV_t, \qquad \scrV_0 = \I, \qquad \text{and} \quad \scrP_t = \scrV_t \scrP_0 \scrV_t^*.
    \end{equation}
    The generators $\scrG_t$ have super-polynomial off-diagonal decay in the sense of~(\ref{SuperPolyODD}), where $g\in\caF$ is independent of $t$.
\end{lemma}

\begin{proof}
    Let
    \begin{equation}
        \hat \scrG_t(k) = \iu \left[\frac{\partial}{\partial t}\hat \scrP_t(k), \hat \scrP_t(k)\right]
    \end{equation}
    be the well-known self-adjoint Kato generator~\cite{KatoAdiabatic} and let $\hat \scrV_t(k)$ be the unique family of unitaries that solve
    \begin{equation}
        \iu\frac{\partial}{\partial t} \hat \scrV_t(k) = \hat \scrG_t(k) \hat \scrV_t(k),\qquad \hat \scrV_0(k) = \I.
    \end{equation}
    Then
    \begin{equation}
        \hat \scrP_t(k) = \hat \scrV_t(k)  \hat \scrP_0(k) \hat \scrV_t(k)^*.
    \end{equation}
    for all $t \in [0, 1]$ and all $k \in \T$. Since the family of projections is smooth, so is the map $(t, k) \mapsto \hat \scrG_t(k)$. All the claims of the lemma now follow by inverse Fourier transform.
    \end{proof}

Finally we show that the data provided by the previous lemma yields an equivalence of symmetric quasi-free states.
\begin{lemma} \label{lem:path of Fermi projections yields path of equivalent states}
Let $[0, 1] \ni t \mapsto \scrG_t$ be a norm-continuous family of self-adjoint operators on $l^2(\Z^2; \C^n)$ such that~(\ref{SuperPolyODD}) holds uniformly in $t$. Let $\scrV_t, \scrP_t$ be the corresponding one-parameter family of unitaries given by~(\ref{SchrodingerEq}). Let $\caA$ be the CAR algebra over $l^2(\Z^2; \C^n)$ and let $\omega_t = \omega_{\scrP_t}$ be the gauge-invariant quasi-free state on $\caA$ corresponding to the projection $\scrP_t$, see~(\ref{QuasiFreeState}). Then there is a $U(1)$ symmetry preserving TDI $F$ such that $\omega_t = \omega_0 \circ \al_t^F$ for all $t \in [0, 1]$.
\end{lemma}

\begin{proof}
    Let $\{ \delta_x \}_{x \in \Z^2}$ be the position basis of $l^2(\Z^2)$ and let $\{e_i \}_{i = 1}^n$ be an orthonormal basis of~$\C^n$. Let us write
    \begin{equation}
        \scrG_t(x, i ;y, j) := \langle \delta_x \otimes e_i, \, \scrG_t \, \delta_y \otimes e_j \rangle
    \end{equation}
    for the matrix elements of $\scrG_t$ with respect to the orthonormal basis $\{ \delta_x \otimes e_i \}$. We take
    \begin{equation}
        F_x(t) = \frac{1}{2}\sum_{y\in\Z^2} \sum_{i, j = 1}^n \, \left( \scrG_t(x, i; y, j) a^*_{x, i} a_{y, j} + \overline{ \scrG_{t}(x, i; y, j) } a^*_{y, j} a_{x, i} \right)
    \end{equation}
    for all $x \in \Z^2$. These are homogeneous self-adjoint elements of $\caA$, and they define a TDI by the super-polynomial decay of the $\scrG_t(x, i ;y, j)$. A short calculation yields
    \begin{equation}
        \sum_{x\in\Z^2} \iu \left[ F_x(t), a(f)\right] = a(\iu \scrG_t f)
    \end{equation}
    which implies that
    \begin{equation}
        \alpha^F_t(a(f)) = a(\scrV_t^*f).
    \end{equation}
    for any $f \in l^2(\Z^2; \C^n)$.
    
    With this, we check that
    \begin{equation}
        (\omega_0 \circ\alpha^F_t) \big(  a^*(f)a(g)  \big) = \omega_0 \big(  a^*(\scrV_t^*f)a(\scrV_t^* g)  \big) =  \left\langle \scrV_t^* g, \scrP_0 \, \scrV_t^* f  \right\rangle = \left\langle g,  \scrP_t \, f\right\rangle = \omega_t \big(  a^*(g) a(f)  \big).
    \end{equation}
    Since the $\omega_t$ are quasi-free we conclude that $\omega_t = \omega_0 \circ \al_t^F$, as required.
\end{proof}

We are now equipped to prove Proposition~\ref{prop:free SRE}. We briefly recall the setting described in Section \ref{sec:non interacting examples}. The single particle Hilbert space is $\caK_m = l^2(\Z^2 ; \C^{2m})$ equipped with a fermionic time-reversal $\scrT$. We consider quasi-free states $\omega_{\scrP}$ corresponding to projections $\scrP$ that are exponentially local \eqref{eq:locality of gapped Fermi projections}, time-reversal invariant ($\scrT \scrP \scrT^* = \scrP$), and translation invariant. It follows that the Fourier transform of $\scrP$ is a smooth projection-valued map $\hat \scrP : \T \rightarrow \Proj_{2r}(\C^{2m})$ with $\Ch(\hat \scrP) = 0$ for some $0 \leq r \leq m$ (recall that the rank of a time-reversal invariant bundle is always even and has vanishing Chern number).

\begin{proof}[Proof of Proposition \ref{prop:free SRE}] \label{proof:free SRE}
    If $2m < 14r$ then in order to apply Lemma~\ref{lem:homotopy of trivial bundle}, we expand the `ambient dimension' from $2m$ to $14r$ by considering the algebra $\caA'$ over $l^2(\Z^2; \C^{2m})\oplus l^2(\Z^2; \C^{14r - 2m})$ and the symmetric state $\omega' = \omega_{\scrP} \gotimes \omega_{\mathrm{empty}}$ corresponding to the projection $\scrQ = \scrP \oplus 0$ and consequently $\hat \scrQ = \hat \scrP \oplus 0$. If $2m > 14r$ we simply take $\scrQ = \scrP$. The projection-valued map $\hat \scrQ = \hat \scrP \oplus 0$ has $\Ch(\hat \scrQ) = 0$ and the ambient dimension is large enough so that Lemma \ref{lem:homotopy of trivial bundle} provides a smooth homotopy $(t, k) \mapsto \hat \scrQ_t(k)$ interpolating between $\hat \scrQ_0 = \hat \scrQ$ and a constant projection $\scrQ_1$.

    Let $\scrQ_t$ be the inverse Fourier transform of $\hat \scrQ_t$. We then apply Lemma \ref{lem:a smooth homotopy of bundles determines a local Kato generator} to obtain unitaries $\scrV_t$ and generators $\scrG_t$ satisfying the assumption of Lemma \ref{lem:path of Fermi projections yields path of equivalent states}. The latter lemma now provides a $U(1)$-invariant TDI $F$ such that $\omega_t' = \omega_0 \circ \al_t^F$ where $\omega_t' = \omega'_{\scrQ_t}$ is the quasi-free state on $\caA'$ determined by the projection $\scrQ_t$. This family of states interpolates between $\omega_{\scrP} \gotimes \omega_{\mathrm{empty}}$ and $\omega_{\scrQ_1}$. Since $\hat \scrQ_1$ is constant, its Fourier transform $\scrQ_1$ is strictly local and $\omega_1 = \omega_{\scrQ_1}$ is a product state. We conclude that $\omega_{\scrP} \gotimes \omega_{\mathrm{empty}}$ is LGA equivalent to a product state. Since $\omega_{\mathrm{empty}}$ is also a product state, we conclude that $\omega_{\scrP}$ is stably SRE, see Section~\ref{sec:states}.
\end{proof}

\section{Flux insertion for free fermions} \label{app:flux insertion for free fermions}

We take up the setting of Section \ref{sec:free fermion examples} and prove the various bounds on free fermion operators on $\scrK_{m}$ used in the arguments presented there. We prove in particular the crucial Proposition~\ref{prop:spectral - qa is trace class} relating the quasi-adiabatic flow and the spectral flow whenever the gap does not close. Recall that we fix an orthonormal basis $\{e_i\}_{i = 1, \cdots, 2m}$ of $\C^{2m}$ so that $\{ \delta_x \otimes e_i\}_{x, i}$ is an orthonormal basis of $\caK_m$. For any operator $\scrA$ on $\caK_m$ we denote by
\begin{equation}
    \scrA(x, i ; y, j) := \langle \delta_x \otimes e_i, \scrA \, \delta_y \otimes e_j \rangle
\end{equation}
its matrix elements with respect to this basis.

Throughout this appendix we will use the following notations; $L = \{  x \in \Z^2 \, : \, x_1 \leq 0 \}$ is the left half plane and $R = \Z^2 \setminus L$ is the right half plane. We let $l_L = (0, -(1, 0))$ be the left horizontal axis and $l_R = (0, (1, 0))$ the right horizontal axis, regarded as a subsets of $\R^2$, and $l = l_L \cup l_R$ for the horizontal axis. We denote by $h = h_{l_L}$ the upper half plane. For any subset $S \subset \Z^2$ we write $\Pi_S$ for the projection on S, \ie $\Pi_S | \delta_x \otimes i \rangle = \delta_{x \in S} | \delta_x \otimes i \rangle$. For any $S \subset \R^2$ and any $x \in \Z^2$ we write $\dist(x, S) = \inf_{y \in S} \norm{x - y}_1$ for the distance from $x$ to the set $S$.

We consider a projector $\scrP$ that is exponentially local, \ie for all $x, y \in \Z^2$ and all $i, j$ we have
\begin{equation} \label{eq:locality of Fermi projection}
    \abs{\scrP(x, i; y, j)} \leq C \ed^{-\eta \norm{x - y}_1}
\end{equation}
for some constants $C, \eta > 0$, and such that $\scrT \scrP \scrT^* = \scrP$. We let $\scrH = \I - \scrP$ be the Hamiltonian with gap $\Delta = (0, 1)$ and Fermi projector $\scrP = \chi_{(-\infty, \mu]}(\scrH)$ for any $\mu \in \Delta$. We denote by $\scrH_{\phi}$ the flux Hamiltonians defined in \eqref{eq:free fermion flux Hamiltonians}. We have $\abs{\scrH_{\phi}(x, i; y, j)} \leq C \ed^{-\eta \norm{x - y}_1}$ for all $\phi$ and all $x, y, i$ and $j$. Recall also that for any $\mu \in \Delta$ we denote by $\scrP^{(\mu)}_{\phi} = \chi_{(-\infty, \mu]}(\scrH_{\phi})$ the Fermi projections of $\scrH_{\phi}$ for Fermi energy $\mu$.

Corresponding to the projectors $\scrP^h_{\phi} = \ed^{\iu \phi \Pi_h} \, \scrP \, \ed^{-\iu \phi \Pi_h}$, where $h$ is the upper half-plane, we have Hamiltonians $\scrH_{\phi}^h = \I - \scrP^h_{\phi} = \ed^{\iu \phi \Pi_h} \,  \scrH_{\phi}^h \, \ed^{-\iu \phi \Pi_h}$ all with spectral gap $\Delta$.

\subsection{Basic locality estimates}

\begin{lemma} \label{lem:basic locality estimates}
   We have
    \begin{align}
        \abs{\scrP_{\phi}^h(x, i ; y, j)} &\leq C \ed^{-\eta \norm{x-y}_1} \label{eq:basic bound 1} \\
        \abs{\partial_{\phi} \scrP_{\phi}^h(x, i; y, j)  } &\leq c \ed^{-\xi ( \abs{x_1 - y_1}  + \abs{x_2} + \abs{y_2} )}  \label{eq:basic bound 2} \\
        \abs{\scrK_{\phi}^h(x, i ; y, j)} &\leq c \ed^{-\xi ( \abs{x_1 - y_1}  + \abs{x_2} + \abs{y_2} )} \label{eq:basic bound 3}
    \end{align}
    for some constants $c, \xi > 0$ and for all $\phi \in \R$, all $x = (x_1, x_2), y = (y_1, y_2) \in \Z^2$ and all $i, j \in \{1,\ldots 2m\}$. 
\end{lemma}

\begin{proof}
    The first claim follows immediately from Eq \eqref{eq:locality of Fermi projection} and the definition $\scrP_{\phi}^h = \ed^{\iu \phi \Pi_h} \scrP \ed^{-\iu \phi \Pi_h}$.

    Let us now bound the matrix elements of $\partial_{\phi} \scrP_{\phi}^h = \iu [ \Pi_h, \scrP_{\phi}^h ]$. If $x, y \in h$ or $x, y \notin h$ then $(\partial_{\phi} \scrP_{\phi}^h)(x, i; y, j) = 0$. If $x \in h$ and $y \not\in h$ then using the first claim of this lemma we obtain
    \begin{equation}
        \abs{ (\partial_{\phi} \scrP_{\phi}^h)(x, i; y, j) } = \abs{ \scrP_{\phi}(x, i ; y , j) } \leq C \ed^{- \eta \norm{x - y}_1} \leq C \ed^{- \eta ( \abs{x_1 - y_1} + \abs{x_2} + \abs{y_2}  )}.
    \end{equation}
    where for the last inequality we used $x \in h$ and $y \not\in h$ so $\norm{x - y}_1 = \abs{ x_1 - y_1} + \abs{x_2 - y_2}  = \abs{x_1 - y_1} + \abs{x_2} + \abs{y_2}$. The same result is obtained if $x \not\in h$ and $y \in h$, thus establishing the second bound. The last bound of the lemma now follows straightforwardly from the definition $\scrK_{\phi}^h = \iu [ \partial_{\phi} \scrP_{\phi}^h, \scrP_{\phi}^h ]$ and the bounds on the matrix elements of $\partial_{\phi} \scrP_{\phi}^h$ and $\scrP_{\phi}^h$.
\end{proof}

\subsection{Comparison of quasi-adiabatic and spectral generators} \label{subsec:comparison of quasi-adiabatic and spectral generators}

Fix a Fermi energy $\mu \in \Delta$ and an interval $[\phi^{(i)}, \phi^{(f)}]$ such that $\mu$ is in a spectral gap $\Delta' = [\mu - \delta, \mu + \delta]$ of $\scrH_{\phi}$ for all $\phi \in [\phi^{(i)}, \phi^{(f)}]$. Then the Kato generators
\begin{equation}
    \scrG^{(\mu)}_{\phi} := -\iu [ \partial_{\phi} \scrP^{(\mu)}_{\phi}, \scrP_{\phi}^{(\mu)} ]
\end{equation}
are well-defined for all $\phi \in [\phi^{(i)}, \phi^{(f)}]$. We want to show that these generators are close to the generators $\scrK_{\phi} = \Pi_L \scrK_{\phi}^h \Pi_L$. First of all, we show some locality preperties. 

\begin{lemma} \label{lem:spectral locality estimates}
    There are constants $c, \eta > 0$ which depend only on the interval $[\phi^{(i)}, \phi^{(f)}]$ such that
    \begin{align}
        \abs{ \scrP_{\phi}^{(\mu)}(x, i; y, j) } &\leq \, c \ed^{-\xi ( \norm{x - y}_1 )} \label{eq:spectral locality estimates bound 1} \\
        \abs{ \partial_{\phi} \scrP_{\phi}^{(\mu)}(x, i; y, j)  } & \leq \, c \ed^{-\xi ( \norm{x - y}_1 + \dist(x, l_L) )} \label{eq:spectral locality estimates bound 2} \\
        \abs{ (\scrP^{(\mu)}_{\phi} - \scrP^h_{\phi})(x, i ; y, j)  }, \abs{ \partial_{\phi} ( \scrP_{\phi}^{(\mu)} - \scrP^h_{\phi})(x, i; y, j)  } &\leq \, c \ed^{- \xi (\norm{ x - y }_1 + \dist(x, l_R) ) } \label{eq:spectral locality estimates bound 3} \\
        \abs{  \scrG_{\phi}^{(\mu)}(x, i; y, j)  } & \leq \, c \ed^{-\xi ( \norm{x - y}_1 + \dist(x, l_L) )} \label{eq:spectral locality estimates bound 4}
    \end{align}
    for all $x, i, y, j$ and $\phi \in [\phi_1, \phi_2]$.
\end{lemma}

\begin{proof}
    By the gap assumption for $\scrH_{\phi}$ we can find a contour $\Gamma$ in the complex plane the circles the spectra of $\scrH_{\phi}$ and $\scrH_{\phi}^h$ below $\mu$ counterclockwise, and such that $\dist(\Gamma, \sigma(\scrH_{\phi})), \dist(\Gamma, \sigma(\scrH^h_{\phi}))$ are bounded below uniformly in $\phi \in [\phi^{(i)}, \phi^{(f)}]$. Then we have
    \begin{equation}
        \scrP^{(\mu)}_{\phi} = - \frac{1}{2 \pi \iu} \oint_{\Gamma} \dd z \, (\scrH_{\phi} - z)^{-1},
    \end{equation}
    from which we obtain
    \begin{equation}
        \partial_{\phi} \scrP_{\phi}^{(\mu)} = \frac{1}{2 \pi \iu} \oint_{\Gamma} \dd z (\scrH_{\phi} - z)^{-1} (\partial_{\phi} \scrH_{\phi}) (\scrH_{\phi} - z)^{-1}.
    \end{equation}
    The bounds \eqref{eq:spectral locality estimates bound 1} and \eqref{eq:spectral locality estimates bound 2} now follow from the exponential locality of the Hamiltonians $\scrH_{\phi}, \scrH^h_{\phi}$, the Combes-Thomas estimates for the resolvents (see for example Theorem 10.5 of \cite{aizenman2015random}) and the fact that $\partial_{\phi} \scrH_{\phi}$ is exponentially localized near the left half-line $l_L$ by construction. The bounds \eqref{eq:spectral locality estimates bound 3} follow in the same way using a Riesz projector formula for $\scrP_{\phi}^h$ and the fact that $\scrH^{h}_{\phi} - \scrH_{\phi}$ is exponentially localized near the right half-line $l_R$. 
    
    Finally, \eqref{eq:spectral locality estimates bound 4} follows from \eqref{eq:spectral locality estimates bound 1} and \eqref{eq:spectral locality estimates bound 2}, and the definitions.
\end{proof}

\begin{lemma} \label{lem:G and K agree at infinity}
    There are constants $c, \xi > 0$ such that
    \begin{equation} \label{eq:G and K^h}
        \abs{ (\scrG^{(\mu)}_{\phi} - \scrK^h_{\phi})(x, i ; y, j)  } \leq c \, \ed^{- \xi (\norm{ x - y }_1 + \dist(x, l_R) ) }
    \end{equation}
    and
    \begin{equation}
        \abs{ (\scrG_{\phi}^{(\mu)} - \scrK_{\phi})(x, i; y, j) } \leq c  \, \ed^{-\xi (\norm{x - y}_1 + \norm{x}_1)}
    \end{equation}
    for all $x, i, y, j$ and $\phi \in [\phi_1, \phi_2]$.
\end{lemma}

\begin{proof}
    The first bound is an immediate consequence of the definitions $\scrG^{(\mu)}_{\phi} = -\iu [ \partial_{\phi} \scrP^{(\mu)}_{\phi}, \scrP^{(\mu)}_{\phi} ]$ and $\scrK^{h}_{\phi} = -\iu [ \partial_{\phi} \scrP^h_{\phi}, \scrP^h_{\phi} ]$ and the bounds provided by Lemma \ref{lem:spectral locality estimates}.

    Now recall that $\scrK_{\phi} = \Pi_L \scrK^h_{\phi} \Pi_L$ and combine this with Lemma \ref{lem:basic locality estimates} to obtain $\abs{\scrK_{\phi}(x, i; y, j)} \leq c' \ed^{-\xi'(\norm{x - y}_1 + \dist(x, l_L))}$ for some $c', \xi' > 0$. Together with the last bound of Lemma \ref{lem:spectral locality estimates} we obtain
    \begin{equation} \label{eq:lemma G and K first bound}
        \abs{ (\scrG_{\phi}^{(\mu)} - \scrK_{\phi})(x, i; y, j) } \leq c'' \ed^{-\xi'' (\norm{x - y}_1 + \dist(x, l_L))}
    \end{equation}
    for some $c'', \xi'' > 0$.

    Similarly, using Lemma \ref{lem:basic locality estimates} we get $\abs{ (\scrK^h_{\phi} - \scrK_{\phi})(x, i; y, j) } \leq c' \ed^{-\xi'( \norm{x - y}_1 + \dist{x, l_R} )}$. Together with the bound \eqref{eq:G and K^h} we obtain
    \begin{equation} \label{eq:lemma G and K second bound}
        \abs{ (\scrG_{\phi}^{(\mu)} - \scrK_{\phi})(x, i; y, j) } \leq c'' \ed^{-\xi'' (\norm{x - y}_1 + \dist(x, l_R))}.
    \end{equation}
    Taking the minimum of \eqref{eq:lemma G and K first bound} and \eqref{eq:lemma G and K second bound} yields the second bound of the Lemma.
\end{proof}

\subsection{Comparison of quasi-adiabatic and spectral evolutions} \label{subsec:comparison of quasi-adiabatic and spectral evolutions}

Remaining in the setting of the previous subsection, we consider a closed interval $[\phi^{(i)}, \phi^{(f)}] \subset \R$ of fluxes where $\mu$ is in a spectral gap $\Delta'$ of $\scrH_{\phi}$. For any $\phi \in [\phi^{(i)}, \phi^{(f)}]$ we let $\scrU^{\qa}_{(\phi, \phi^{(i)})}$ be the unique solution to
\begin{equation}
    \scrU^{\qa}_{(\phi^{(i)}, \phi^{(i)})} = \I, \quad \iu \partial_{\phi} \scrU^{\qa}_{(\phi, \phi^{(i)})} = \scrK_{\phi} \scrU^{\qa}_{(\phi, \phi^{(i)})}.
\end{equation}
Similarly, let $\scrU^{(\mu)}_{(\phi, \phi^{(i)})}$ be the unique solution to
\begin{equation}
    \scrU^{(\mu)}_{(\phi^{(i)}, \phi^{(i)})} = \I, \quad \iu \partial_{\phi} \scrU^{(\mu)}_{(\phi, \phi^{(i)})} = \scrG^{(\mu)}_{\phi} \scrU^{(\mu)}_{(\phi, \phi^{(i)})}.
\end{equation}

Both unitary evolutions are generated by exponentially local generators, so that the following lemma applies to both:
\begin{lemma} \label{lem:locally generated unitaries are local}
    Let $[\phi^{(i)}, \phi^{(f)}] \ni \phi \mapsto \scrA_{\phi}$ be a family of bounded self adjoint operators and suppose there are constants $c, \xi > 0$ such that
    \begin{equation} \label{eq:decay of scrA}
        \abs{ \scrA_{\phi}(x, i; y, j) } \leq c \ed^{- 2\xi \norm{x - y}_1}
    \end{equation}
    for all $\phi \in [\phi^{(i)}, \phi^{(f)}]$. Let $\scrV_{\phi}$ be the unique solution to $\scrV_{\phi^{(i)}} = \I$ and $\iu \partial_{\phi} \scrV_{\phi} = \scrA_{\phi} \scrV_{\phi}$. Then there is a constant $C > 0$ such that
    \begin{equation}
        \abs{ \scrV_{\phi}(x, i; y, j) } \leq C \ed^{- \xi \norm{x - y}_1}.
    \end{equation}
\end{lemma}

\begin{proof}
    Define
    \begin{equation}
        a_n( x, y ) := \sup_{\substack{\vec \phi \in [\phi^{(i)}, \phi^{(f)}]^n \\ i, j \in \{1, \cdots, 2m \}}} \, \abs{ ( \scrA_{\phi_1} \cdots \scrA_{\phi_n} )( x, i ; y, j )   }. \label{eq:an defined}
    \end{equation}
    for all $n \geq 1$. We show by induction that
    \begin{equation}
        a_n(x, y) \leq c^n \kappa^{n-1} \ed^{- \xi \norm{x - y}_1}
    \end{equation}
    where
    \begin{equation}
        \kappa = 2m \, \sup_{x \in \Z^2} \, \sum_{z \in \Z^2} \ed^{- \xi \norm{x - z}_1}.
    \end{equation}
    By the assumption \eqref{eq:decay of scrA} the claimed bound holds true for $n = 1$. Proceeding by induction,
    \begin{align}
        a_{n+1}(x, y) &\leq \sup_{\phi \in [\phi_1, \phi_2]} \, \sup_{i \in \{1, \cdots, 2m\}} \, \sum_{z, i'} \, \abs{ \scrA_{\phi}(x, i; z, i') } \times a_n(z, y) \\
                        &\leq c^{n+1} \kappa^{n-1} \sum_{z, i'}  \ed^{ -2\xi \norm{x - z}_1 } \times \ed^{- \xi \norm{z - y}_1} \leq c^{n+1} \kappa^n \ed^{-\xi \norm{x - y}_1}
    \end{align}
    as required.

    The bound on the matrix elements of $\scrV_{\phi}$ now follows from
    \begin{equation}
        \abs{ (\scrV_{\phi} - \I)(x, i; y, j) } \leq \sum_{n = 1}^{\infty} \frac{\abs{\phi}^n}{n!} a_n(x, y).
    \end{equation}
\end{proof}

\begin{lemma} \label{lem:qa and spectral evolutions are close}
    The difference $\scrU^{\qa}_{(\phi, \phi^{(i)})} - \scrU^{(\mu)}_{(\phi, \phi^{(i)})}$ is trace class for all $\phi \in [\phi^{(i)}, \phi^{(f)}]$.
\end{lemma}

\begin{proof}
    It is sufficient to show that $(\scrU^{(\mu)}_{(\phi, \phi^{(i)})})^* \scrU^{\qa}_{(\phi, \phi^{(i)})} - \I$ is trace class. We have
    \begin{equation}
        (\scrU^{(\mu)}_{(\phi, \phi^{(i)})})^* \scrU^{\qa}_{(\phi, \phi^{(i)})} - \I = \iu \int_{\phi^{(i)}}^{\phi} \dd \phi' ( \scrU^{(\mu)}_{(\phi', \phi^{(i)})} )^*\, ( \scrG^{(\mu)}_{\phi'} - \scrK_{\phi'} ) \, \scrU^{\qa}_{(\phi', \phi^{(i)})}.
    \end{equation}
    By Lemma \ref{lem:locally generated unitaries are local} and the exponential locality of $\scrG_{\phi}^{(\mu)}$ and $\scrK_{\phi}$, the unitaries $\scrU^{(\mu)}_{(\phi', \phi^{(i)})}$ and $\scrU^{\qa}_{(\phi', \phi^{(i)})}$ are exponentially local. Combining this with the bound on $\scrG^{(\mu)}_{\phi} - \scrK_{\phi}$ from Lemma \ref{lem:G and K agree at infinity} we find that there are constants $c, \xi > 0$ such that
    \begin{equation}
        \abs{ \big(  (\scrU^{(\mu)}_{(\phi, \phi^{(i)})})^* \scrU^{\qa}_{(\phi, \phi^{(i)})} - \I \big)(x, i; y, j)} \leq c \ed^{-\xi ( \norm{x - y}_1 + \norm{x}_1 )}.
    \end{equation}
    The trace class claim now follows from Lemma~1 of~\cite{aizenman1998localization}.
\end{proof}

\begin{proposition} \label{prop:spectral - qa is trace class}
    For all $\phi \in [\phi^{(i)}, \phi^{(f)}]$, the difference of projections $$\scrU^{\qa}_{(\phi, \phi^{(i)})} \scrP_{\phi_1}^{(\mu)} \scrU^{\qa}_{(\phi, \phi^{(i)})} - \scrP^{(\mu)}_{\phi}$$ is trace class, and the trace vanishes.
\end{proposition}

\begin{proof}
    The fact that this difference is trace class follows immediately from $$\scrP_{\phi}^{(\mu)} = \scrU^{(\mu)}_{(\phi, \phi^{(i)})} \scrP_{\phi^{(i)}}^{(\mu)} (\scrU^{(\mu)}_{(\phi, \phi^{(i)})})^*$$ and Lemma \ref{lem:qa and spectral evolutions are close}. The trace of the difference is the index of the pair of projections, which is also the Fredholm index of the Fredholm operator $\scrP_{\phi^{(i)}}^{(\mu)}  (\scrU^{(\mu)}_{(\phi, \phi^{(i)})})^* \scrU^{\qa}_{(\phi, \phi^{(i)})}  \scrP_{\phi^{(i)}}^{(\mu)}$ seen as an operator on $\Ran \scrP_{\phi^{(i)}}^{(\mu)}$ (Theorem 5.2 of \cite{AvronSeilerSimon}). Since the index of a Fredholm operator which is a compact perturbation of the identity vanishes, the vanishing of the trace follows again from Lemma \ref{lem:qa and spectral evolutions are close}.
\end{proof}

\bibliographystyle{abbrvArXiv}
\bibliography{bib}


\end{document}